\documentclass[11pt]{article}
\usepackage{graphicx} % Required for inserting images
\usepackage{amsmath}
\usepackage{amsthm}
\usepackage{amssymb}
\usepackage{thmtools}
\usepackage{xcolor}
\usepackage{subcaption}
\usepackage{authblk}

\usepackage{tikz}
\usepackage{bbm}

\usepackage{hyperref}
\usepackage{cleveref}

\usepackage[margin=1in]{geometry}

\usepackage{natbib}
\usepackage{enumitem}
\newcommand{\Caugt}{\mathcal{C}^{\text{aug}}_t}
\newcommand{\Ct}{\mathcal{C}_t}

\newtheorem{theorem}{Theorem}

\newtheorem{proposition}[theorem]{Proposition}
\newtheorem{lemma}[theorem]{Lemma}
\newtheorem{definition}{Definition}

\newtheorem{example}{Example}
\newtheorem{assumption}{Assumption}
\renewcommand\thmcontinues[1]{Continued}

\DeclareMathOperator*{\argmax}{argmax}
\DeclareMathOperator*{\argmin}{argmin}

\newcommand{\Wcostly}{W_{\text{costly}}}
\newcommand{\Wcheap}{W_{\text{cheap}}}
\newcommand{\wcostly}{w_{\text{costly}}}

\newcommand{\wcheap}{w_{\text{cheap}}}

\newcommand{\MPlatform}{M^{\text{E}}}
\newcommand{\MTrivial}{M^{\text{R}}}
\newcommand{\MIdeal}{M^{\text{I}}}

\newcommand{\CCheap}{C}
\newcommand{\UtilityCostly}{U^I_b}
\newcommand{\CCostlyBaseline}{C^I_b}

\newcommand{\CEngagement}{C^E}

\newcommand{\V}{V}

\newcommand{\muengagement}{\mu^{\text{e}}}

\newcommand{\muideal}{\mu^{\text{i}}}

\newcommand{\murandom}{\mu^{\text{r}}}

\newcommand{\UserConsumption}{\textrm{UCQ}}

\newcommand{\UserUtility}{\textrm{UW}}

\newcommand{\RealizedEngagement}{\textrm{RE}}

\date{}

\begin{document}

\author[1]{Nicole Immorlica}
\author[2]{Meena Jagadeesan}
\author[1]{Brendan Lucier}
\affil[1]{\textit{Microsoft Research}}
\affil[2]{\textit{University of California, Berkeley}}

\title{Clickbait vs. Quality: How Engagement-Based Optimization Shapes the Content Landscape in Online Platforms\footnote{Authors in alphabetical order. This research was conducted in part while MJ was at Microsoft Research.}}

\maketitle

\begin{abstract}

Online content platforms commonly use engagement-based optimization when making recommendations. This encourages content creators to invest in quality, but also rewards gaming tricks such as clickbait. To understand the total impact on the content landscape, we study a game between content creators competing on the basis of engagement metrics and analyze the equilibrium decisions about investment in quality and gaming. First, we show the content created at equilibrium exhibits a \textit{positive correlation between quality and gaming}, and we empirically validate this finding on a Twitter dataset. Using the equilibrium structure of the content landscape, we then examine the downstream performance of engagement-based optimization along several axes.
Perhaps counterintuitively, the average quality of content consumed by users can decrease at equilibrium as gaming tricks become more costly for content creators to employ. Moreover, engagement-based optimization can perform worse in terms of user utility than a baseline with random recommendations, and engagement-based optimization is also suboptimal in terms of realized engagement relative to quality-based optimization. Altogether, our results highlight the need to consider content creator incentives when evaluating a platform's choice of optimization metric.
\end{abstract}

\section{Introduction}\label{sec:introduction}

Content recommendation platforms typically optimize \textit{engagement metrics} such as watch time, clicks, retweets, and comments (e.g., \cite{tiktok21, twitter2023algorithm}). Since engagement metrics increase with content quality, one might hope that engagement-based optimization would lead to desirable recommendations. However, engagement-based optimization has led to a proliferation of clickbait \citep{youtube2019}, incendiary content \citep{M20}, divisive content \citep{RBL21} and addictive content \citep{BST22}. A driver of these negative outcomes is that engagement metrics not only reward \textit{quality}, but also reward \textit{gaming tricks} such as clickbait that worsen the user experience. 

In this work, we examine how engagement-based optimization shapes the landscape of content available on the platform. We focus on the role of strategic behavior by content creators: competition to appear in a platform's recommendations influences what content they are incentivized to create \citep{BT18, JGS22, HKJKD22}. In the case of engagement-based optimization, we expect that creators strategically decide how much effort to invest in quality versus how much effort to spend on gaming tricks, both of which increase engagement. For example, since the engagement metric for Twitter includes the number of retweets \citep{twitter2023algorithm}---which includes both quote retweets (where the retweeter adds a comment) and non-quote retweets (without any comment)---creators can either increase quote retweets by using offensive or sensationalized language \citep{MPG23} or increase non-quote retweets by putting more effort into the quality of their content (Example \ref{example:linear}). When the engagement metric for video content includes total watch time \citep{tiktok21}, creators may either increase the ``span'' of their videos---by investing in quality---or instead increase the ``moreishness'' by leveraging behavioral weaknesses of users such as temptation \citep{KMR22} (Example \ref{example:KMR}). When the engagement metric includes clicks, creators can rely on clickbait headlines \citep{youtube2019} or actually improve content quality (Example \ref{example:clicks}). 

Intuitively, creators must balance two opposing forces when incorporating quality and gaming tricks in the content that they create. On one hand, it is expensive for creators to invest in quality, but it may be much cheaper to utilize gaming tricks that also increase engagement. On the other hand, gaming tricks generate disutility for users, which might discourage them from engaging with the content even if it is recommended by the platform. This raises the questions:
\begin{quote}
 \textit{Under engagement-based optimization, how do creators balance between quality and gaming tricks at equilibrium? What is the resulting impact on the content landscape and on the downstream performance of engagement-based optimization?}  
\end{quote}

To investigate these questions, we propose and analyze a game between content creators competing for user consumption through a platform that performs engagement-based optimization. We model the content creator as jointly choosing investment in quality and utilization of gaming tricks. Both quality and gaming tricks increase engagement from consumption, and utilizing gaming tricks is relatively cheaper for the creators than investing in quality. However, gaming decreases user utility, while quality increases user utility, and a user will not consume the content if their utility from consumption is negative. We study the Nash equilibrium in the game between the content creators. 

We first examine the balance between gaming tricks and quality amongst content created at equilibrium (Section \ref{sec:positivecorrelation}). Interestingly, we find that there is a \textit{positive correlation} between gaming and investment at equilibrium: higher-quality content typically exhibits \textit{higher} levels of gaming tricks. We prove that equilibria exhibit this positive correlation (Figure \ref{fig:equilibrium}; Theorem \ref{thm:positivecorrelationhomogeneous}), and we also empirically validate this finding on a Twitter dataset \citep{MCPWD23} (Figure \ref{fig:positivecorrelation} and Table \ref{tab:correlation_coefficients}). These results suggest that gaming tricks and quality should be viewed as \textit{complements}, rather than substitutes. 

Accounting for how the platform's metric shapes the content landscape at equilibrium, we then analyze the downstream performance of engagement-based optimization (Section \ref{sec:performance}). We uncover striking properties of engagement-based optimization along several performances axes and discuss implications for platform design (Figure \ref{fig:downstreamimpacts}). 
\begin{itemize}[leftmargin=*]
    \item \textbf{Content Quality.} First, we examine the average quality of content consumed by users and show that it can \emph{decrease} as gaming tricks become more costly for creators (Figure \ref{fig:qualityconsumption}; Theorem \ref{thm:comparisonuserconsumptiongaming}). In other words, as it becomes more difficult for content creators to game the engagement metric, the average content quality at equilibrium becomes worse. From a platform design perspective, this suggests that increasing the transparency of the platform's metric (which intuitively reduces gaming costs for creators) may improve the average quality of content consumed by users. 
    \item \textbf{User Engagement.} Next, we examine the realization of user engagement metrics at the equilibrium of content generation and user consumption. Even though engagement-based optimization perfectly optimizes engagement on a fixed content landscape, engagement-based optimization can perform worse than other baselines (e.g., optimizing directly for quality) at equilibrium (Figure \ref{fig:realizedengagement}; Theorem \ref{thm:comparisonengagement}). From a platform design perspective, this suggests that even if the platform's true objective is realized engagement, the platform might still prefer approaches other than engagement-based optimization when accounting for the way content creators will respond. 
    \item \textbf{User Welfare.} Finally, we examine the user welfare at equilibrium. We show that engagement-based optimization can lead to lower user welfare at equilibrium than even the conservative baseline of randomly recommending content (Figure \ref{fig:userwelfare}; Theorem \ref{thm:comparisonuserutility}). From a platform design perspective, this suggests that engagement-based optimization may not retain users in a competitive marketplace in the long-run. 
\end{itemize}
\noindent Altogether, these results illustrate the importance of factoring in the endogeneity of the content landscape when assessing the downstream impacts of engagement-based optimization.

\subsection{Related Work}\label{subsec:relatedwork}

Our work connects to research threads on \textit{content creator competition in recommender systems} and \textit{strategic behavior in machine learning}.

\paragraph{Content-creator competition in recommender systems.} An emerging line of work has proposed game-theoretic models of content creator competition in recommender systems, where content creators strategically choosing what content to create \citep{BTK17, BT18, BRT20} or the quality of their content \citep{DBLP:conf/www/GhoshM11, qian2022digital}. Some models embed content in a continuous, multi-dimensional action space, characterizing when specialization occurs \citep{JGS22} and the impact of noisy recommendations \citep{HKJKD22}. Other models capture that content creators compete for engagement \citep{YLNWX23} and general functions of platform ``scores'' across the content landscape \citep{YLSLZWX23}. These models have also been extended to dynamic settings, including where the platform learns over time \citep{DBLP:conf/innovations/GhoshH13, DBLP:conf/aaai/LiuH18, HJJS23}, where content creators learn over time \citep{BRT20, PMB23}. Notably, \citet{BSDX23} study a dynamic setting where the platform learns over time and content creators strategically choose the probability of feedback (clickrate) of their content. However, while these works all assume that creator utility depends only on winning recommendations (or only on content scores according to the platform metric \citep{YLNWX23, YLSLZWX23}), our model incorporates \textit{misalignment} between the platform's (engagement) metric and user utility.\footnote{A rich line of work (e.g., \citep{EW16, MBH21, KMR22, SVNAM21}) has identified sources of misalignment between engagement metrics and user utility and broader issues with inferring user preferences from observed behaviors; these sources of misalignment motivated us to incorporate gaming tricks which increase engagement but reduce user utility into our model.} In particular, our model and insights rely on the fact that creators only derive utility if their content is recommended \textit{and} the content generates nonnegative user utility.

Several other works study content creator competition under different modelling assumptions: e.g., where content quality is fixed and all creator actions are gaming \citep{MPG23}, where content creators have fixed content but may dynamically leave the platform over time  \citep{MCBSZB20, BT23, HHLOS24}, where the recommendation algorithm biases affect market concentration but content creators have fixed content \citep{calvano2023artificial, castellini2023recommender}, where the platform designs a contract determining payments and recommendations \citep{ZKJJ23}, where the platform creates its own content \citep{aridor2021recommenders}, and where the platform designs badges to incentivize user-generated content \citep{ISS15}. This line of work also builds on Hotelling models of product selection from economics (e.g. \citep{H29, S79}, see \citet{ProductDifferentiation} for a textbook treatment).

\paragraph{Strategic behavior in machine learning.} A rich line of work on \textit{strategic classification} (e.g. \citep{bruckner12pred, hardt16strat}) focuses primarily on agents strategically adapting their features in classification problems, whereas our work focuses on agents competing to win users in recommender systems. Some works also consider improvement (e.g. \citep{kleinberg19improvement, HILW20, ABBN22}), though also with a focus on classification problems. One exception is \citep{LGB21}, which studies ranking problems; however, the model in \citep{LGB21} considers all effort as improvement, whereas our model distinguishes between clickbait and quality. Other topics studied in this research thread include shifts to the population in response to a machine learning predictor (e.g. \citep{PZMH20}), strategic behavior from users (e.g. \citep{HMP23}), and incentivizing exploration (e.g., \citep{KMP13, FKKK14, SS21}).

\section{Model}\label{sec:model}

We study a stylized model for content recommendation in which an online platform recommends to each user a single piece of digital content within the content landscape available on the platform.\footnote{Our model can also capture a stream of content (e.g., see Example \ref{example:KMR}), even though we abstract away from this by focusing on one recommendation at a time.}  There are $P \ge 2$ content creators who each create a single piece of content and compete to appear in recommendations.  Building on the models of \citet{BT18, JGS22, HKJKD22, YLSLZWX23}, the content landscape is \textit{endogeneously} determined by the multi-dimensional actions of the content creators.

\subsection{Creator Costs, User Utility, and Platform Engagement}\label{subsec:assumptions}

Since our focus is on investment versus gaming, we project pieces of digital content into 2 dimensions $w= [\wcostly, \wcheap] \in \mathbb{R}^2_{\ge 0}$. The more costly dimension $\wcostly$ denotes a measure of the content's \textit{quality}, whereas the cheap dimension $\wcheap$ reflects the extent of \textit{gaming tricks} present in the content. These measures are normalized so that $w = [0,0]$ represents content generated by a creator who exerted no effort on quality or gaming.

We specify below how the costly and cheap dimensions impact \textit{creator costs}, \textit{user utility}, and \textit{platform engagement}. Using these specifications, we then provide additional intuition for the \textit{qualitative interpretation of quality and gaming tricks} in our model. 

\paragraph{Creator Costs.}
Each content creator pays a (one-time) cost of $c(w) \ge 0$ to create content $w \in \mathbb{R}_{\ge 0}^2$. We assume that $c$ is continuously differentiable in $w$ and satisfies the following additional assumptions. First, investing in quality content is costly:  $(\nabla (c(w)))_1 > 0$ for all $w \in \mathbb{R}_{\ge 0}^2$. Moreover, engaging in gaming tricks is either always free or always incurs a cost: either $(\nabla (c(w)))_2 > 0$ for all $w \in \mathbb{R}_{\ge 0}^2$ or $(\nabla (c(w)))_2 = 0$ for all $w \in \mathbb{R}_{\ge 0}^2$. Furthermore, creators have the option to opt out by not investing costly effort in either gaming tricks or quality: $c([0,0]) = 0$. Finally, costs go to $\infty$ in the limit: $\sup_{\wcostly} c([\wcostly, 0]) = \infty$.

\paragraph{User Utility.}
Each user has a type $t \in \mathcal{T} \subseteq \mathbb{R}_{\ge 0}$ that reflects the user's relative tolerance for gaming tricks. We assume that the type space $\mathcal{T}$ is finite.
A user with type $t \in \mathcal{T}$ receives utility $u(w,t) \in \mathbb{R}$ from consuming content $w \in \mathbb{R}_{\ge 0}^2$, where the utility function is normalized so that the user's outside option offers $0$ utility.  We assume that $u$ is continuously differentiable in $w$ for each $t \in \mathcal{T}$ and satisfies the following additional assumptions. Users derive positive utility from $\wcostly$ and negative utility from $\wcheap$: 
\begin{itemize}[leftmargin=*]
%[label=(B\arabic*)]
    \item For each $t \in \mathcal{T}$ and $\wcostly \in \mathbb{R}_{\ge 0}$: the utility $u([\wcostly, \wcheap], t)$ is strictly decreasing in $\wcheap$ and  approaches $-\infty$ as $\wcheap \rightarrow \infty$. 
    \item For each $t \in \mathcal{T}$ and $\wcheap \in \mathbb{R}_{\ge 0}$: the utility $u([\wcostly, \wcheap], t)$ is strictly increasing in $\wcostly$ and approaches $\infty$ as $\wcostly \rightarrow \infty$. 
\end{itemize}
Furthermore, higher types are more likely to have a nonnegative user utility than lower types, which captures that higher types are less sensitive to gaming tricks than lower types:  
\begin{itemize}[leftmargin=*]
%[label=(B\arabic*),start=3]
\item For any $w \in \mathbb{R}_{\ge 0}^2$ and $t, t' \in \mathcal{T}$ such that $t' > t$: if $u(w,t) \ge 0$, then it holds that $u(w, t') \ge 0$. 
\end{itemize}

\paragraph{Engagement.}
If a user chooses to consume content $w$, this interaction generates platform engagement $\MPlatform(w) \in \mathbb{R}$. The engagement metric $\MPlatform(w)$ depends on the content $w$ but is independent of the user's type $t$ (conditional on the user choosing to consume the content). We assume that $\MPlatform$ is continuously differentiable in $w$ and satisfies the following additional assumptions. First, both cheap gaming tricks and investment in quality increase the engagement metric: $(\nabla \MPlatform(w))_1, (\nabla \MPlatform(w))_2 > 0$ for all $w \in \mathbb{R}_{\ge 0}^2$. Moreover, the engagement metric is nonnegative: $\MPlatform(w) \ge 0$ for all $w \in \mathbb{R}_{\ge 0}^2$. Finally, the relative cost of gaming tricks versus costly investment is less than the relative benefit: 
$\frac{(\nabla c(w))_2}{(\nabla c(w))_1} < \frac{(\nabla \MPlatform(w))_2}{(\nabla \MPlatform(w))_1}$ for all $w \in \mathbb{R}_{\ge 0}^2$.  In other words, it is more cost-effective for a creator to increase the engagement metric via gaming than via quality, for a user who would choose to consume the content either way.

\paragraph{Qualitative Interpretation of Quality and Gaming Tricks.} With this formalization of creator costs, user utility, and platform engagement in place, we turn to the qualitative interpretation of quality as measured by $\wcostly$ and gaming tricks as measured by $\wcheap$. Both quality and gaming tricks reflect effort by creators that increases engagement; however, quality captures effort that is beneficial to users (increases user utility), whereas gaming tricks captures effort that is harmful to users (reduces user utility). Moreover, since a creator can simultaneously invest effort into both quality and gaming tricks, a single piece of digital content can exhibit both gaming tricks and quality at the same time. In fact, high-quality content which also exhibits a sufficient level of gaming tricks can generate arbitrarily low user utility, which illustrates that quality does \textit{not} capture a user's level of appreciation of the content. We defer further discussion of quality and gaming tricks to Section \ref{subsec:examples}, where we instantiate our model within several real-world examples. 

\subsection{Timing and Interaction between the Platform, Users, and Content Creators}

The interaction between the platform, users, and content creators defines a game that proceeds in stages. 
 The timing of the game is as follows: 
\begin{enumerate}[leftmargin=*,label=\textbf{Stage \arabic*:}]
    \item Each content creator $i \in [P]$ simultaneously chooses what content $w_i \in \mathbb{R}_{\ge 0}^2$ to create. %according to a game that we describe below. 
    These choices give rise to a content landscape $\mathbf{w} = (w_1, \ldots, w_P)$.  
    \item A user with type $t \sim \mathcal{T}$ is uniformly drawn and comes to the platform.
    \item  The platform observes the user's type and evaluates content $w$ according to a metric $M : \mathbb{R}_{\geq 0}^2 \to \mathbb{R}$ that maps each piece of content $w_i$ to a score $M(w_i)$. The platform optimizes $M$ over content available in the content landscape that generates nonnegative utility for the user. More formally, the platform selects content creator
    \[i^*(M; \textbf{w}) \in \argmax_{i \in [P]} (M(w_i) \cdot \mathbbm{1}[u(w_i, t) \ge 0]),\] breaking ties uniformly at random, and recommends the content $w_{i^*(M; \textbf{w})}$ to the user.
    \item The user consumes the  the recommended content $w_{i^*(M; \textbf{w})}$ if and only if $u(w_{i^*(M; \textbf{w})}, t) \ge 0$ (i.e., if and only if the content is at least as appealing as their outside option).
\end{enumerate}

We assume that content creators
know the user utility function $u$ and the distribution of $\mathcal{T}$ but do not know the specific realization of $t \sim \mathcal{T}$ in \textbf{Stage 2}. On the other hand, the platform can observe the realization $t \sim \mathcal{T}$. The platform can also observe the full content landscape $\textbf{w}$ and knows the user utility function $u$. This provides the platform with sufficient information to solve the optimization problem $\argmax_{i \in [P]} (M(w_i) \cdot \mathbbm{1}[u(w_i, t) \ge 0])$ in \textbf{Stage 3}.\footnote{The platform may be able to evaluate $\argmax_{i \in [P]} (M(w_i) \cdot \mathbbm{1}[u(w_i, t) \ge 0])$ with less information. For example, if $M = \MPlatform$, then $\MPlatform(w)$ can typically be estimated from observable data such as user behavior patterns without knowledge of $\wcostly$ and $\wcheap$. Moreover, since $\mathbbm{1}[u(w_i, t) \ge 0]$ captures the event that users click on the content $w_i$, if the platform has a predictor for clicks, this would provide them an estimate of $\mathbbm{1}[u(w_i, t) \ge 0]$.} The user knows their own type $t$ and the utility function $u$, and can also observe the content $w$ recommended to them, so they can evaluate whether $u(w_{i^*(M; \textbf{w})}, t) \ge 0$.\footnote{In reality, users may not always be able to perfectly observe $\wcostly$ and $\wcheap$ (or gauge their own utility) without consuming the content. Our model makes the simplifying assumption that user choice is noiseless.}

\paragraph{Equilibrium decisions of content creators.} The recommendation process defines a game played between the content creators, who strategically choose their content $w_i \in \mathbb{R}_{\ge 0}^2$ in \textbf{Stage 1}.  We assume that values are normalized so that a content creator receives a value of $1$ for being shown to a user.  Since the goods are digital, production costs are one-time and incurred regardless of whether the user consumes the content.
Creator $i$'s expected utility is therefore
\begin{equation}
\label{eq:producerutility}
   U_i(w_i; \mathbf{w}_{-i}) := \mathbb{E}[\mathbbm{1}[i^*(M; \textbf{w}) = i]] - c(w),
\end{equation}
where the expectation is over any randomness in user types $\mathcal{T}$. We allow content creators to randomize over their choice of content, and write $\mu_i \in \Delta(\mathbb{R}_{\ge 0}^2)$ for such a mixed strategy.  A (mixed) Nash equilibrium $(\mu_1, \ldots, \mu_P)$, for $\mu_i \in \Delta(\mathbb{R}_{\ge 0}^2)$, is a profile of mixed strategies that are mutual best-responses. 
Since the content creators are symmetric in our model, we will focus primarily on symmetric mixed Nash equilibria in which each creator employs the same mixed strategy, which must exist (see Theorem~\ref{thm:equilibriumexistence} below).
Note that the Nash equilibrium specifies the distribution over content landscapes $\mathbf{w}$. 

\paragraph{The platform's choice of metric $M$ in Stage 3.} 
We primarily focus on \textit{engagement-based optimization} where $M = \MPlatform$, meaning that the platform optimizes for engagement. As a benchmark, we also consider \textit{investment-based optimization} where $M(w) = \MIdeal(w) := \wcostly$ does not reward gaming tricks; however, note that this baseline is idealized, since $\wcostly$ is not always identifiable from observable data in practice. As another baseline, we consider \textit{random recommendations} where $M(w) = \MTrivial(w) := 1$ which captures choosing uniformly at random from all content that generates nonnegative user utility.

\subsection{Running examples}\label{subsec:examples}
We provide instantation of our models that serves as running examples throughout the paper. 
\begin{example}
\label{example:linear}
Consider an online platform such as Twitter which uses retweets as one of the terms its objective \citep{twitter2023algorithm}. However, Twitter does not differentiate between quote retweets (where the retweeter adds a comment) and non-quote retweets (where there is no added comment). Creators can cheaply increase quote retweets by increasing the offensiveness or sensationalism of the content \citep{MPG23}, or increase non-quote retweets by actually improving content quality. As a stylized model for this, let $\wcheap$ be the offensiveness of the content and let $\wcostly$ capture costly investment into content quality. Let the utility function of a user with type $t \in \mathcal{T} \subseteq \mathbb{R}_{>0}$ be the linear function $u(w, t) = \wcostly - (\wcheap / t) + \alpha$, where $\alpha \in \mathbb{R}$ is the baseline utility from no effort and $t$ captures the user's tolerance to offensive content. Let the platform metric $\MPlatform(w) = \wcostly + \wcheap$ and cost function $c(w) = \wcostly + \gamma \cdot \wcheap$ for $\gamma \in [0, 1)$ also be linear functions. The platform metric captures the idea that the platform does not distinguish between different types of retweets; the cost function captures the idea that it is relatively easier for creators to insert sensationalism into tweets, which requires just a few word changes, compared to improving content quality, which might require, for example, time-intensive fact-checking.
\end{example}

\begin{example}
\label{example:KMR}
Consider an online platform such as TikTok \citep{tiktok21} that incorporates watch time into its objective. Creators can increase watch time by: a) creating ``moreish'' content that keeps users watching a video stream even after they are deriving disutility from it, or b) increasing ``span'' by increasing the amount of substantive content, as modelled in \citet{KMR22}. More formally, let $\wcostly := \frac{q}{1-q}$ be a reparameterized version of the span $q\in [0,1]$, let $\wcheap := \frac{p}{1-p}$ be a reparameterized version of the moreishness $p\in [0,1]$.\footnote{In the model in \citet{KMR22}, users have two modes: System 1 (the ``addicted'' mode) and System 2 (the ``rational'' mode). Roughly speaking, the moreishness $p$ is the probability that the user continues to watch the video stream while in System 1, and the span $q$ is the analogous probability for System 2.} For a given user, let $v$ be the value derived from each time step from watching substantive content, let $W$ be the outside option for each time step, and let $t := v/W - 1 > 0$ capture the shifted ratio. In this notation, the engagement metric $\MPlatform$ and user utility $u$ from \citet{KMR22} take the following form: $\MPlatform([\wcheap,\wcostly]) := \wcostly + \wcheap + 1$
and $u(w, t) :=  W \cdot t \cdot \left(\wcostly- \wcheap/t + 1 \right)$
We further specify the cost function based on a linear combination of the expected amount of ``span'' time and the expected amount of ``moreish'' time that the user consumes: $c(w) := \wcostly + \gamma \cdot \wcheap$ 
where $\gamma \in [0,1)$ specifying the  cost of increasing moreishness relative to increasing span.\footnote{While Example \ref{example:linear} and Example \ref{example:KMR} differ in terms of real-world interpretations, the functional forms in the two examples are very similar. In particular, the cost functions $c(w)$ are identical, and the engagement functions $\MPlatform$ are identical up to a scalar shift of $1$. The user utility $u$ in Example \ref{example:KMR} is equal to the user utility $u$ in Example \ref{example:linear} with $\alpha = 1$ and with a multiplicative shift of $W \cdot t$.  } 
\end{example}

\begin{example}
\label{example:clicks}
Consider an online platform such as YouTube that historically used clicks as one of the terms in their objective \citep{youtube2019}. Creators can cheaply increase clicks by leveraging clickbait titles or thumbnails  or by increasing the quality of their content. As a stylized model for this, let $\wcheap$ capture how flashy or sensationalized the title or thumbnail is, and let $\wcostly$ capture the quality of the content. 
The number of clicks $\MPlatform(w)$ increases with both clickbait $\wcheap$ and quality $\wcostly$, and user utility $u(w, t)$ increases with quality $\wcostly$ and decreases with clickbait $\wcheap$. A user quits the platform if their utility falls below zero. (This means the event $\mathbbm{1}[u(w, t) \ge 0]$ captures that the user does not quit the platform, rather than the event that the user clicks the content, for this particular example.) 
\end{example}

\begin{figure}[t]
    \centering
        \begin{subfigure}[b]{0.32\textwidth}
        \centering
        \includegraphics[width=\textwidth]{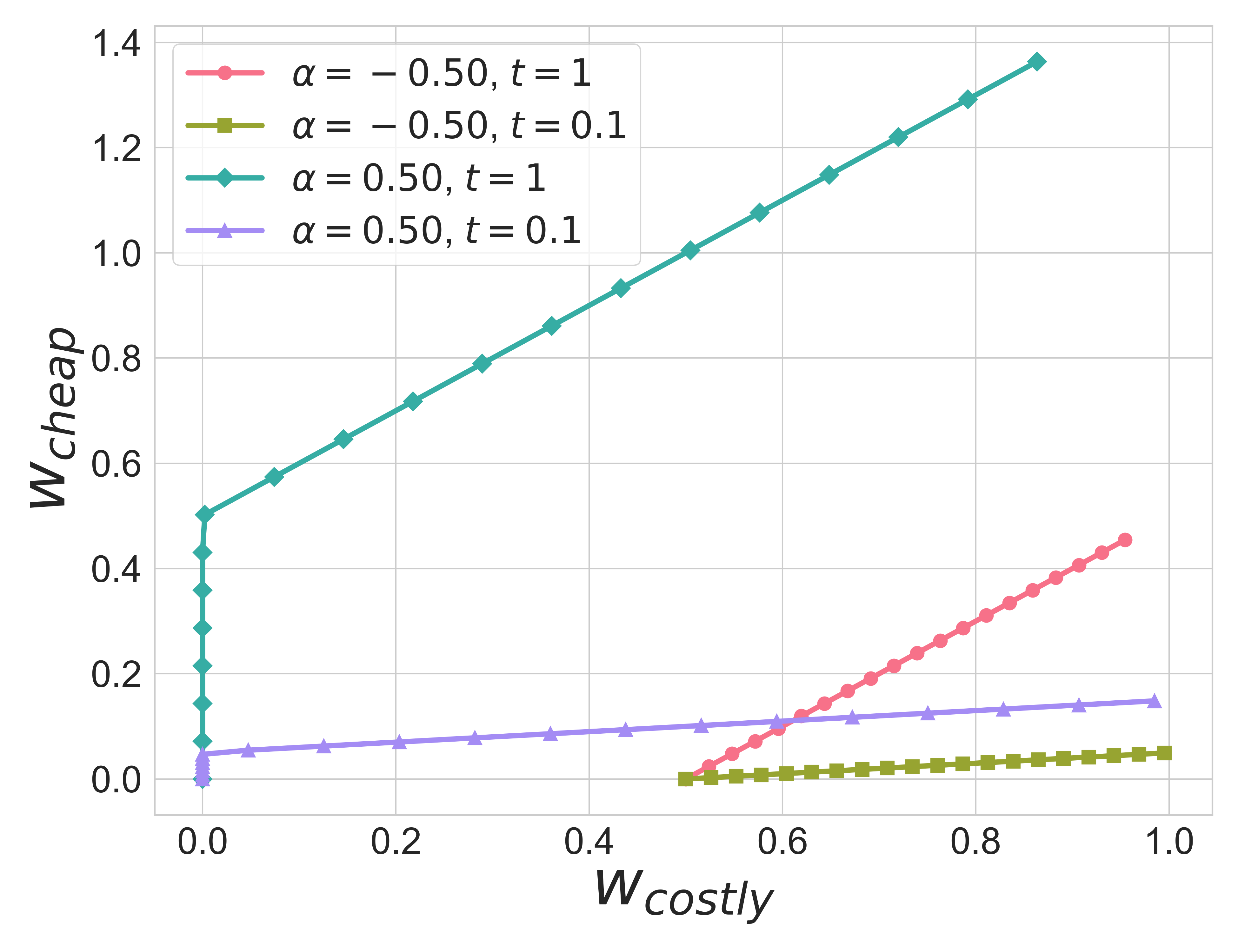}
        \caption{$|\mathcal{T}| = 1$}
        \label{fig:1type}
    \end{subfigure}
    \begin{subfigure}[b]{0.32\textwidth}
        \centering
        \includegraphics[width=\textwidth]{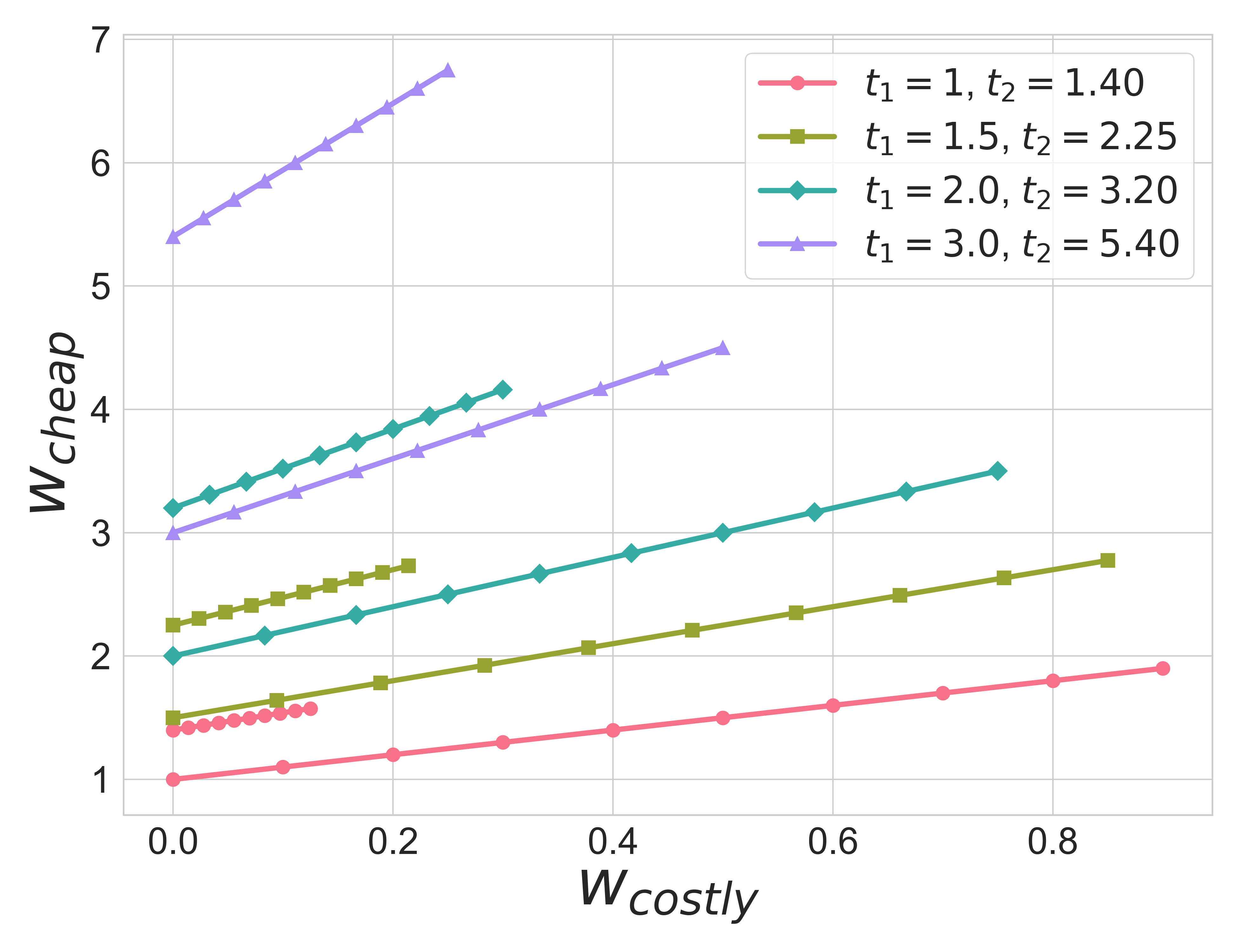}
        \caption{$|\mathcal{T}| = 2$}
        \label{fig:2types}
    \end{subfigure}
    \hfill
    \begin{subfigure}[b]{0.32\textwidth}
        \centering
        \includegraphics[width=\textwidth]{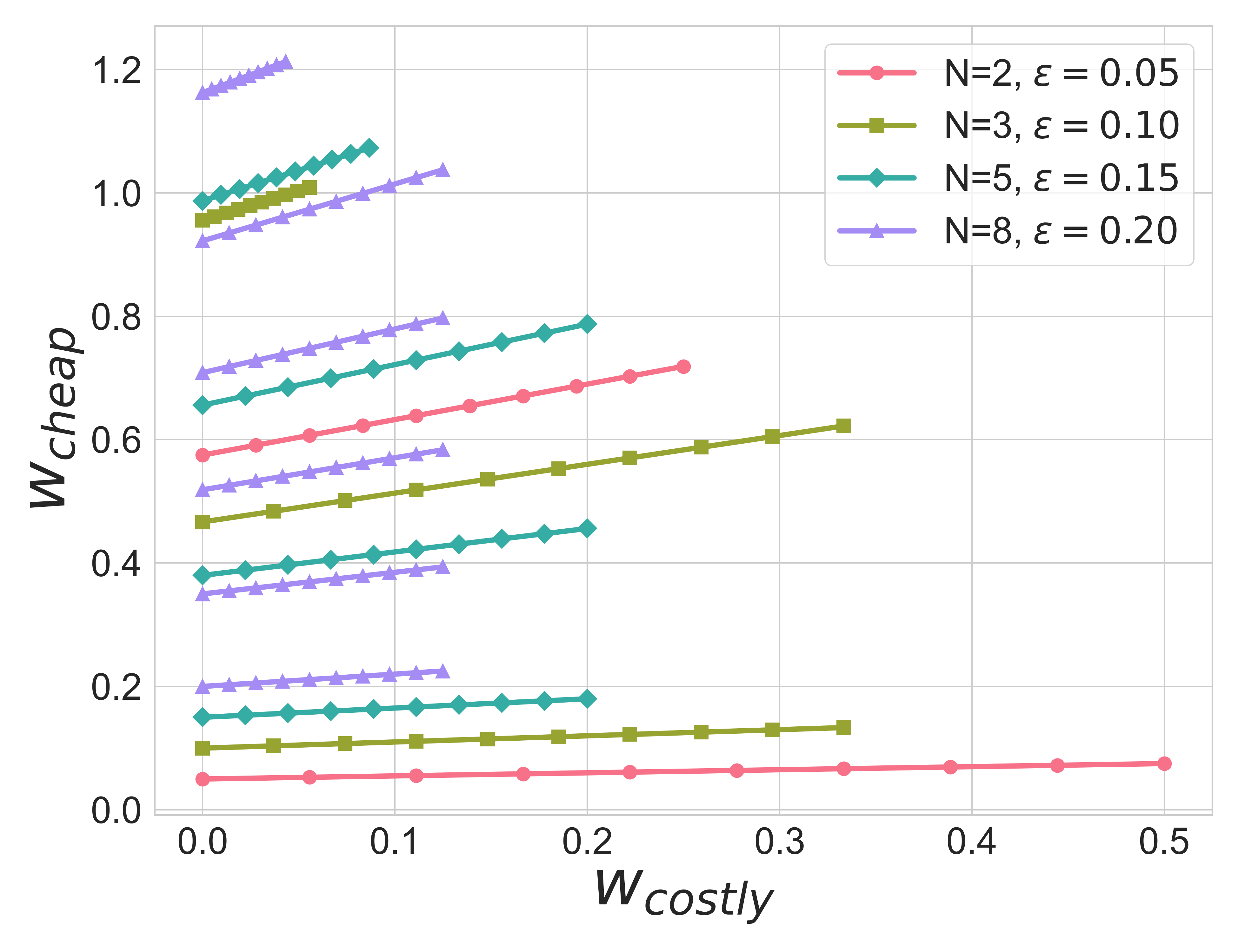}
        \caption{$|\mathcal{T}| = N$}
        \label{fig:Ntypes}
    \end{subfigure}
    \caption{Support of a symmetric mixed equilibrium for engagement-based optimization in \Cref{example:linear}. The parameter settings are $\gamma = 0.1$ (left), $\alpha = 1$, $\gamma = 0$, $\mathcal{T} = \left\{t_1, t_2\right\}$ (middle), and $\alpha = 1$, $\gamma = 0$, $\mathcal{T} = \mathcal{T}_{N, \epsilon}$ (right). The support exhibits positive correlation between gaming tricks $\wcheap$ and investment in quality $\wcostly$ (Proposition \ref{prop:positivecorrelation} and Theorem \ref{thm:positivecorrelationhomogeneous}). For homogeneous users (left), the slope varies with the type $t$ and the intercept varies with the baseline utility $\alpha$ (Theorem \ref{thm:onetype}). For heterogeneous users with $N$ well-separated types (right), the support consists of $N'$ disjoint line segments with varying slopes and intercepts, where $N' < N$ in several cases (Theorem \ref{thm:Ntypes}).} 
    \label{fig:equilibrium}
\end{figure}

\subsection{Equilibrium existence and overview of equilibrium characterization results}\label{subsec:equilibriumcharacterization}

We show that a symmetric mixed equilibrium exists for engagement-based optimization for arbitrary setups.  
\begin{theorem}
\label{thm:equilibriumexistence}
Let $\mathcal{T} \subseteq \mathbb{R}_{\ge 0}$ be any finite type space. Then a symmetric mixed equilibrium exists in the game between content creators with $M = \MPlatform$. 
\end{theorem}
\noindent Since the game has an infinite action space and has discontinuous utility functions, the proof of Theorem \ref{thm:equilibriumexistence} relies on equilibrium existence technology for \textit{discontinuous} games \citep{reny99}. We defer the full proof to \Cref{appendix:proofsmodel}.

Although the symmetric mixed equilibrium does not appear to permit a clean closed-form characterization in general,   we compute closed-form expressions for a symmetric mixed equilibrium $\muengagement(P, c, u, \mathcal{T})$ under further structural assumptions
%for engagement-based optimization 
(Figure \ref{fig:equilibrium}; \Cref{sec:equilibriumcharacterizations}). When users are \textit{homogeneous} (i.e. $\mathcal{T} = \left\{t\right\}$), we compute a symmetric mixed equilibrium for all possible settings of $P$, $c$, and $u$ (Figure \ref{fig:1type}; Theorem \ref{thm:onetype}). We also consider \textit{heterogeneous} users  (i.e. where $|\mathcal{T}| > 1$) under further restrictions: we assume the gaming tricks are costless, place a linearity assumption on the costs $c$ and engagement metric $\MPlatform$ that is satisfied by Examples \ref{example:linear}-\ref{example:KMR}, and focus on the case of $P = 2$ creators. We compute a symmetric mixed equilibrium for arbitrary type spaces $\mathcal{T} = \left\{t_1, t_2\right\}$ with two types (Figure \ref{fig:2types};  Theorem \ref{thm:2types}) and for arbitrarily large type spaces with sufficiently ``well-separated'' types such as $\mathcal{T}_{N, \epsilon} := \left\{(1+ \epsilon) (1 + 1/N)^{i-1} - 1 \mid 1 \le i \le N \right\}$ (Figure \ref{fig:Ntypes}; Theorem \ref{thm:Ntypes}).

We also provide closed-form expressions for a symmetric mixed equilibrium for investment-based optimization  $\muideal(P, c, u, \mathcal{T})$ and random recommendations $\murandom(P, c, u, \mathcal{T})$ under certain structural assumptions (\Cref{sec:equilibriumcharacterizationsbaselines}).

\section{Positive correlation between quality and gaming tricks}\label{sec:positivecorrelation}

When the platform optimizes engagement metrics $\MPlatform$, each content creator \textit{jointly} determines how much to utilize gaming tricks and invest in quality. The creators' equilibrium decisions of how to balance gaming and quality in turn determine the properties of content in the content landscape. In this section, we show that there is a positive correlation between gaming and quality: that is, content that exhibits higher levels of gaming typically exhibits higher investment in quality. We prove that the equilibria satisfy this property (\Cref{subsec:theoreticalbalance}), and we empirically validate this property on a dataset \citep{MCPWD23} of Twitter recommendations (\Cref{subsec:empirical}). 

\subsection{Theoretical analysis of balance between gaming and quality}\label{subsec:theoreticalbalance}

We theoretically analyze the balance of gaming and quality at equilibrium as follows. Since the content landscape $\mathbf{w} = [w_1, \ldots, w_P]$ at equilibrium consists of content $w_i \sim \mu_i$ for $i \in [P]$, the set of content that shows up in the content landscape with nonzero probability is equal to $\cup_{i \in [P]} \text{supp}(\mu_i)$. We examine the relationship between the quality $\wcostly$ and the level of gaming $\wcheap$ for $w \in \cup_{i \in [P]} \text{supp}(\mu_i)$. 

For general type spaces, we show that the set $\cup_{i \in [P]} \text{supp}(\mu_i)$ of content is contained in a union of curves, each exhibiting ``positive correlation'' between $\wcheap$ and $\wcostly$ (Figure \ref{fig:equilibrium}). 
\begin{proposition}
\label{prop:positivecorrelation}
Let $\mathcal{T} \subseteq \mathbb{R}_{\ge 0}$ be any finite type space, and suppose that gaming is not costless (i.e. $(\nabla(c(w)))_2 > 0$ for all $w \in \mathbb{R}_{\ge 0}^2$). There exist weakly increasing functions $f_t: \mathbb{R}_{\ge 0} \rightarrow \mathbb{R}_{\ge 0}$ for each $t \in \mathcal{T}$ such that at any (mixed) Nash equilibrium $(\mu_1, \mu_2, \ldots, \mu_P)$ in the game with $M = \MPlatform$, the set of content $\cup_{i \in [P]} \text{supp}(\mu_i)$ is contained in the following set: 
\[\cup_{i \in [P]} \text{supp}(\mu_i) \subseteq \left(\cup_{t \in \mathcal{T}} \underbrace{\left\{(f_t(\wcheap), \wcheap) \mid \wcheap \ge 0 \right\}}_{(A)}\right) \cup \underbrace{\left\{ (0,0)
%w \mid c(w) = 0 
\right\}}_{(B)}.\]
\end{proposition}

Proposition \ref{prop:positivecorrelation} guarantees positive correlation within each of (at most) $|\mathcal{T}|$ curves in the support. While this does not guarantee positive correlation across the full support in general, it does imply this global form of positive correlation for the \textit{homogeneous users}. We make this explicit in the following corollary of Proposition \ref{prop:positivecorrelation}. 
\begin{theorem}
\label{thm:positivecorrelationhomogeneous}
Suppose that users are homogenous (i.e. $\mathcal{T} = \left\{t\right\}$) and gaming is not costless (i.e. $\nabla(c(w))_2 > 0$ for all $w \in \mathbb{R}_{\ge 0}^2$). Let $(\mu_1, \mu_2, \ldots, \mu_P)$ be any (mixed) Nash equilibrium in the game with $M = \MPlatform$, and let $w^1, w^2 \in \cup_{i \in [P]} \text{supp}(\mu_i)$ be any two pieces of content in the support. If $\wcheap^2 \ge \wcheap^1$, then $\wcostly^2 \ge \wcostly^1$.
\end{theorem}
Theorem \ref{thm:positivecorrelationhomogeneous} shows that a creator's investment in quality content weakly increases with the creator's utilization of gaming tricks. This illustrates a positive correlation between gaming tricks and investment in quality in the content landscape. Perhaps surprisingly, this positive correlation indicates that even high-quality content on the content landscape will have clickbait headlines or exhibit other gaming tricks. Thus, gaming tricks and investment should be viewed as \textit{complements} rather than substitutes.

We provide a proof sketch of Proposition \ref{prop:positivecorrelation} (Theorem \ref{thm:positivecorrelationhomogeneous} follows immediately as a corollary).
\begin{proof}[Proof sketch of Proposition \ref{prop:positivecorrelation}]
Let us first interpret the two types of sets in Proposition \ref{prop:positivecorrelation}. For each $t$, the set (A) is a one-dimensional curve specified by $f_t$ where the costly component is weakly increasing in the cheap component. 
% In the proof of Theorem \ref{thm:positivecorrelation}, we
We construct $f_t$ to be the \textit{minimum-investment function}
\[f_t(\wcheap) = \inf \left\{\wcostly \mid \wcostly \ge 0, u([\wcostly, \wcheap], t) \ge 0 \right\}.\]
The value $f_t(\wcheap)$ captures the minimum investment level in quality needed to achieve nonnegative utility for type $t$ users, given $\wcheap$ utilization of gaming tricks. For example, the function $f_t$ takes the following form in \Cref{example:linear}:
\begin{example}[continues=example:linear]% 
The function $f_t$ can be taken to be $f_t(\wcheap) = \max(0, (\wcheap/t) - \alpha)$ (this follows from Lemma \ref{lemma:one-to-one} in \Cref{appendix:auxlemmas} and Lemma \ref{lemma:ctaugunion} in \Cref{appendix:proofssecpositivecorrelation}). 
As $t$ increases (and users becomes more tolerant to gaming tricks), the slope of $f_t$ decreases. As a result, an increase in utilization of gaming tricks results in less of an increase in investment in quality. 
\end{example}
\noindent The set (B) of costless actions captures creators ``opting out'' of the game by not expending any costly effort in producing their content.

We show that the set (A) captures all of the content that a creator might reasonably select if they are optimizing for being recommended to a user with type $t$. In particular, if a creator is optimizing for type-$t$ users, they will invest the minimum amount in quality to maintain nonnegative utility on those users. We further 
% In the proof, we 
show that when best-responding to the other content creators, a creator will either optimize for winning one of the user types $t \in \mathcal{T}$ or opt out by expending no costly effort. We defer the full proof to \Cref{appendix:proofssecpositivecorrelation}. 
\end{proof}

\subsection{Empirical analysis on Twitter dataset}\label{subsec:empirical}

\begin{figure*}[t]
    \centering
        \begin{subfigure}[b]{0.32\textwidth}
        \centering
        \includegraphics[width=\textwidth]{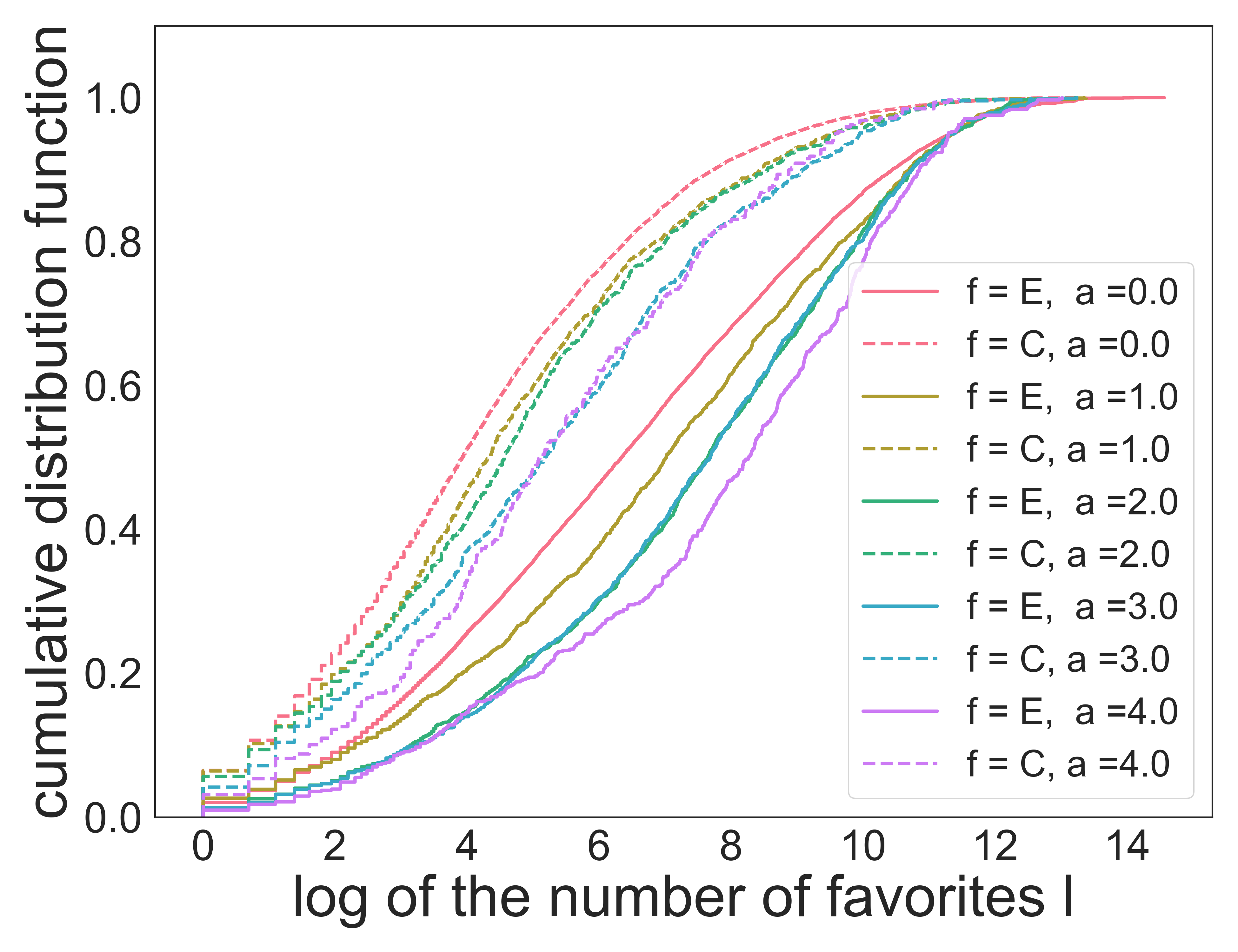}
        \caption{$\mathcal{G} = \left\{ P, \neg P \right\}$}
        \label{fig:all}
    \end{subfigure}
    \begin{subfigure}[b]{0.32\textwidth}
        \centering
        \includegraphics[width=\textwidth]{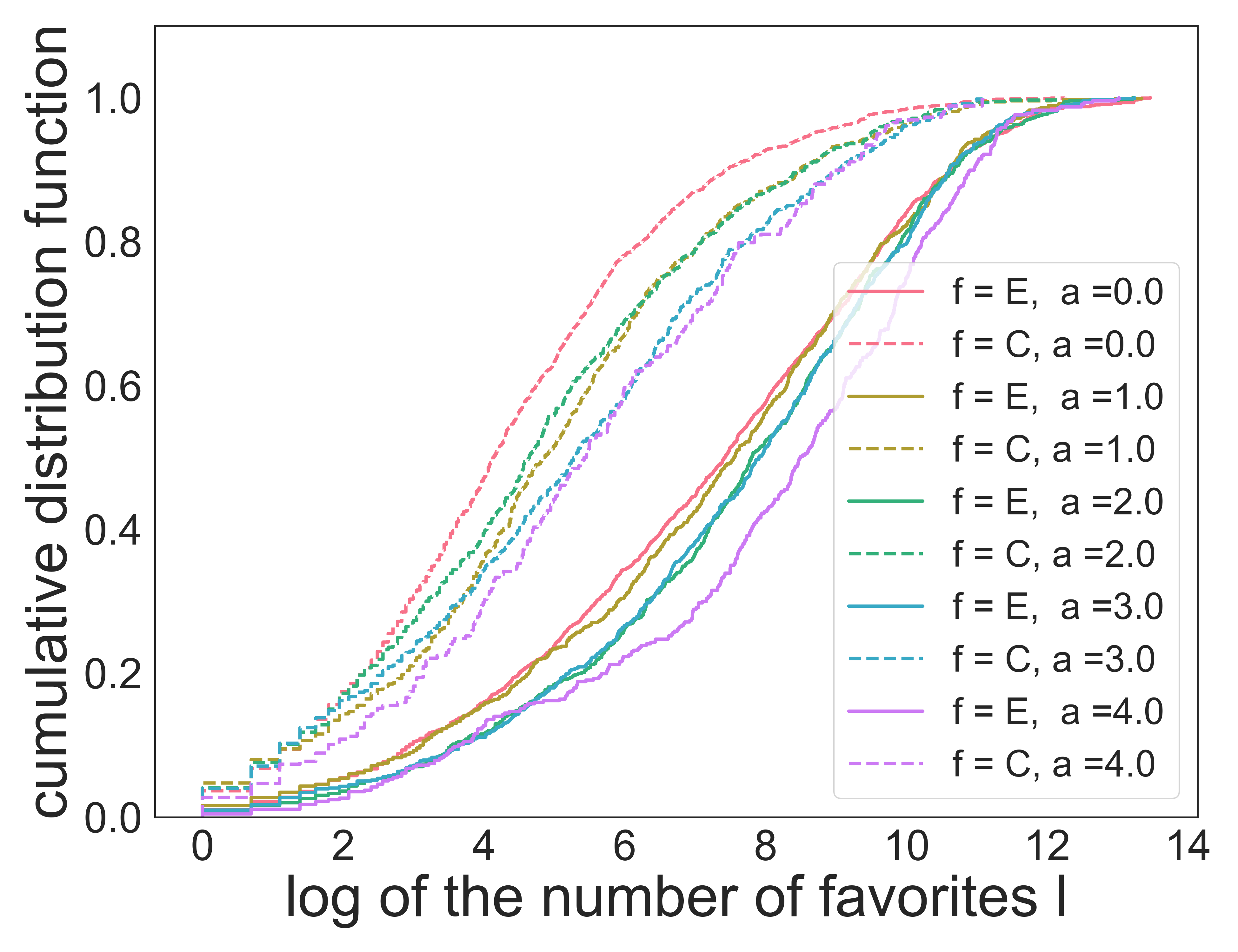}
        \caption{$\mathcal{G} = \left\{ P \right\}$}
        \label{fig:political}
    \end{subfigure}
    \hfill
    \begin{subfigure}[b]{0.32\textwidth}
        \centering
        \includegraphics[width=\textwidth]{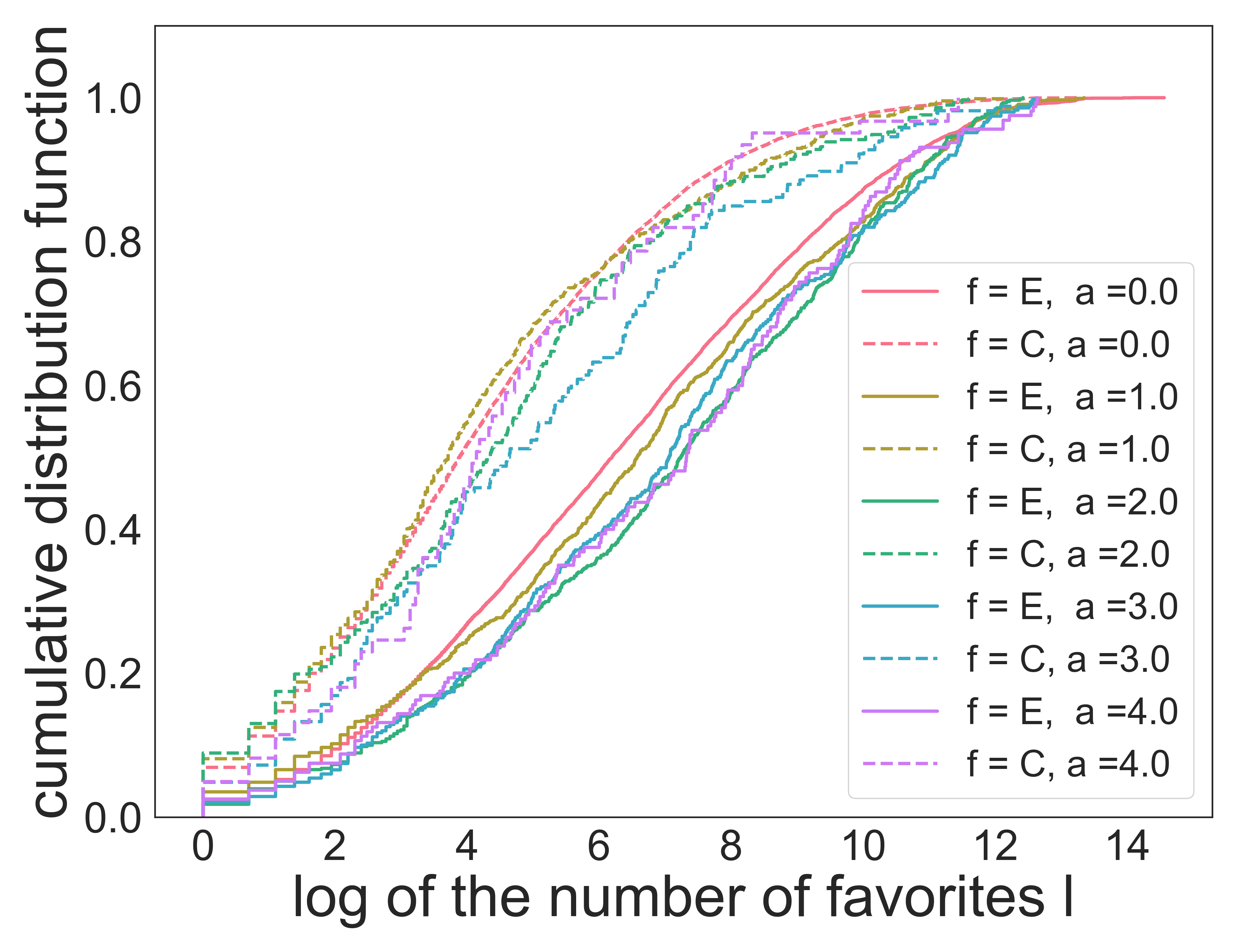}
        \caption{$\mathcal{G} = \left\{ \neg P \right\}$}
        \label{fig:entertainment}
    \end{subfigure}
    \caption{Cumulative distribution function $H_{a, f, \mathcal{G}}$ of the number of favorites ($\wcostly = l$) conditioned on different angriness levels ($\wcheap = a$) on a dataset \citep{MCPWD23} of tweets from the engagement-based feeds ($f = E$) and chronological feeds ($f = C$). The tweet genre is unrestricted (left), restricted to political tweets (middle), and restricted to not political tweets (right). The cdf for higher values of $a$ appears to stochastically dominate the cdf for lower values of $a$, suggesting a positive correlation between $\wcheap$ and $\wcostly$. The stochastic dominance is more pronounced for political tweets than for non-political tweets, and it occurs for engagement-based and chronological feeds. }
    \label{fig:positivecorrelation}
\end{figure*}

We next provide empirical validation for the positive correlation between gaming and investment on a Twitter dataset \citep{MCPWD23}. 
The dataset consists of survey responses from 1730 participants, each of whom was asked several questions about each of the top ten tweets in their personalized and chronological feeds. Using the user survey responses, we associate each tweet with a tuple:
\[(f, g, a, l)  \in \left\{E, C \right\} \times \left\{P, \neg P \right\} \times \left\{0, 1, 2, 3 , 4\right\} \times \mathbb{Z}_{\ge 0}.\]
The \textit{feed} $f \in \left\{E, C\right\}$ captures whether the tweet was in the user's engagement-based feed ($f = E$) or chronological feed ($f = C$). The \textit{genre} $g \in \left\{P, \neg P \right\}$ captures whether the user labelled the content as in the political genre $(g = P)$ or not $(g = \neg P)$. The \textit{angriness level} $a \in \left\{0, 1, 2, 3 , 4\right\}$ captures the reader's evaluation of how angry the author appears in their tweet, rated numerically between $0$ and $4$.\footnote{The survey question asked: ``How is [author-handle] feeling in their tweet?'' \citep{MCPWD23}} The \textit{number of favorites} $l \in \mathbb{Z}_{\ge 0}$ captures the number of favorites (i.e. ``heart reactions'') of the tweet. Let $D$ be the multiset $D$ of tuples from the tweets in the dataset, and let $\mathcal{D}$ be the distribution where $(f, g, a, l)$ is drawn uniformly from the multiset $D$.  

We map this empirical setup to \Cref{example:linear} as follows. Since $\wcheap$ is intended to capture the offensiveness of content in \Cref{example:linear}, we estimate $\wcheap$ by the angriness level $a$. Since $\wcostly$ is intended to capture the costly investment into content quality in \Cref{example:linear}, we estimate $\wcostly$ by the number of favorites $l$. We expect that increasing author angriness $\wcheap$ 
decreases user utility, drawing upon intuition from \citet{M20} that incendiary or divisive content drives engagement by provoking outrage in users. Furthermore, we expect that higher quality content would generally receive more favorites $\wcostly$ and lead to higher user utility, (if the author angriness level is held constant).\footnote{The dataset \citep{MCPWD23}  also includes other author and reader emotions besides author angriness (such as author happiness or reader sadness). The reason that we focus on author angriness is that we believe it to most closely match the interpretation of ``gaming tricks'' in our model: while we expect increasing author angriness to decrease user utility (as described above), we might not expect increasing other emotions, such as author happiness, to decrease user utility.}

We analyze the relationship between the number of favorites ($\wcostly$) and the angriness $(\wcheap)$ with two different approaches: 
\begin{itemize}[leftmargin=*]
\item \textit{Stochastic dominance of conditional distributions:} Given an angriness level $a \in \left\{0, 1, 2, 3 , 4\right\}$, feed $f \in \left\{E, C \right\} $ and subset of genres $\mathcal{G} \subseteq \left\{P, \neg P \right\}$, consider the random variable $\ln(L)$ where $(F, G, A, L)$ is drawn from the conditional distribution $\mathcal{D} \mid (A = a, F = f, G \in \mathcal{G})$. We let $H_{a, f,\mathcal{G}}$ denote the cumulative density function of this random variable. We visually examine the extent to which $H_{a, f, \mathcal{G}}$ stochastically dominates $H_{a', f, \mathcal{G}}$ when $a > a'$. 
\item \textit{Correlation coefficient:} Given a feed $f \in \left\{E, C \right\} $ and subset of genres $\mathcal{G} \subseteq \left\{P, \neg P \right\}$, we compute the multiset 
\begin{equation}
 S_{f, \mathcal{G}}:= \left\{(a, l) \mid (f, g, a, l) \in D \mid f = f', g' \in \mathcal{G} \right\}
\end{equation}
We compute the Spearman's rank correlation coefficient $\rho_{f, \mathcal{G}} \in [-1, 1]$ of the multiset $S_{f, \mathcal{G}}$ and a corresponding p-value $p_{f, \mathcal{G}}$.\footnote{The p-value is for a one-sided hypothesis test with null hypothesis that $a$ and $l$ have no ordinal correlation, calcuated using the scipy.stats.spearmanr Python library. } 
\end{itemize}

\begin{table}
    \centering
    \begin{tabular}{c|c|c|c}
         & $\mathcal{G} = \left\{P, \neg P \right\}$ & $\mathcal{G} = \left\{P \right\}$  & $\mathcal{G} = \left\{\neg P \right\}$ \\ \hline
      $f = E$ & $0.131$  & $0.092$  & $0.048$  \\  
      & $(2  \cdot 10^{-76})$ & $(2 \cdot 10^{-10})$ & $(3  \cdot 10^{-9})$  \\\hline
      $f = C$ & $0.086$ & $0.138$  & $0.004$  \\
     & $(2 \cdot 10^{-33})$  & $(4.49  \cdot 10^{-19})$ & $(3 \cdot 10^{-1})$ \\
    \end{tabular}
    \caption{Correlation coefficient $\rho_{f, \mathcal{G}}$ (with $p$-value $p_{f, \mathcal{G}}$ in parentheses) between the number of favorites ($\wcostly = l$) and the angriness level ($\wcheap = a$) on a dataset \citep{MCPWD23} of tweets from the engagement-based feeds ($f = E$) and chronological feeds ($f = C$) and across political ($P$) and non-political ($\neg P$) tweets. The correlation coefficient is positive (though weak) and statistically significant in all cases except for non-political tweets in the chronological feed. Moreover, correlations are stronger for political than for non-political tweets. }
    \label{tab:correlation_coefficients}
\end{table}

\paragraph{Stochastic dominance of conditional distributions.} Figure \ref{fig:positivecorrelation} shows the cumulative distribution function $H_{a, f, \mathcal{G}}$ for different values of $a$, $f$, and $\mathcal{G}$. The primary finding is that in all of the plots, the cdf for higher values of $a$ visually appears to stochastically dominate the cdf for lower values of $a$. This stochastic dominance reflects a higher author's angriness level $\wcheap = a$ leads to higher numbers of favorites $\wcostly = l$, thus suggesting that content with higher levels of gaming $\wcheap$ also exhibit higher quality $\wcostly$. 

Interestingly, the stochastic dominance is most pronounced when  $\mathcal{G} = \left\{P, \neg P \right\}$ and $\mathcal{G} = \left\{P\right\}$, but less pronounced when $\mathcal{G} = \left\{\neg P\right\}$. This aligns with the intuition that increasing author angriness more effectively increases engagement for political tweets than for non-political tweets.\footnote{For non-political tweets, we expect other types of gaming tricks are employed.} Moreover, within $\mathcal{G} = \left\{P, \neg P \right\}$ and $\mathcal{G} = \left\{P \right\}$, the stochastic dominance occurs for both $f = E$ and $f= C$. We view each of $f = E$ and $f= C$ as capturing a different slice of the content landscape: the fact that stochastic dominance occurs in two different slices suggests it broadly occurs in the content landscape.

\paragraph{Correlation coefficient.} Table \ref{tab:correlation_coefficients} shows $\rho_{f, \mathcal{G}}$ for different genres of tweets and feeds. Interestingly, the correlation coefficient is positive in all cases, which suggests that content with higher levels of gaming tend to exhibit higher levels of investment in quality. However, the correlation is somewhat weak: this may be due to angriness ratings being incomparable across different survey participants. Nonetheless, the correlation is stronger for political content, which again aligns with the intuition that increasing author angriness is more effective in increasing engagement for political tweets.\footnote{For many other emotions (both positive and negative) measured in \citep{MCPWD23}, the analogous correlation coefficients are also positive. For negative emotions, these coefficients could also be interpreted as correlations between gaming tricks and quality within our model. On the other hand, for positive emotions, where increasing the level of the positive emotion might \textit{increase} (rather than decrease) user utility, the resulting correlation coefficient does not have a clear interpretation within our model. }

\section{Performance of engagement-based optimization at equilibrium}\label{sec:performance}

\begin{figure}[t]
    \centering
        \begin{subfigure}[b]{0.32\textwidth}
        \centering
        \includegraphics[width=\textwidth]{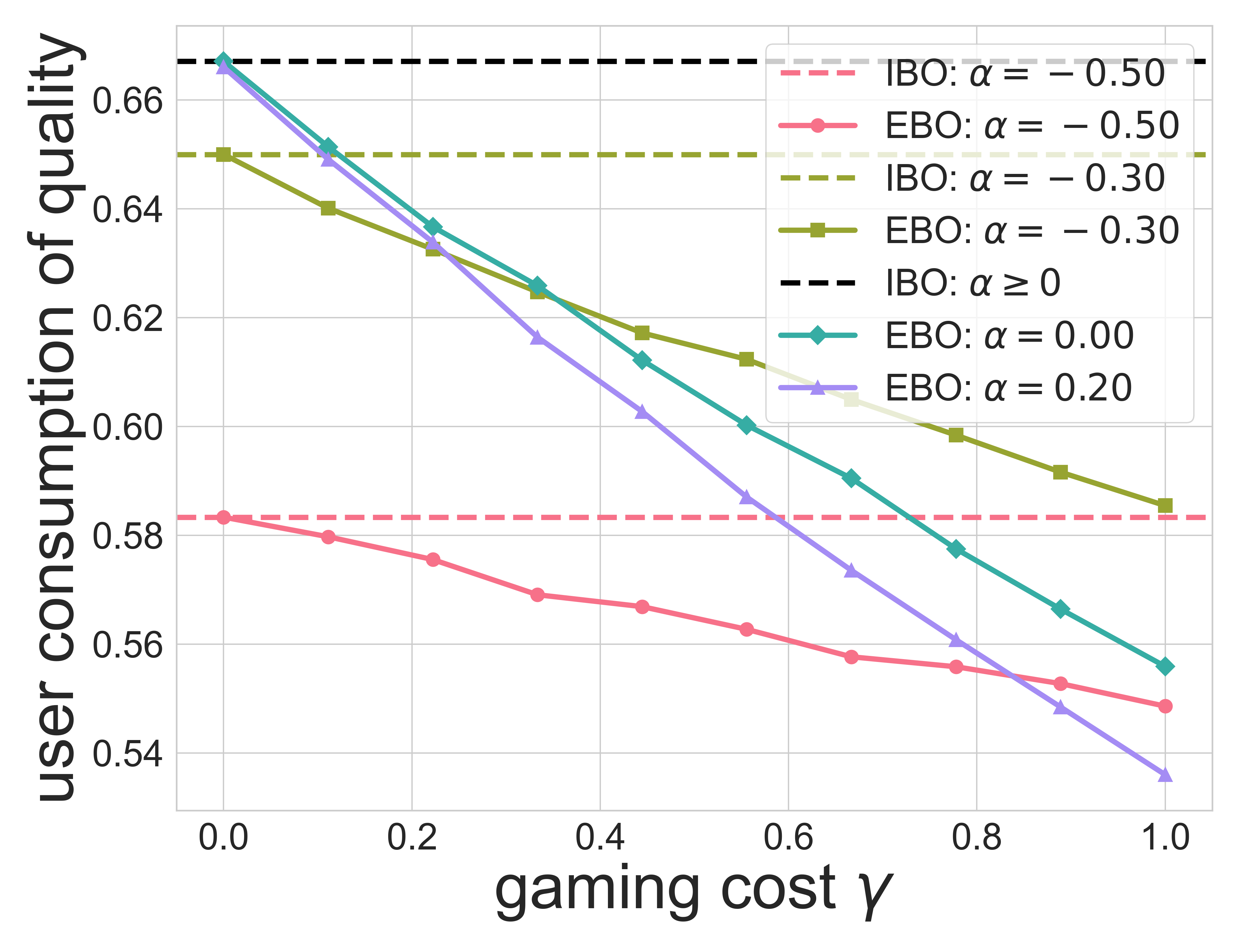}
        \caption{User consumption of quality}
        \label{fig:qualityconsumption}
    \end{subfigure}
    \begin{subfigure}[b]{0.32\textwidth}
        \centering
        \includegraphics[width=\textwidth]{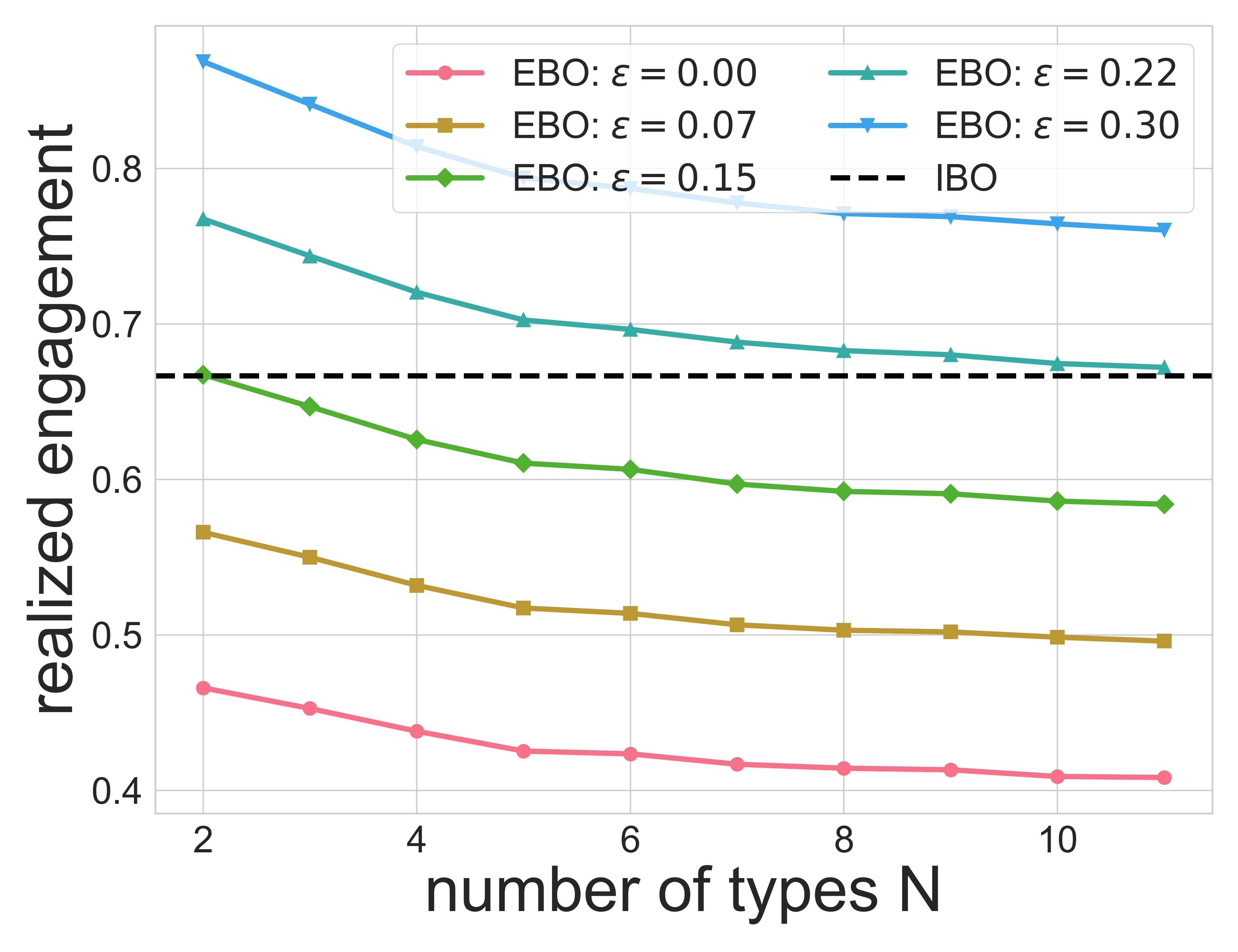}
        \caption{Realized engagement}
        \label{fig:realizedengagement}
    \end{subfigure}
    \hfill
    \begin{subfigure}[b]{0.32\textwidth}
        \centering
        \includegraphics[width=\textwidth]{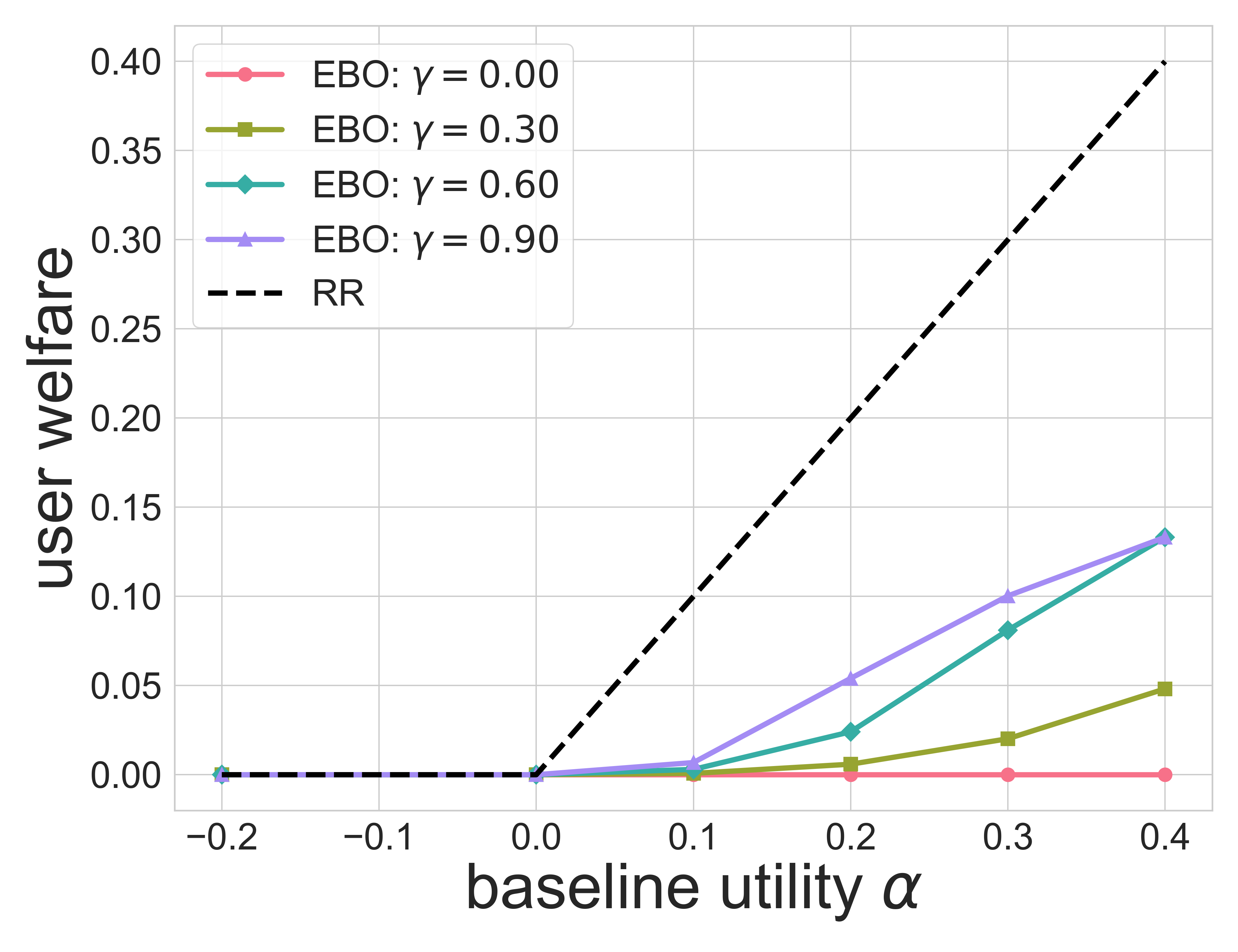}
        \caption{User welfare}
        \label{fig:userwelfare}
    \end{subfigure}
    \caption{Equilibrium performance of engagement-based optimization (EBO) in \Cref{example:linear} with $P = 2$ creators along several performance axes (left to right). The performance is numerically estimated from 100,000 samples from the equilibrium distributions (\Cref{sec:equilibriumcharacterizations}). The parameter settings are $\mathcal{T} = \left\{1\right\}$ (left), $\mathcal{T} = \mathcal{T}_{N, \epsilon}$, $\alpha = 1$, and $\gamma = 0$ (middle), and $\mathcal{T} = \left\{5\right\}$ (right). The equilibrium performance of investment-based optimization (IBO) and random recommendations (RR) are analytically computed from the equilibrium distributions (Theorem \ref{thm:investmentbased} and Theorem \ref{thm:randomrecommendations}) and shown as baselines. User consumption of quality  can decrease with gaming costs (left; Theorems \ref{thm:comparisonuserconsumptiongaming}-\ref{thm:comparisonuserconsumptioninvestment}), realized engagement can be lower for EBO than for IBO (middle; Theorem \ref{thm:comparisonengagement}), and user welfare can be lower for EBO than for RR (right; Theorem \ref{thm:comparisonuserutility}).}
    \label{fig:downstreamimpacts}
\end{figure}

In this section, taking into account the structure of the the content landscape at equilibrium, we investigate the downstream performance of engagement-based optimization.  As baselines, we consider investment-based optimization (an idealized baseline that optimizes directly for quality $\MIdeal(w) = \wcostly$) and random recommendations (a trivial baseline that results in randomly choosing from content that achieves nonnegative user utility). We highlight striking aspects of these comparisons (Figure \ref{fig:downstreamimpacts}), considering three qualitatively different performance axes: user consumption of quality (\Cref{sec:equilibriuminvestment}),  realized engagement (\Cref{sec:platformimplications}), and user utility (\Cref{subsec:userimplications}).   

Our comparisons take into account the \textit{endogeneity of the content landscape}: i.e., that the content landscape at equilibrium depends on the choice of metric. The possibility of multiple equilibria casts ambiguity on which equilibrium to consider. To resolve this ambiguity, we will focus on the (symmetric mixed) equilibria from our characterization results in \Cref{sec:equilibriumcharacterizationsbaselines} and \Cref{sec:equilibriumcharacterizations}.  That is, throughout this section we will focus on equilibria  $\muengagement(P, c, u, \mathcal{T})$ (for engagement-based optimization), $\muideal(P, c, u, \mathcal{T})$  (for investment-based optimization), and  $\murandom(P, c, u, \mathcal{T})$ (for random recommendations).

For ease of exposition, in this section, we focus on Example \ref{example:linear} for different parameter settings of the baseline utility $\alpha$, the gaming cost level $\gamma$, the number of creators $P$, and the type space $\mathcal{T}$.  The results in this section directly translate to other instantations of our model including \Cref{example:KMR}.

\subsection{User consumption of quality}\label{sec:equilibriuminvestment}
We first consider the average quality of content consumed by users (formalized below), focusing on the case of homogeneous users in \Cref{example:linear}. We show that as gaming costs increase, the performance of engagement-based optimization \textit{worsens}; in fact, unless gaming is \textit{costless}, engagement-based optimization performs strictly worse than investment-based optimization.

We formalize user consumption of quality by: 
\[\UserConsumption(M; \mathbf{w}) := \mathbb{E}\left[
\MIdeal(w_{i^*(M; \textbf{w})}) \cdot \mathbbm{1}[u(w_{i^*(M; \textbf{w})}, t) \ge 0]\right],\]
which only counts content quality if the content is actually consumed by the user. 
Taking into account the endogeneity of the content landscape, the user consumption of quality at a symmetric mixed Nash equilibrium $\mu^M$ is:
\[\mathbb{E}_{\mathbf{w} \sim (\mu^M)^P}
%(\mu^{M}_1 \times \ldots \times \mu^M_P)} 
\left[\UserConsumption(M; \mathbf{w})\right].\]

The following result shows that in \Cref{example:linear} the average user consumption of quality strictly \textit{decreases} as gaming costs (parameterized by $\gamma$) become more expensive (Figure \ref{fig:qualityconsumption}).   
\begin{theorem}
\label{thm:comparisonuserconsumptiongaming}
Suppose that users are homogeneous (i.e. $\mathcal{T} = \left\{t\right\}$). For any sufficiently large baseline utility $\alpha > -1$, bounded gaming costs $\gamma \in [0,1)$, and any number of creators $P \ge 2$, the user consumption of quality $\mathbb{E}_{\mathbf{w} \sim \left(\muengagement(P, c, u, \mathcal{T})\right)^P}[\UserConsumption(\MPlatform; \mathbf{w})]$ for engagement-based optimization is strictly decreasing in $\gamma$.
\end{theorem}
\begin{proof}[Proof sketch of Theorem \ref{thm:comparisonuserconsumptiongaming}]
For sufficiently large values of $\wcostly$, creators  compete their utility down to $0$, so the only remaining strategic choice is how they choose to trade off effort spent on gaming versus investment. If gaming is costly, then creators need to expend more of their effort on gaming to achieve a desired increase in engagement, so they will necessarily devote less effort to investment in quality. In contrast, if gaming is costless, creators devote all of their effort to investment. To formalize this intuition, we explicitly compute user consumption of quality using the equilibrium characterization. We defer the proof to \Cref{appendix:comparisonuserconsumptiongaming}.
\end{proof}

Theorem \ref{thm:comparisonuserconsumptiongaming} thus has a striking consequence for platform design: to improve user consumption of quality, it can help to reduce the costs of gaming tricks as much as possible. One concrete approach for reducing gaming costs is to increase the \textit{transparency} of the platform's metric, for example by publishing the metric in an interpretable manner. In particular, if a content creator does not have access to the platform's metric, they would have to expend effort to learn the metric to game it; on the other hand, transparency would reduce these costs. Perhaps countuitively, our results suggest that increasing transparency can \textit{improve} user consumption of quality in the presence of strategic content creators.\footnote{This finding bears some resemblance to results in the strategic classification literature \citep{GNETR21, BPWZ22}. For example, \citet{GNETR21} shows that transparency is the optimal policy in terms of optimizing the decision-maker's accuracy. However, a lack of transparency is suboptimal in \cite{GNETR21} because it prevents the decision-maker from being able to fully anticipating strategic behavior; in contrast, a lack of transparency is suboptimal in our setting because it leads effort to be spent on figuring how to game the classifier rather than investing in quality.} In particular, our results suggest the recent trend  
of recommender systems publishing their algorithms (e.g., \citet{twitter2023algorithm})  may improve user consumption of quality content, and encourage the continued release of recommendation algorithms more broadly. 

To further understand the impact of gaming costs $\gamma$, we compare the performance of engagement-based optimization with the performance of investment-based optimization (which does not depend on $\gamma$). We treat the performance of investment-based optimization as an ``idealized baseline'' for $\UserConsumption$: the reason is that for any \textit{fixed} content landscape $\mathbf{w}$, investment-based optimization maximizes the $\UserConsumption(M; \textbf{w})$ across all possible metrics $M$, because the objectives exactly align. 
The following result shows that engagement-based optimization performs strictly worse than investment-based optimization unless gaming tricks are \textit{costless} (Figure \ref{fig:qualityconsumption}). 
\begin{theorem}
\label{thm:comparisonuserconsumptioninvestment}
Suppose that users are homogeneous (i.e. $\mathcal{T} = \left\{t\right\}$). 
For any sufficiently large baseline utility $\alpha > -1$, bounded gaming costs $\gamma \in [0,1)$, and any number of creators $P \ge 2$, it holds that:
\[\mathbb{E}_{\mathbf{w} \sim \left(\muengagement(P, c, u, \mathcal{T})\right)^P}[\UserConsumption(\MPlatform; \mathbf{w})] \le \mathbb{E}_{\mathbf{w} \sim \left(\muideal(P, c, u, \mathcal{T})\right)^P}[\UserConsumption(\MIdeal; \mathbf{w})],\]
with equality if and only if $\gamma = 0$. 
\end{theorem}

Theorem \ref{thm:comparisonuserconsumptioninvestment} illustrates that reducing the gaming costs to $0$ is \textit{necessary} for engagement-based optimization to perform as well as the idealized baseline. This serves as a further motivation for a social planner to try to reduce gaming costs as much as possible, for example through increased transparency as discussed above. 

We caution that reducing gaming costs to $0$ is not \textit{sufficient} to guarantee that engagement-based optimization performs as well as investment-based optimization, if users are heterogeneous. The following result shows that for heterogeneous users, the average user consumption of quality of engagement-based optimization can be significantly lower than the average user consumption of quality of investment-based optimization. 
\begin{proposition}
\label{prop:userconsumptionNtypes}
For any $N \ge 2$, there exists an instance with $\gamma = 0$ and a type space $\mathcal{T}$ of $N$ well-separated types such that the average user consumption of quality of engagement-based optimization is less than the average user consumption of quality of investment-based optimization: 
\[\mathbb{E}_{\mathbf{w} \sim \left(\muengagement(P, c, u,  \mathcal{T})\right)^P}[\UserConsumption(\MPlatform; \mathbf{w})] \le \frac{1}{N} < \frac{2}{3} = \mathbb{E}_{\mathbf{w} \sim \left(\muideal(P, c, u, \mathcal{T})\right)^P}[\UserConsumption(\MIdeal; \mathbf{w})].\]   
\end{proposition}
Proposition \ref{prop:userconsumptionNtypes}  illustrates that it is possible for the average user consumption of quality of engagement-based optimization to approach $0$ as $N \rightarrow \infty$ while the performance of investment-based optimization stays constant and nonzero, when gaming costs are reduced to $0$.
This result suggests that other interventions, beyond reducing gaming costs, may be necessary to ensure that engagement-based optimization does not substantially degrade the overall quality of content being consumed by users.

\subsection{Realized engagement}\label{sec:platformimplications}
We next consider how much engagement is realized by user consumption patterns, when accounting for the fact that users only consume recommendations that generate nonnegative utility for them. We show that even though engagement-based optimization maximizes realized engagement on any \textit{fixed} content landscape, engagement-based optimization can be suboptimal at equilibrium, when taking into account the endogeneity of the content landscape.  

We formalize realized engagement by 
\[\RealizedEngagement(M; \textbf{w}) := \mathbb{E}[\MPlatform(w_{i^*(M; \textbf{w})}) \cdot \mathbbm{1}[u(w_{i^*(M; \textbf{w})}, t)\ge 0]].\]
When taking into account the endogeneity of the content landscape, the realized engagement at a symmetric mixed Nash equilibrium $\mu^M$ is $\mathbb{E}_{\mathbf{w} \sim (\mu^M)^P} \left[\RealizedEngagement(M; \mathbf{w})\right]$.

To show the suboptimality of engagement-based optimization at equilibrium, we construct an instance where engagement-based optimization performs strictly worse than investment-based optimization (Figure \ref{fig:realizedengagement}). 
\begin{theorem}
\label{thm:comparisonengagement}
For sufficiently large $N$, there exists an instance with a type space $\mathcal{T}$ of $N$ well-separated types such that the realized engagement of engagement-based optimization is less than the realized engagement of investment-based optimization: 
\[\mathbb{E}_{\mathbf{w} \sim \left(\muengagement(P, c, u, \mathcal{T})\right)^2}[\RealizedEngagement(\MPlatform; \mathbf{w})] < \mathbb{E}_{\mathbf{w} \sim \left(\muideal(P, c, u, \mathcal{T})\right)^2}[\RealizedEngagement(\MIdeal; \mathbf{w})].\]
\end{theorem}
\begin{proof}[Proof sketch of 
Theorem \ref{thm:comparisonengagement}]

The main ingredient of the proof of Theorem \ref{thm:comparisonengagement} is constructing and analyzing an instance where engagement-based optimization achieves a low realized engagement. 
We construct the instance to be Example \ref{example:linear} with costless gaming ($\gamma = 0$), baseline utility $\alpha = 1$, and $P = 2$ creators, and type space $\mathcal{T}_{N, \epsilon} := \left\{(1+ \epsilon) (1 + 1/N)^{i-1} - 1 \mid 1 \le i \le N \right\}$ for sufficiently small $\epsilon > 0$ and sufficiently large $N$.\footnote{While Theorem \ref{thm:comparisonengagement} focuses on the limit as $N \rightarrow \infty$, the numerical estimates shown in Figure \ref{fig:realizedengagement} suggest that the result applies for any $N \ge 2$.}

One key aspect of $\mathcal{T}_{N,\epsilon}$ is that the heterogeneity in user types segments the market and significantly reduces investment in quality. In particular, since user types are well-separated in the type space $\mathcal{T}_{N, \epsilon}$, creators can't realistically compete for multiple types at the same time and must choose a single type to focus on. At a high-level, this segments the market, and a single creator can only hope to win $O(1/N)$ users and thus only invests $O(1/N)$ in costly effort. (Note that there are some subtleties, because a creator who targets a lower type might win a higher type if none of the other creators target the higher type in that particular realization of randomness.) In the limit as $N \rightarrow \infty$, we show that the investment in quality approaches $0$. 

However, for engagement-based optimization, a low investment in quality does not directly imply a low realized engagement. This is because if users are high type (and thus highly tolerant of gaming tricks), creators can utilize a high level of gaming tricks without investing at all in quality, while still maintaining nonnegative utility for these users. Thus, to show that the realized engagement is low, we must consider the distribution over user types. 
The construction of $\mathcal{T}_{N,\epsilon}$  appropriately balances two forces: (1) making the user types sufficiently well-separated to reduce the investment in quality, and (2) making the user types as low as possible to reduce engagement from gaming tricks. 

To analyze the realized engagement of this construction, we first upper bound $\RealizedEngagement(\MPlatform; \textbf{w})$ by the maximum engagement achieved by any content in the content landscape $\max_{w \in \textbf{w}} \MPlatform(w)$. It then remains to analyze the engagement distribution of $\MPlatform(w)$ for $w$ in the equilibrium distribution. Although the cdf of the engagement distribution is messy for any given value of $N$, it approaches a continuous distribution in the limit. This enables us to show that:
  \[\limsup_{\epsilon \rightarrow 0} \limsup_{N \rightarrow \infty} \mathbb{E}_{\mathbf{w} \sim \left(\muengagement(2, c, u, \mathcal{T}_{N, \epsilon})\right)^2}[\RealizedEngagement(\MPlatform; \mathbf{w})] < \mathbb{E}_{\mathbf{w} \sim \left(\muideal(2, c, u, \mathcal{T}_{N, \epsilon})\right)^2}[\RealizedEngagement(\MIdeal; \mathbf{w})]. \]
We defer the proof to \Cref{appendix:proofrealizedengagement}. 
  
\end{proof}

Theorem \ref{thm:comparisonengagement} has an interesting platform design consequence: even if the platform wants to optimize realized engagement (e.g., because their revenue comes from advertising), engagement-based optimization is not necessarily the optimal approach. In particular, Theorem \ref{thm:comparisonengagement} illustrates potential benefits of using investment-based optimization, even though $\MIdeal$ does not directly reward engagement. A practical challenge is that investment-based optimization is often difficult for the platform to implement because $\wcostly$ may not be directly observable; nonetheless, the platform may be able to perform a noisy version of investment-based optimization by collecting (sparse) feedback about content quality from users. Theorem \ref{thm:comparisonengagement} raises the possibility that noisy versions of investment-based optimization may be worthwhile for the platform to pursue, even if the platform's goal is to maximize realized engagement. 

As a caveat, Theorem \ref{thm:comparisonengagement} does rely on users types being heterogeneous. In fact, the following result shows that for homogeneous users, engagement-based optimization performs at least as well as investment-based optimization in terms of realized engagement.
\begin{proposition}
 \label{prop:comparisonengagementhomogeneous}
Suppose that users are homogeneous ($\mathcal{T} = \left\{t\right\}$). For any sufficiently large baseline utility $\alpha > -1$, bounded gaming costs $\gamma \in [0, 1)$, and any number of creators $P \ge 2$, the realized engagement of engagement-based optimization is at least as large as the the realized engagement of investment-based optimization:
 \[\mathbb{E}_{\mathbf{w} \sim \left(\muengagement(P, c, u, \mathcal{T})\right)^P}[\RealizedEngagement(\MPlatform; \mathbf{w})] \ge \mathbb{E}_{\mathbf{w} \sim \left(\muideal(P, c, u, \mathcal{T})\right)^P}[\RealizedEngagement(\MIdeal; \mathbf{w})]. \]
\end{proposition}
\noindent While Proposition \ref{prop:comparisonengagementhomogeneous} does show that engagement-based optimization can generate nontrivial realized engagement when users are homogeneous, we expect that when the user base exhibits sufficient diversity in tolerance towards gaming tricks, investment-based optimization would be an appealing alternative to engagement-based optimization.

\subsection{User welfare}\label{subsec:userimplications}

Finally, we consider user utility realized by user consumption patterns, which can be interpreted as \textit{user welfare}. We show that engagement-based optimization can alarmingly perform worse than random recommendations in terms of user welfare.\footnote{We view random recommendations as a conservative baseline, since $\MTrivial$ does not reward investment or gaming.} 

We formalize user welfare by 
\[\UserUtility(M; \textbf{w}) := \mathbb{E}[u(w_{i^*(M; \textbf{w})}, t) \cdot \mathbbm{1}[u(w_{i^*(M; \textbf{w})}, t) \ge 0]].\]
Taking into account the endogeneity of the content landscape, the user welfare at a symmetric mixed Nash equilibrium $\mu^M$ is $\mathbb{E}_{\mathbf{w} \sim (\mu^M)^P} \left[\UserUtility(M; \mathbf{w})\right]$.  

The following result shows that for homogeneous users, engagement-based optimization always performs at least as poorly as random recommendations, and can even perform strictly worse than random recommendations under certain conditions (Figure \ref{fig:userwelfare}).  
\begin{theorem}
\label{thm:comparisonuserutility}
Suppose that users are homogeneous (i.e. $\mathcal{T} = \left\{t\right\}$), and that gaming costs $\gamma \in [0,1)$ are bounded. If baseline utility $\alpha > 0$ is positive, the user welfare of engagement-based optimization is strictly lower than the user welfare of random recommendations: 
\[\mathbb{E}_{\mathbf{w} \sim \left(\muengagement(P, c, u, \mathcal{T})\right)^P}[\UserUtility(\MPlatform; \mathbf{w})] < \mathbb{E}_{\mathbf{w} \sim \left(\murandom(P, c, u, \mathcal{T})\right)^P}[\UserUtility(\MTrivial; \mathbf{w})].\] If baseline utility $\alpha \le 0$ is nonpositive, engagement-based optimization and random recommendations both result in zero user welfare: 
\[\mathbb{E}_{\mathbf{w} \sim \left(\muengagement(P, c, u, \mathcal{T})\right)^P}[\UserUtility(\MPlatform; \mathbf{w})] = \mathbb{E}_{\mathbf{w} \sim \left(\murandom(P, c, u, \mathcal{T})\right)^P}[\UserUtility(\MTrivial; \mathbf{w})] = 0.\]
\end{theorem}
\begin{proof}[Proof sketch of Theorem \ref{thm:comparisonuserutility}]
We first focus on the simple case where gaming tricks are free $(\gamma = 0)$ and the baseline utility is positive ($\alpha \ge 0$). For engagement-based optimization, creators will increase gaming tricks until the user utility drops down to $0$, which means the user welfare at equilibrium is $0$. In contrast, for random recommendations, creators do not expend effort on either gaming tricks or investment; thus, the user welfare at equilibrium is $u([0, 0], t) > 0$, which is strictly higher than the user welfare for engagement-based optimization. The other cases, though a bit more involved, follow from similar intuition: for engagement-based optimization, creators choose the balance between gaming tricks and investment in quality that drives user utility as close to zero as possible, whereas for random recommendations, creators choose the minimum amount of investment to achieve nonzero user utility. We defer the full proof to \Cref{appendix:comparisonuserutility}. 
\end{proof}

From a platform design perspective, Theorem \ref{thm:comparisonuserutility} highlights the pitfalls of engagement-based optimization for users. In particular, the user welfare of engagement-based optimization can fall below the conservative baseline where users randomly select content on their own (and the content landscape shifts in response). This suggests that engagement-based optimization may not retain users in the long-run, especially in a competitive marketplace with multiple platforms.

It is important to note that for \textit{heterogeneous users}, engagement-based optimization does not always perform as poorly as random recommendations. 
In the following result, we turn to Example \ref{example:KMR} and construct instances with 2 user types where engagement-based optimization outperforms random recommendations.\footnote{While all of our other results in \Cref{sec:performance} apply to Examples \ref{example:linear} and \ref{example:KMR}, Propositions \ref{prop:userutility2types}-\ref{prop:userutility2typeswellseparated} only apply to Example \ref{example:KMR}. The extra factor of $t$ in the utility function formalization in Example \ref{example:KMR} turns out to be necessary for these results.}  
\begin{proposition}
\label{prop:userutility2types}
Consider \Cref{example:KMR}. There exist instances with 2 types where the user welfare of engagement-based optimization is higher than the user welfare of random recommendations: 
\[\mathbb{E}_{\mathbf{w} \sim \left(\muengagement(2, c, u, \mathcal{T})\right)^2}[\UserUtility(\MPlatform; \mathbf{w})] > \mathbb{E}_{\mathbf{w} \sim \left(\muideal_{2, c, u, \mathcal{T}}\right)^2}[\UserUtility(\MTrivial; \mathbf{w})]. \]
\end{proposition}
\noindent On the other hand, we also construct instances with two types where user welfare of random recommendations outperforms the user welfare of engagement-based optimization, thus behaving similarly to the case of homogeneous users. 
\begin{proposition}
\label{prop:userutility2typeswellseparated}
Consider \Cref{example:KMR}. There exist instances with 2 types where the user welfare of engagement-based optimization is lower than the user welfare of random recommendations: 
\[\mathbb{E}_{\mathbf{w} \sim \left(\muengagement(2, c, u, \mathcal{T})\right)^2}[\UserUtility(\MPlatform; \mathbf{w})] < \mathbb{E}_{\mathbf{w} \sim \left(\muideal_{2, c, u, \mathcal{T}}\right)^2}[\UserUtility(\MTrivial; \mathbf{w})]. \]
\end{proposition}
\noindent Interestingly, the construction in
Proposition \ref{prop:userutility2types} relies on the types \textit{not} being too well-separated while the construction in Proposition \ref{prop:userutility2typeswellseparated} relies on the types being sufficiently well-separated. This raises the interesting question of characterizing the relative performance of random recommendation and engagement-based optimization in greater generality, which we defer to future work. 

\section{Equilibrium characterization results for baseline approaches }\label{sec:equilibriumcharacterizationsbaselines}

Within our analysis in \Cref{sec:performance}, we leveraged closed-form equilibrium characterizations for investment-based optimization and random recommendations in several cases. In this section, we state these characterizations. To state our characterizations, we define a distribution $\muideal(P, c, u, \mathcal{T})$ for investment-based optimization and a distribution $\murandom(P, c, u, \mathcal{T})$ for random recommendations. 

Since neither baseline approach directly incentivizes gaming tricks, the distributions $\muideal(P, c, u, \mathcal{T})$ and $\murandom(P, c, u, \mathcal{T})$ both satisfy $\wcheap = 0$ for all $w$ in the support (i.e., the marginal distribution of $\Wcheap$ is a point mass at $0$). We can thus convert the two-dimensional action space into a one-dimensional action space specified by $\wcostly$, where the cost function is 
\begin{equation}
\label{eq:baselinecost}
\CCostlyBaseline(\wcostly) :=  c([\wcostly, 0])    
\end{equation}
and the utility function is: 
\begin{equation}
 \label{eq:baselineutility}  
 \UtilityCostly(\wcostly, t) :=  u([\wcostly, 0], t).
\end{equation}

We place the following structural assumptions on the type space and utilities which simplify the equilibrium structure. For each type $t \in \mathcal{T}$, let $\beta_t$ be the minimum level of investment needed to achieve nonnegative utility:
\[ \beta_t := \min \left\{ \wcostly \mid \UtilityCostly(\wcostly) \ge 0 \right\}.\]
We assume that either (1) users are homogeneous ($\mathcal{T} = 1$), or (2) users are heterogeneous and no user requires investment to achieve nonnegative utility (i.e., $\beta_t = 0$ for all $t \in \mathcal{T}$). 

We now specify the marginal distribution of quality $\Wcostly$ for investment-based optimization (Section \ref{sec:investment}) and random recommendations (Section \ref{sec:randomrecommendations}). 
We defer the proofs to \Cref{appendix:equilibriumcharacterizationsbaselines}.

\subsection{Characterization for investment-based optimization}\label{sec:investment}

We first consider investment-based optimization where $M = \MIdeal$. 

When users are homogeneous ($\mathcal{T} = \left\{t \right\}$), we define the marginal distribution of $\Wcostly$ for $\muideal(P, c, u, \mathcal{T})$ by: 
\[ 
\mathbb{P}[\Wcostly \le \wcostly] = 
\begin{cases}
\min\left(1, \CCostlyBaseline(\beta_{t})\right)^{1/(P-1)} & \text{ if } 0 \le \wcostly \le \beta_{t} \\
\min\left(1, \CCostlyBaseline(\wcostly) \right)^{1/(P-1)} & \text{ if } \wcostly \ge \beta_{t}.
\end{cases}
\]

When users are heterogeneous and $\beta_t = 0$ for all $t \in \mathcal{T}$, we define the marginal distribution of $\Wcostly$ for $\muideal(P, c, u, \mathcal{T})$ by: 
\[ 
\mathbb{P}[\Wcostly \le \wcostly] = 
\min\left(1, \CCostlyBaseline(\wcostly) \right)^{1/(P-1)}.
\]

We show that $\muideal(P, c, u, \mathcal{T})$ is a symmetric mixed equilibrium.  
\begin{theorem}
\label{thm:investmentbased}
Suppose that either (a) $|\mathcal{T}| = 1$ or (b) $\beta_t = 0$ for all $t \in \mathcal{T}$. Then, the distribution $\muideal(P, c, u, \mathcal{T})$ is a symmetric mixed Nash equilibrium in the game with $M = \MIdeal$.
\end{theorem}

\subsection{Characterization for random recommendations}\label{sec:randomrecommendations}

We next consider random recommendations where $M = \MTrivial$. 

First, we consider the case where users are homogeneous (i.e., $\mathcal{T} = \left\{t\right\}$). Let $\kappa$ be minimum cost to achieve $0$ user utility, truncated at $1$: that is, $\kappa := \min\left(1, \CCostlyBaseline(\beta_t)\right)$. Let the probability $\nu$ be defined as follows: $\nu = 0$ if $\kappa \le 1/P$, and otherwise $\nu \in [0,1]$ is the unique value such that  
that $\sum_{i=0}^{P-1}\nu^i =  P \cdot \kappa$. We define the marginal distribution of $\Wcostly$ for $\murandom(P, c, u, \mathcal{T})$ by
\begin{align*}
&\mathbb{P}[\Wcostly = \wcostly] = \begin{cases}
    \nu & \text{ if } \wcostly = 0 \\
    1 - \nu & \text{ if } \wcostly = \beta_t \\
    0 & \text{ otherwise. }
\end{cases}
\end{align*}

When users are heterogeneous and $\beta_t = 0$ for all $t \in \mathcal{T}$, we define the marginal distribution of $\Wcostly$ for $\murandom(P, c, u, \mathcal{T})$ to be a point mass at $\wcostly = 0$. 

We show that $\murandom(P, c, u, \mathcal{T})$ is a symmetric mixed equilibrium.
\begin{theorem}
\label{thm:randomrecommendations}
Suppose that either (a) $|\mathcal{T}| = 1$ or (b) $\beta_t = 0$ for all $t \in \mathcal{T}$. Then, the distribution $\murandom(P, c, u, \mathcal{T})$ is a symmetric mixed Nash equilibrium in the game with $M = \MTrivial$.
\end{theorem}

\section{Equilibrium characterization for engagement optimization}\label{sec:equilibriumcharacterizations}

Within our analysis of the performance of  engagement-based optimization in Section \ref{sec:performance}, we implicitly leveraged closed-form characterizations of the symmetric mixed equilibria for engagement-based optimization in several concrete instantiations. In this section, we state these closed-form characterizations. To state our characterizations, we define a distribution $\muengagement(P, c, u, \mathcal{T})$ over $\mathbb{R}^2_{\ge 0}$, when $\mathcal{T}$ has a single type (Definition \ref{def:homogeneous}), when $\mathcal{T}$ has two types under further assumptions (Definition \ref{def:2types}), 
and when $\mathcal{T}$ has $N$ well-separated types under further assumptions (Definition \ref{def:Ntypes}). We will show $\muengagement(P, c, u, \mathcal{T})$ is a symmetric mixed Nash equilibria for engagement-based optimization in each case.

To simplify the notation in our specification of $\muengagement(P, c, u, \mathcal{T})$, we convert the two-dimensional action space into the following union of $|T|$ one-dimensional curves that specifies the support of the equilibria. We define the \textit{minimum-investment functions} $f_t: \mathbb{R}_{\ge 0} \rightarrow \mathbb{R}_{\ge 0}$, as follows:
\begin{equation}
\label{eq:ft}
 f_t(\wcheap) := \inf \left\{\wcostly \mid \wcostly \ge 0, u([\wcostly, \wcheap], t) \ge 0 \right\}, 
\end{equation}
so $f_t$ captures the amount of investment needed to offset the disutility from $\wcheap$ level of gaming tricks for users of type $t$. Within each one-dimensional curve, the content $w$ is entirely specified by the cheap component $\wcheap$, which motivates us to define a one-dimensional cost function for content along each curve: 
\begin{equation}
\label{eq:ccheap}
 \CCheap_t(\wcheap) := c([f_t(\wcheap), \wcheap]). 
\end{equation}
For example, the functions $f_t$ and $\CCheap_t$ take the following form in \Cref{example:linear}:
\begin{example}[continues=example:linear]% 
The functions $f_t$ and $\CCheap_t$ are as follows: 
\begin{align*}
f_t(\wcheap) &= \max(0, (\wcheap/t) - \alpha) \\
\CCheap_t(\wcheap) &= 
\begin{cases}
\wcheap (\gamma + 1/t) - \alpha & \text{ if } \wcheap > \max(0, t \cdot \alpha) \\
\wcheap \cdot \gamma  & \text{ if } \wcheap \le t \cdot \alpha. 
\end{cases}
\end{align*}
As $t$ increases (and users becomes more tolerant to gaming tricks), the slope of $f_t$ and $\CCheap_t$ both decrease. The minimum-investment $f_t$ is independent of $\gamma$, but the cost function increases with $\gamma$. 
\end{example}

In \Cref{appendix:homogeneouscharacterization}, we focus on homogeneous users. In \Cref{appendix:additionalassumptions}, we state additional assumptions for the case of heterogeneous users. In \Cref{appendix:characterizationNtypes}, we focus on well-separated types, and in \Cref{appendix:characterization2types} we consider two arbitrary types. We defer proofs to \Cref{appendix:proofsequilibriumcharacterizations}.

\subsection{Equilibrium characterization for homogeneous users}\label{appendix:homogeneouscharacterization}
 
We first focus on the case where $\mathcal{T} = \left\{ t \right\}$ has a single type (Figure \ref{fig:1type}).
\begin{definition}
\label{def:homogeneous}
We define the distribution $(\Wcostly, \Wcheap) \sim \muengagement(P, c, u, \mathcal{T})$ over $\mathbb{R}^2_{\ge 0}$ as follows. Let $f_t$ be defined by \eqref{eq:ft}, and let $\CCheap_t$ be defined by \eqref{eq:ccheap}. The marginal distribution $\Wcheap$ is defined by: 
\[\mathbb{P}[\Wcheap \le \wcheap] = \left(\min(1, \CCheap_t(\wcheap))\right)^{1/(P-1)}.\]
For each $\wcheap \in \text{supp}(\Wcheap)$, the conditional distribution $\Wcostly \mid \Wcheap = \wcheap$ is defined as follows: if $\wcheap > 0$, then $\Wcostly \mid \Wcheap = \wcheap$ is a point mass at $f_t(\wcheap)$; if $\wcheap = 0$, then $\Wcostly \mid \Wcheap = \wcheap$ is a point mass at $0$.    
\end{definition}
For example, the distribution takes the following form within Example \ref{example:linear}.
\begin{example}[continues=example:linear]% 
Let $P = 2$, $\gamma = 0.1$, $\alpha = 0.5$, and $|\mathcal{T}| = 1$. Then, $\Wcheap$ and $\Wcostly$ are both distributed as uniform distributions and $\muengagement(P, c, u, \mathcal{T})$ is supported on a line segment (Figure \ref{fig:1type}).
\end{example}

We prove that $\muengagement(P, c, u, \mathcal{T})$ is a symmetric mixed equilibrium.
\begin{theorem}
\label{thm:onetype}
If $|\mathcal{T}| = 1$, the distribution $\muengagement(P, c, u, \mathcal{T})$ is a symmetric mixed equilibrium in the game with $M = \MPlatform$. 
\end{theorem}
\noindent In fact, we further prove that $\muengagement(P, c, u, \mathcal{T})$ is the unique symmetric mixed equilibrium when gaming tricks are costly. 
\begin{theorem}
\label{thm:onetypeunique}
Suppose that $|\mathcal{T}| = 1$ and gaming is costly (i.e. $(\nabla c(w))_2 > 0$ for all $w \in \mathbb{R}_{\ge 0}^2$). Then, if $\mu$ is a symmetric mixed equilibrium in the game with $M = \MPlatform$, it holds that $\mu = \muengagement(P, c, u, \mathcal{T})$. 
\end{theorem}
\noindent The fact that $\muengagement(P, c, u, \mathcal{T})$ is the unique symmetric mixed equilibrium under costly gaming tricks and homogeneous users provides additional justification for our focus on $\muengagement(P, c, u, \mathcal{T})$ in \Cref{sec:performance}. 

We do note that although $\muengagement(P, c, u, \mathcal{T})$ is unique within the class of symmetric equilibrium, there typically do exist asymmetric equilibria. For example, if $P =3$, the mixed strategy profile where $\mu_1$ is a point mass at $[0,0]$ and $\mu_2 = \mu_3 = \muengagement(2, c, u, t)$ is an equilibrium. Extending our analysis and results to asymmetric equilibria is an interesting direction for future work. 

\subsection{Additional assumptions for characterization results for multiple types}\label{appendix:additionalassumptions}

In our characterization results for heterogeneous users, we require the following additional assumptions. One key assumption is the following linearity condition on the \textit{induced cost function} given by the optimization program: 
\begin{equation}
\label{eq:inducedcost}
\CEngagement_t(m) := \min_w c(w) \text{ s.t. } u(w, t) \ge 0, \MPlatform(w) \ge m.    
\end{equation}
which captures the minimum production cost to create content with engagement at least $m$ and nonnegative user utility.
\begin{assumption}[Linearity of cost functions]
\label{assumption:linearity}
We assume that there exists coefficients $a_t > 0$ for $t \in \mathcal{T}$, intercept $b > 0$, and shift parameter $s \ge -1 \cdot \min_{w \in \mathbb{R}^2_{\ge 0}} \MPlatform(w)$ such that:
\begin{enumerate}
    \item The coefficients $a_t$ are strictly decreasing: $a_{t_1} > a_{t_2}$ for all $t_2 > t_1$. 
    \item The induced cost function is a \textit{nonnegative part of a linear function}: that is, $\CEngagement_t(m) = \max(0, a_t (m + s) - 1)$ for all $m \in \mathbb{R}$. 
\end{enumerate}
\end{assumption}
\noindent Assumption \ref{assumption:linearity} guarantees that there is a linear relationship between costs and engagement. Apart from Assumption \ref{assumption:linearity}, we further assume that gaming tricks are costless (that is, $(\nabla(c(w)))_2 = 0$ for all $w \in \mathbb{R}_{\ge 0}^2$) and that $u([0,0], t) \ge 0$ for all $t \in \mathcal{T}$ (i.e. no costly effort is required to meet the user utility constraint for any user).

These assumptions are satisfied by the linear functional forms in \Cref{example:linear} and \Cref{example:KMR} with specific parameter settings. 
\begin{example}[continues=example:linear]
For this setup with $\alpha = 1$ and $\gamma = 0$, the cost function assumptions are satisfied for $a_t = \frac{1}{1 + t}$ and $s = 1$. 
%$b = 1$, 
\end{example}
\begin{example}[continues=example:KMR]
For this setup with $\gamma = 0$, the cost function assumptions are satisfied for $a_t = \frac{1}{1 + t} = \frac{W}{v}$
%$b = 1$,
and $s = 0$.  
\end{example}

\subsection{Characterization for $N$ well-separated types}\label{appendix:characterizationNtypes}

Interestingly, even under the assumptions in \Cref{appendix:additionalassumptions}, the symmetric mixed equilibrium structure is already complex for the case of 2 arbitrary types (as we will show in \Cref{appendix:characterization2types}). Nonetheless, the equilibrium structure turns out to be significantly cleaner under a ``well-separated'' assumption on the types: $a_{t_1} \ge 1.5 a_{t_2}$. This motivates us to restrict to ``well-separated'' types in our analysis of type spaces $\mathcal{T}$ of arbitrary size. The appropriate generalization of the 2-type condition turns out to be:
\[
 a_{t_1} \ge \left(1 + 
 %\frac{1}{Nb} 
\frac{1}{N} 
 \right) a_{t_2} \ge \ldots \ge \left(1 + 
 \frac{1}{N} 
 %\frac{1}{Nb} 
 \right)^{N'-1} a_{N'} > 0.   
\]

As a warmup, let's first consider the case of 2 well-separated types satisfying $a_{t_1} \ge 1.5 a_{t_2}$. The equilibrium is a mixture of $2$ distributions, one for each type. The distribution for type $t_i$ looks similar to the equilibrium distribution in Definition \ref{def:homogeneous} for homogeneous users of type $t_i$, with the modification that there is a factor of $2$ multiplier on the cumulative density function of $\Wcheap$. 
\begin{definition}
\label{def:2typeswellseparated}
Let $\mathcal{T} = \left\{t_1, t_2 \right\}$ be a type space consisting of two types. Furthermore, suppose that gaming tricks are costless (that is, $(\nabla(c(w)))_2 = 0$ for all $w \in \mathbb{R}_{\ge 0}^2$) and suppose that $u([0, 0], t) \ge 0$ for all $t \in \mathcal{T}$. Suppose that Assumption \ref{assumption:linearity} holds with coefficients satisfying $a_{t_1} \ge 1.5 a_{t_2} > 0$. We define the distribution $(\Wcheap, \Wcostly) \sim \muengagement(P, c, u, \mathcal{T}, \MPlatform)$ to be be a mixture of the following 2 distributions $(\Wcheap^1, \Wcostly^1)$ and $(\Wcheap^2, \Wcostly^2)$, where the mixture weights are $0.5$ and $0.5$. Let $f_t$ be defined by \eqref{eq:ft}, and let $\CCheap_t$ be defined by \eqref{eq:ccheap}. The random variables $\Wcheap^1$ and $\Wcheap^2$ are defined by: 
\begin{align*}
 \mathbb{P}[\Wcheap^1 \le \wcheap] &= 
 \min\left(2 \cdot \CCheap_{t_1}(\wcheap),1\right) \\
  \mathbb{P}[\Wcheap^2 \le \wcheap] &= \min\left(4 \cdot \CCheap_{t_2}(\wcheap),1\right),
\end{align*}
and where for $1 \le i \le 2$ and $\wcheap \in \text{supp}(\Wcostly^i)$, the distribution $\Wcostly^i \mid \Wcheap^i = \wcheap$ is a point mass at $f_{t_i}(\wcheap)$. 
\end{definition}
\begin{proposition}
\label{prop:2typeswellseparated}
Let $\mathcal{T} = \left\{t_1, t_2 \right\}$ be a type space consisting of two types. Furthermore, suppose that gaming tricks are costless (that is, $(\nabla(c(w)))_2 = 0$ for all $w \in \mathbb{R}_{\ge 0}^2$) and suppose that $u([0, 0], t) \ge 0$ for all $t \in \mathcal{T}$. Suppose that Assumption \ref{assumption:linearity} holds with coefficients satisfying $a_{t_1} \ge 1.5 a_{t_2} > 0$. Let $\muengagement(P, c, u, \mathcal{T}, \MPlatform)$ be defined according to Definition \ref{def:2typeswellseparated}. Then, $\muengagement(P, c, u, \mathcal{T}, \MPlatform)$ is a symmetric mixed equilibrium in the game with $M = \MPlatform$. 
\end{proposition}

We are now ready to generalize Definition \ref{def:2typeswellseparated} to $N \ge 2$ ``well-separated'' types. The distribution is again $\muengagement(P, c, u, \mathcal{T})$ a mixture of distributions: however, it is surprisingly not a mixture of $N$ distributions, but rather a mixture of $N' \le N$ distributions corresponding to the first $N'$ types $t_1, \ldots, t_{N'}$. The distribution for $t_i$ again looks similar to the equilibrium distribution in Definition \ref{def:homogeneous} for homogeneous users with type $t_i$, but again with a multiplicative rescaling on the cdf of $\Wcheap$. The multiplicative rescaling is $N$ for $N'-1$ out of $N'$ types.  
\begin{definition}
\label{def:Ntypes}
Let $\mathcal{T} = \left\{t_1, \ldots, t_N \right\}$ be a type space consisting of $N$ types, let $P = 2$, 
%let $u$ satisfy (A1)-(A3), let $c$ satisfy (B1)-(B3), and let $\MPlatform$ satisfy (C1)-(C2). 
suppose that gaming tricks are costless (that is, $(\nabla(c(w)))_2 = 0$ for all $w \in \mathbb{R}_{\ge 0}^2$) and suppose that $u([0, 0], t) \ge 0$ for all $t \in \mathcal{T}$. Suppose that Assumption \ref{assumption:linearity} holds with coefficients satisfying
\[
 a_{t_1} \ge \left(1 + 
 %\frac{1}{Nb} 
\frac{1}{N} 
 \right) a_{t_2} \ge \ldots \ge \left(1 + 
 \frac{1}{N} 
 %\frac{1}{Nb} 
 \right)^{N'-1} a_{N'} > 0.   
\]
We define the distribution $(\Wcheap, \Wcostly) \sim \muengagement(P, c, u, \mathcal{T})$ to be a mixture of the following $N'$ distributions $\left\{(\Wcheap^i, \Wcostly^i)\right\}_{1 \le i \le N'}$, where $N' \in \mathbb{Z}_{\ge 1}$ is the minimum number such that $\sum_{i=1}^{N'} \frac{1}{N-i+1} \ge 1$. The mixture weight $\alpha^i$ on $(\Wcheap^i, \Wcostly^i)$ is 
\[\alpha^i :=
\begin{cases}
  \frac{1}{N-i+1} & \text{ if }  1 \le i \le N' -1\\
  1 - \sum_{i'=1}^{N'-1} \frac{1}{N-i'+1} & \text{ if }  i = N'.
\end{cases}
\] The random vectors $(\Wcheap^i, \Wcostly^i)$ are defined as follows. Let $f_t$ be defined by \eqref{eq:ft}, and let $\CCheap_t$ be defined by \eqref{eq:ccheap}. 
The marginal distribution of $\Wcheap^i$ is defined by:
\[
\mathbb{P}[\Wcheap^i \le \wcheap] =
\begin{cases}
 \min\left(N \cdot \CCheap_{t_i}(\wcheap), 1\right) & \text{ if } 1 \le i \le N' - 1 \\
 \min\left(\frac{N}{N-N'+1} \cdot \left(1 - \sum_{j=1}^{N'-1} \frac{1}{N-j+1} \right)^{-1} \cdot \CCheap_{t_i}(\wcheap), 1\right) & \text{ if } i = N'. 
\end{cases}
\]
For each $1 \le i \le N'$ and $\wcheap \in \text{supp}(\Wcheap^i)$, the conditional distribution $\Wcostly^i \mid \Wcheap^i = \wcheap$ is a point mass at $f_{t_i}(\wcheap)$. 

\end{definition}
For the cases of $N$ well-separated types, we show that $\muengagement(P, c, u, \mathcal{T})$ is a symmetric mixed Nash equilibrium. 
\begin{theorem}
\label{thm:Ntypes}    
Let $\mathcal{T} = \left\{t_1, \ldots, t_N \right\}$ be a type space consisting of $N$ types, let $P = 2$, suppose that gaming tricks are costless (that is, $(\nabla(c(w)))_2 = 0$ for all $w \in \mathbb{R}_{\ge 0}^2$), and suppose that $u([0, 0], t) \ge 0$ for all $t \in \mathcal{T}$. Suppose that Assumption \ref{assumption:linearity} holds with coefficients satisfying
\begin{equation}
\label{eq:wellseparated}
 a_{t_1} \ge \left(1 + \frac{1}{N} \right) a_{t_2} \ge \ldots \ge \left(1 + \frac{1}{N} \right)^{N'-1} a_{N'} > 0.   
\end{equation}
Let $\muengagement(P, c, u, \mathcal{T})$ be defined according to Definition \ref{def:Ntypes}. Then, $\muengagement(P, c, u, \mathcal{T})$ is a symmetric mixed Nash equilibrium in the game with $M = \MPlatform$. 
\end{theorem}

\subsection{Characterization for 2 types}\label{appendix:characterization2types}

\begin{figure}[h]
    \centering
    \begin{subfigure}{.45\textwidth}
        \centering

            \begin{tikzpicture}
    % Define the points
    \coordinate (x) at (0,1); % top line left endpoint
    \coordinate (y) at (1,1); % top line right endpoint
    \coordinate (a) at (2.5,0); % bottom line left endpoint
    \coordinate (b) at (3.5,0); % bottom line right endpoint
  \fill[red] (x) circle (2pt);
            \fill[red] (y) circle (2pt);
            \fill[blue] (a) circle (2pt);
            \fill[blue] (b) circle (2pt);
    % Draw the lines with colors
    \draw[red, thick] (x) -- (y) node[midway,above=10pt,font=\small, text=red] {supp$(\V | T = t_1)$};
    \draw[blue, thick] (a) -- (b) node[midway,below=10pt,font=\small, text=blue] {supp$(\V | T = t_2)$};

    % Label the points
    \node[left][text=red] at (x) {$\frac{1}{a_{t_1}}$};
    \node[right][text=red] at (y) {$\frac{3}{2 \cdot a_{t_1}}$};
    \node[left][text=blue] at (a) {$\frac{1}{a_{t_2}}$};
    \node[right][text=blue] at (b) {$\frac{5}{4 \cdot a_{t_1}}$};
\end{tikzpicture}
        \caption{Case 1: $a_{t_1} / a_{t_2} \ge 1.5$}
        \label{fig:sub1}
    \end{subfigure}%
        \hfill 
    \begin{subfigure}{.45\textwidth}
        \centering
                \begin{tikzpicture}
    % Define the points
    \coordinate (x) at (0,1); % top line left endpoint
    \coordinate (y) at (2,1); % top line right endpoint
    \coordinate (a) at (1,0); % bottom line left endpoint
    \coordinate (b) at (3,0); % bottom line right endpoint
  \fill[red] (x) circle (2pt);
            \fill[red] (y) circle (2pt);
            \fill[blue] (a) circle (2pt);
            \fill[blue] (b) circle (2pt);
    % Draw the lines with colors
    \draw[red, thick] (x) -- (y) node[midway,above=10pt,font=\small, text=red] {supp$(\V | T = t_1)$};
    \draw[blue, thick] (a) -- (b) node[midway,below=10pt,font=\small, text=blue] {supp$(\V | T = t_2)$};

    % Label the points
    \node[left][text=red] at (x) {$\frac{1}{a_{t_1}}$};
    \node[right][text=red] at (y) {$\frac{1}{2 a_{t_2} \cdot \left( \frac{a_{t_1}}{a_{t_2}} - 1 \right)}$};
    \node[left][text=blue] at (a) {$\frac{1}{a_{t_2}}$};
    \node[right][text=blue] at (b) {$\frac{1}{a_{t_2}} \cdot \left( 2 - \frac{a_{t_1}}{2 a_{t_2}} \right)$};
\end{tikzpicture}
        \caption{Case 2: $(5 - \sqrt{5})/2 \le a_{t_1} / a_{t_2} \le 1.5$}
        \label{fig:sub2}
    \end{subfigure}%%%
    \vspace{1cm} % Add some vertical space between rows

    \begin{subfigure}{.9\textwidth}
        \centering
            \begin{tikzpicture}
    % Define the points
    \coordinate (x) at (0,1); % top line left endpoint
    \coordinate (y) at (4,1); % top line right endpoint
    \coordinate (a) at (1.5,0); % bottom line left endpoint
    \coordinate (b) at (3,0); % bottom line right endpoint
  \fill[red] (x) circle (2pt);
            \fill[red] (y) circle (2pt);
            \fill[blue] (a) circle (2pt);
            \fill[blue] (b) circle (2pt);
    % Draw the lines with colors
    \draw[red, thick] (x) -- (y) node[midway,above=10pt,font=\small, text=red] {supp$(\V | T = t_1)$};
    \draw[blue, thick] (a) -- (b) node[midway,below=10pt,font=\small, text=blue] {supp$(\V | T = t_2)$};

    % Label the points
    \node[left][text=red] at (x) {$\frac{1}{a_{t_1}}$};
    \node[right][text=red] at (y) {$\frac{1}{a_{t_1}} + \left(\frac{1}{a_{t_1}} - \frac{1}{2 a_{t_2}} \right) \left(\frac{3 - \frac{a_{t_1}}{a_{t_2}}}{2 - \frac{a_{t_1}}{a_{t_2}}} \right)$};
    \node[left][text=blue] at (a) {$\frac{1}{a_{t_2}}$};
    \node[right][text=blue] at (b) {$\frac{3 - \frac{a_{t_1}}{a_{t_2}}}{2 a_{t_2} \left(2 - \frac{a_{t_1}}{a_{t_2}} \right)}$};
\end{tikzpicture}
        \caption{Case 3: $1 < a_{t_1} / a_{t_2} \le (5 - \sqrt{5})/2$}
        \label{fig:sub3}
    \end{subfigure}
    \caption{The support of $(\V,T)$ in Definition \ref{def:2types} for different values of $a_{t_1} / a_{t_2}$. The red  line shows the support of $\V \mid T = t_1$, and the blue line shows the support of $\V \mid T = t_2$. If $a_{t_1}$ and $a_{t_2}$ are sufficiently far apart (Case 1), then the supports are disjoint. When $a_{t_1}$ and $a_{t_2}$ become closer (Case 2), the supports start to overlap, and when $a_{t_1}$ and $a_{t_2}$ are sufficiently close (Case 3), the support of $\V \mid T = t_2$ is contained in the support of $\V \mid T = t_1$.}
    \label{fig:MT2types}
\end{figure}
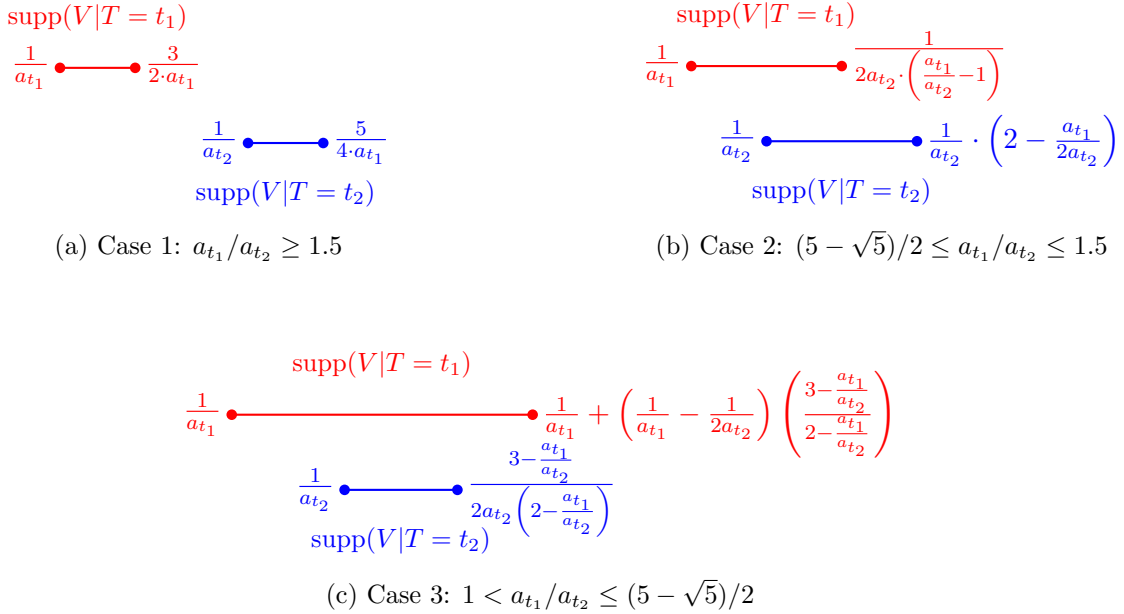

For the case of 2 arbitrary types, it is cleaner to work in the following reparametrized space $S := \left\{(\MPlatform([\wcostly, \wcheap]) - s, t) \mid t \in \mathcal{T}, u([\wcostly, \wcheap], t) = 0 \right\} \subseteq \mathbb{R} \times \left\{t_1, t_2 \right\}$ than directly over the content space $\mathbb{R}_{\ge 0}$. We map each $(v, t) \in S$ to the unique content $h(v, t) \in \mathbb{R}_{\ge 0}^2$ of the form $h(v, t) = [f_{t}(\wcheap), \wcheap]$ such that $\MPlatform([f_{t}(\wcheap), \wcheap]) = v - s$. Conceptually, $h(v,t)$ captures content with engagement $v - s$ optimized for winning type $t$. In our characterization, rather than define directly a distribution over content, we instead define a random vector $(\V, T)$ over $S$, which corresponds a distribution $W$ over content defined so $W \mid (\V, T) = (v,t)$ is a point mass at $h(v, t)$. 

We split our characterization into three cases depending on the relationship between $\text{supp}(V \mid T = t_1)$ and $\text{supp}(V \mid T = t_2)$ (see Figure \ref{fig:MT2types}). When types are well-separated, it turns out that $\text{supp}(V \mid T = t_1)$ and $\text{supp}(V \mid T = t_2)$. When types are closer together, the supports are two overlapping line segments, and when types are very close together, the support $\text{supp}(V \mid T = t_2)$ is contained in the support $\text{supp}(V \mid T = t_1)$ and $\text{supp}(V \mid T = t_2)$. We formally define the characterization as follows. 
\begin{definition}
\label{def:2types}
Let $\mathcal{T} = \left\{t_1, t_2 \right\}$ be a type space consisting of two types, let $P = 2$, suppose that gaming tricks are costless (that is, $(\nabla(c(w)))_2 = 0$ for all $w \in \mathbb{R}_{\ge 0}^2$), and suppose that $u([0, 0], t) \ge 0$ for all $t \in \mathcal{T}$. Suppose that Assumption \ref{assumption:linearity} holds with parameters $b$, $s$, and $a_{t_1} > a_{t_2} > 0$. We define the distribution $W \sim \muengagement(P, c, u, \mathcal{T})$ as follows. Let $f_t$ be defined by \eqref{eq:ft}. Below we define a random vector $(\V, T)$ over $\mathbb{R} \times \left\{t_1, t_2 \right\}$; the distribution 
$W \mid (\V, T) = (v,t)$ is a point mass at $W = [f_{t}(\wcheap), \wcheap]$ such that $\MPlatform([f_{t}(\wcheap), \wcheap]) = v - s$.

\textbf{Case 1 ($a_{t_1} / a_{t_2} \ge 1.5$):} 
We define the random vector $(\V,T)$ so 
where $\V$ has density $g$ defined to be:  
\[
g(v) :=
\begin{cases}
0 & \text{ if } v \le \frac{1}{a_{t_1}} \\
a_{t_1} & \text{ if } \frac{1}{a_{t_1}} \le v \le \frac{3}{2 a_{t_1}} \\
2 a_{t_2} & \text{ if } \frac{1}{a_{t_2}} \le v \le \frac{5}{4 a_{t_2}} \\ 
0 & \text{ if } v \ge \frac{5}{4 a_{t_2}} \\ 
\end{cases}
\]
and where $T \mid \V = v$ is distributed according to 
\[ 
\begin{cases}
   \mathbb{P}[T = t_1 \mid \V =v] = 1 & \text{ if } v \in \left[\frac{1}{a_{t_1}},  \frac{3}{2 a_{t_1}}\right] \\
   \mathbb{P}[T = t_1 \mid \V =v] = 0 & \text{ if } v \in \left[\frac{1}{a_{t_2}} ,  \frac{5}{4 a_{t_2}} \right]
\end{cases}
\]

\textbf{Case 2 ($(5 - \sqrt{5})/2 \le a_{t_1} / a_{t_2} \le 1.5$)} 
We define the random vector $(\V,T)$ so 
where $\V$ has density $g$ defined to be:  
\[
g(v) :=
\begin{cases}
0 & \text{ if } v\le \frac{1}{a_{t_1}} \\
a_{t_1} & \text{ if } \frac{1}{a_{t_1}} \le v \le \frac{1}{a_{t_2}} \\
2 a_{t_2} & \text{ if } \frac{1}{a_{t_2}} \le v \le \frac{1}{a_{t_2}} \left( 2 - \frac{a_{t_1}}{2 a_{t_2}} \right)\\ 
0 & \text{ if } v \ge \frac{1}{a_{t_2}} \left( 2 - \frac{a_{t_1}}{2 a_{t_2}} \right) \\ 
\end{cases}
\]
and where $T \mid \V = v$ is distributed according to 
\[ 
\begin{cases}
   \mathbb{P}[T = t_1 \mid \V =v] = 1 & \text{ if } v \in \left[\frac{1}{a_{t_1}}, \frac{1}{a_{t_2}} \right] \\
  \mathbb{P}[T = t_1 \mid \V = v] = \frac{a_{t_1}}{a_{t_2}} - 1  & \text{ if } v \in \left[\frac{1}{a_{t_2}}, \frac{1}{2 a_{t_2} \cdot \left(\frac{a_{t_1}}{a_{t_2}} - 1\right)} \right] \\
   \mathbb{P}[T = t_1 \mid \V =v] = 0 & \text{ if } v \in \left[\frac{1}{2 a_{t_2} \cdot \left(\frac{a_{t_1}}{a_{t_2}} - 1\right)} , \frac{1}{a_{t_2}} \left( 2 - \frac{a_{t_1}}{2 a_{t_2}} \right) \right]
\end{cases}
\]
 
\textbf{Case 3 ($1 < a_{t_1} / a_{t_2} \le (5 - \sqrt{5})/2$)} 
We define the random vector $(\V, T)$ so where $\V$ has density $g$ defined to be:  
\[
g(v) :=
\begin{cases}
0 & \text{ if } v \le \frac{1}{a_{t_1}} \\
a_{t_1} & \text{ if } \frac{1}{a_{t_1}} \le v \le \frac{1}{a_{t_2}} \\
2 a_{t_2} & \text{ if }  \frac{1}{a_{t_2}} \le v \le \frac{3 - \frac{a_{t_1}}{a_{t_2}}}{2 a_{t_2} \left(2 - \frac{a_{t_1}}{a_{t_2}} \right)} \\ 
a_{t_1} & \text{ if } \frac{3 - \frac{a_{t_1}}{a_{t_2}}}{2 a_{t_2} \left(2 - \frac{a_{t_1}}{a_{t_2}} \right)}  \le v \le \frac{1}{a_{t_1}} + \left(\frac{1}{a_{t_1}} - \frac{1}{2 a_{t_2}} \right) \left(\frac{3 - \frac{a_{t_1}}{a_{t_2}}}{2 - \frac{a_{t_1}}{a_{t_2}}} \right)  \\ 
0 & \text{ if } v \ge \frac{1}{a_{t_1}} + \left(\frac{1}{a_{t_1}} - \frac{1}{2 a_{t_2}} \right) \left(\frac{3 - \frac{a_{t_1}}{a_{t_2}}}{2 - \frac{a_{t_1}}{a_{t_2}}} \right). 
\end{cases}
\]
and where $T \mid \V = v$ is distributed according to 
\[ 
\begin{cases}
  \mathbb{P}[T = t_1 \mid \V = v] = 1 
  %\text{ and }  \mathbb{P}[T = t_2 \mid M =m] = 0
  & \text{ if } v \in \left[\frac{1}{a_{t_1}}, \frac{1}{a_{t_2}} \right] \\
  \mathbb{P}[T = t_1 \mid \V = v] = \frac{a_{t_1}}{a_{t_2}} - 1 
  & \text{ if } v \in \left[\frac{1}{a_{t_2}}, \frac{3 - \frac{a_{t_1}}{a_{t_2}}}{2 a_{t_2} \left(2 - \frac{a_{t_1}}{a_{t_2}} \right)}\right] \\
  \mathbb{P}[T = t_1 \mid \V = v] = 1 
  %\text{ and }  \mathbb{P}[T = t_2 \mid M =m] = 0
  & \text{ if } v \in \left[\frac{3 - \frac{a_{t_1}}{a_{t_2}}}{2 a_{t_2} \left(2 - \frac{a_{t_1}}{a_{t_2}} \right)}, \frac{1}{a_{t_1}} + \left(\frac{1}{a_{t_1}} - \frac{1}{2 a_{t_2}} \right) \left(\frac{3 - \frac{a_{t_1}}{a_{t_2}}}{2 - \frac{a_{t_1}}{a_{t_2}}} \right) \right]
\end{cases}
\]
\end{definition}

For the case of 2 types, we show that $\muengagement(P, c, u, \mathcal{T}, \MPlatform)$ is a symmetric mixed equilibrium.
\begin{theorem}
\label{thm:2types}    
Let $\mathcal{T} = \left\{t_1, t_2 \right\}$ be a type space with 2 types, let $P = 2$, 
suppose that gaming tricks are costless (that is, $(\nabla(c(w)))_2 = 0$ for all $w \in \mathbb{R}_{\ge 0}^2$) and suppose that $u([0, 0], t) \ge 0$ for all $t \in \mathcal{T}$. Suppose that Assumption \ref{assumption:linearity} holds with coefficients $a_{t_1} > a_{t_2} > 0$. Let  $\muengagement(P, c, u, \mathcal{T})$ be defined according to Definition \ref{def:2types}. Then  $\muengagement(P, c, u, \mathcal{T})$ is a symmetric mixed equilibrium in the game with $M = \MPlatform$. 
\end{theorem}

\section{Discussion}

In this work, we study content creator competition for engagement-based recommendations that reward both quality and gaming tricks (e.g. clickbait). Our model further captures that a user only tolerates gaming tricks in sufficiently high-quality content, which also shapes content creator incentives. Our first result (Theorem \ref{thm:positivecorrelationhomogeneous}) suggests that gaming and quality are complements for the content creators, which we empirically validate on a Twitter dataset. We then analyze the downstream performance of engagement-based optimization at equilibrium. We show that 
higher gaming costs can lead to lower average consumption of quality (Theorem \ref{thm:comparisonuserconsumptiongaming}), engagement-based optimization can be suboptimal even in terms of (realized) engagement (Theorem \ref{thm:comparisonengagement}), the user welfare of engagement-based optimization can fall below that of random recommendations (Theorem \ref{thm:comparisonuserutility}).

More broadly, our results illustrate how content creator incentives can influence the downstream impact of a content recommender system, which poses challenges when evaluating a platform's metric. In particular, there is a disconnect between how a platform's engagement metric behaves on a fixed content landscape and how the same metric behaves on an endogeneous content landscape shaped by the metric. Interestingly, this disconnect manifests in both performance measures relevant to the platform and performance measures relevant to society as a whole. We hope that our work encourages future evaluations of recommendation policies--- both of platform metrics and societal impacts---to carefully account for content creator incentives. 

\section{Acknowledgments}

We thank Nikhil Garg and Smitha Milli for useful comments on this paper. 

\bibliographystyle{plainnat}

\bibliography{bib.bib}

\newpage 
\appendix

\section{Auxiliary definitions and lemmas}\label{appendix:auxlemmas}

In our analysis of equilibria, it will be helpful to work with several quantities. We first define $\Ct$ to be the set of content that achieves $0$ utility for that type. That is:
\begin{equation}
\label{eq:Ct}
    \mathcal{C}_t := \left\{[\wcostly, \wcheap] \mid u([\wcostly, \wcheap], t) = 0 \right\}.
\end{equation}
We also define an augmented version of these sets that also includes content with $\wcostly = 0$ that achieving positive utility. 
\begin{equation}
\label{eq:caugt}
  \Caugt := \left\{[\wcostly, \wcheap] \mid u([\wcostly, \wcheap], t) = 0 \right\} \cup \left\{[0, \wcheap] \mid \wcheap \in [0, \min_{w' \mid u(w', t) = 0} \wcheap')] \right\},   
\end{equation}

The set $\Caugt$ turns out to be closely related to the function $f_t$ defined in \eqref{eq:ft}. 
\begin{lemma}
\label{lemma:one-to-one}
The set $\Caugt$ can be written as: 
\[\Caugt = \left\{ [f_t(\wcheap), \wcheap] \mid \wcheap \ge 0 \right\} \]
where $f_t$ is defined by \eqref{eq:ft}. 
\end{lemma}
\begin{proof}
First, we show that $\Caugt \subseteq \left\{ (f_t(\wcheap), \wcheap) \mid \wcheap \ge 0 \right\}$. If $w \in \Caugt$, then either $u(w, t) = 0$ or $\wcostly = 0$ and $\wcheap \in [0, \min_{w' \mid u(w', t) = 0} \wcheap')]$. If $u(w, t) = 0$, since investing in quality is costly, it must hold that $\wcostly = f_t(\wcheap)$. Next, suppose that  $\wcheap \in [0, \min_{w' \mid u(w', t) = 0} \wcheap')]$. We observe that $\min_{w' \mid u(w', t) = 0} \wcheap'$ is the unique value of $\wcheap'$ such that $u([0, \wcheap'], t) = 0$. This implies that $u([0, \wcheap], t) \ge 0$, so $f_t(\wcheap) = 0$ as desired. 

Next, we show that $\left\{ [f_t(\wcheap), \wcheap] \mid \wcheap \ge 0 \right\} \subseteq \Caugt$. Let $w = [f_t(\wcheap), \wcheap]$ for some $\wcheap \ge 0$. If $u(w, t) = 0$, then $w \in \Caugt$ as desired. If $u(w,t) > 0$, then it must hold that $\wcostly = 0$ (otherwise, it would be possible to lower $\wcostly$ while keeping utility nonnegative, which contradicts the fact that $f_t(\wcheap) = \wcostly$), so $w \in \Caugt$.  
\end{proof}
We prove that the function $f_t$ is weakly increasing. 
\begin{lemma}
\label{lemma:increasing}
The function $f_t$ as defined in \eqref{eq:ft} is weakly increasing. Moreover, the function $\MPlatform([f_t(\wcheap), \wcheap])$ is strictly increasing in $\wcheap$. 
\end{lemma}
\begin{proof}
Suppose that $\wcheap^1 \ge \wcheap^2$. We claim that $f_t(\wcheap^1) \ge f_t(\wcheap^2)$. To see this, note that 
\[u([f_t(\wcheap^1), \wcheap^2], t) > u([f_t(\wcheap^1), \wcheap^1], t) \ge 0,\]
which proves the first statement.

To see that $\MPlatform([f_t(\wcheap), \wcheap])$ is increasing, note that $f_t$ is a weakly increasing function (see Lemma \ref{lemma:increasing}) and that $\MPlatform$ is strictly increasing in both of its arguments. 
\end{proof}

We next show the following properties of the optima of \eqref{eq:inducedcost}. 
\begin{lemma}
\label{lemma:uniqueness}
The optimization program $\inf_{w \in \mathbb{R}_{\ge 0}^2} c(w) \text{ s.t. } u(w, t) \ge 0, \MPlatform(w) \ge m$ satisfies the following properties:
\begin{enumerate}[leftmargin=*]
\item For any $m \in \left\{\MPlatform(w) \mid w \in \Caugt \right\}$, the optimization program is feasible and any optimum $w^*$ satisfies $w^* \in \Caugt$. 
\item If $m \in \left\{\MPlatform(w) \mid w \in \Caugt \right\}$ and $\CEngagement_t(m) > 0$, the optimization program has a unique optimum $w^*$ and moreover $\MPlatform(w^*) = m$. 
\end{enumerate}
\end{lemma}
\begin{proof}
Suppose that $m \in \left\{\MPlatform(w) \mid w \in \Caugt \right\}$. 

First, we show that the optimization program is feasible. Suppose that $w$ is such that $\MPlatform(w) = m$. Using the fact that $u([\wcostly', \wcheap], t)$ approaches $\infty$ as $\wcostly' \rightarrow \infty$, we see that there exists $\wcostly' \ge \wcostly$ such that $\MPlatform([\wcostly', \wcheap]) \ge \MPlatform(w) = m$ and $u([\wcostly', \wcheap], t) \ge 0$, as desired. 

Next, we show that there exists $w \in \mathbb{R}_{\ge 0}^2$ such that $u(w,t) \ge 0$, $\MPlatform(w) \ge m$, and $c(w) = \CEngagement_t(m)$. To make the domain compact, observe that there exists $w' \in \mathbb{R}_{\ge 0}^2$ such that $\MPlatform(w') = m$ by assumption, which means that $\CEngagement_t(m) \le c(w')$. The set
\begin{align*}
 &\left\{w \in \mathbb{R}_{\ge 0}^2 \mid c(w) \le c(w'), u(w, t) \ge 0, \MPlatform(w) \ge m \right\} \\
 &= \left\{w \in \mathbb{R}_{\ge 0}^2 \mid \MPlatform(w) \ge m \right\} \cap \left\{w \in \mathbb{R}_{\ge 0}^2 \mid u(w,t) \ge 0 \right\} \cap c^{-1} \left([0, c(w')]\right) .
\end{align*}
The first two terms are closed, and the last term is compact (because the preimage of a continuous function of a compact set is compact). This means that the intersection is compact. Now, we use the fact that the $\inf$ of a continuous function over compact set is achievable. 

Let $w^*$ be an optima. We show the following two properties: 
\begin{enumerate}[label=(P\arabic*)]
    \item If $\wcostly^*, \wcheap^* > 0$, then $\MPlatform(w^*) = m$. 
    \item If $\wcostly^* > 0$, then $u(w^*, t) = 0$. 
\end{enumerate}

First, we show (P1). Assume for sake of contradiction that $\MPlatform(w) > m$. Let $d$ be the direction normal to $\nabla u(w)$ where the costly coordinate is negative and the cheap coordinate is negative. 
We see that 
\begin{align*}
 \langle d, \nabla \MPlatform(w^*) \rangle &= -|d_1| (\nabla \MPlatform(w^*))_1 - |d_2| (\nabla \MPlatform(w^*))_2 < 0 \\
  \langle d, \nabla c(w^*) \rangle &= -|d_1| (\nabla c(w^*))_1 - |d_2| (\nabla c(w^*))_2 < 0  \\
  \langle d, \nabla u(w^*) \rangle &= 0.
\end{align*}
This proves there exists $\epsilon > 0$ such that $w' = w + \epsilon d$ satisfies $\MPlatform(w') \ge m$, $u(w', t) \ge 0$, and $c(w') < c(w^*)$, which is a contradiction. 

Next, we show (P2). Assume for sake of contradiction that $u(w^*, t) > 0$. Let $d$ be the normal direction to $\nabla \MPlatform(w^*)$ where the costly coordinate is negative and the cheap coordinate is positive. We see that 
\begin{align*}
 \langle d, \nabla u(w^*, t) \rangle &= -|d_1| (\nabla u(w))_1 + d_2 (\nabla u(w))_2 < 0 \\
  \langle d, \nabla \MPlatform(w^*) \rangle &= 0,
\end{align*}
Moreover, we can see that $\langle d, \nabla c(w^*) \rangle = -|d_1| (\nabla c(w))_1 + d_2 (\nabla c(w))_2 < 0$, since this can be written as:
\[\frac{(\nabla c(w))_1}{(\nabla c(w))_2} > \frac{|d_2|}{|d_1|} = \frac{(\nabla \MPlatform(w))_1}{(\nabla \MPlatform(w))_2},\]
which holds by assumption. This proves there exists $\epsilon > 0$ such that $w' = w + \epsilon d$ satisfies $\MPlatform(w') \ge m$, $u(w', t) \ge 0$, and $c(w') < c(w^*)$, which is a contradiction. 

We now show that $w^* \in \Caugt$. First, suppose that $\wcostly^* = 0$. Then, using the fact that $u(w^*, t) \ge 0$, we see that $f_t(\wcheap^*) = 0 = \wcostly*$, so by Lemma \ref{lemma:one-to-one}, $w^* \in \Caugt$. Next, suppose that $\wcostly^* > 0$. Then we see that $u(w^*, t) = 0$ by (P2), so $w^* \in \Caugt$. 

For the remainder of the analysis, we assume that $c(w^*) = \CEngagement_t(m) > 0$. 

If gaming is costless ($(\nabla(c(w)))_2 = 0$ for all $w$) and $c(w^*) > 0$, then it must hold that $\wcostly^* > 0$. This implies that $u(w^*, t) = 0$. This means that there is a unique value $w \in \Caugt$ such that $c(w) = \CEngagement_t(m)$, so this implies that $w^*$ is the unique optima. If $\wcheap^* > 0$, then we can apply (P1) to see that $\MPlatform(w^*) = m$. If $\wcheap^* = 0$, the fact that $[0, \wcostly^*] \in \Caugt$ implies that $\MPlatform(w^*) = \inf_{w \in \Caugt} \MPlatform(w)$. By the assumption that $m \in \left\{\MPlatform(w) \mid w \in \Caugt \right\}$, this means that $m = \MPlatform(w^*)$ as desired.

If gaming is costly ($(\nabla(c(w)))_2 = 0$ for all $w$) and $\CEngagement_t(m) > 0$, then there is a unique value $w \in \Caugt$ such that $c(w) = c(w^*)$, which shows there is a unique optima. If $\wcheap^* > 0$ and $\wcostly^* > 0$, then (P1) implies that $m = \MPlatform(w^*)$. If $\wcheap^* = 0$, then the fact that $[0, \wcostly^*] \in \Caugt$ implies that $\MPlatform(w^*) = \inf_{w \in \Caugt} \MPlatform(w)$. Finally, suppose that $\wcostly^* = 0$. Assume for sake of contradiction that 
$\MPlatform([0, \wcheap^*]) > m$. Then there exists $\wcheap < \wcheap^*$ such that $\MPlatform([0, \wcheap]) \ge m$, $c([0, \wcheap']) < c([0, \wcheap^*])$, and $u([0, \wcheap'], t) \ge u([0, \wcheap^*], t)$, which would mean that $w^*$ is not an optima, which is a contradiction. 
\end{proof}

We next prove the following properties of the equilibrium characterizations for Example \ref{example:linear} in the case of homogeneous users. First, we analyze the marginal distribution of quality of the symmetric mixed equilibrium for engagement-based optimization. 
\begin{proposition}
\label{prop:equilibriumhomogenousengagement}

% - wcheap / t + alpha = 0
% wcheap = t * alpha 
Consider \Cref{example:linear} with sufficiently high baseline utility $\alpha > -1$, bounded gaming costs $\gamma \in [0, 1)$, and homogeneous users ($\mathcal{T} = \left\{t\right\}$). LThe distribution $(\Wcostly, \Wcheap) \sim \muengagement(P, c, u, \mathcal{T})$ (where $\muengagement(P, c, u, \mathcal{T})$ is specified as in Theorem \ref{thm:investmentbased}) satisfies:
\[\mathbb{P}[\Wcostly \le \wcostly] = 
\begin{cases}
(-\alpha)^{1/(P-1)} & \text{ if }  0 \le \wcostly \le -\alpha \\
  \left(\min(1, \wcostly + \gamma \cdot t \cdot (\wcostly + \alpha))\right)^{1/(P-1)} & \text{ if } \wcostly \ge \max(0, -\alpha).
\end{cases}
\]
\end{proposition}
\begin{proof}
Let $\beta_t = \min \left\{ \wcostly \mid u([\wcostly, 0]) \ge 0 \right\}$ be the minimum investment level. We apply the equilibrium characterization in \Cref{thm:onetype}. We split into two cases: (1) $\beta_t > 0$ and (2) $\beta_t = 0$.

\paragraph{Case 1: $\beta_t > 0$.} The minimum-investment function $f_t$ is strictly increasing so $f_t^{-1}$ is well-defined. Using the the equilibrium characterization in \Cref{thm:onetype}, we observe that: 
\[\mathbb{P}[\Wcostly \le \wcostly] = 
\begin{cases}
 \left(\min(1, c([\beta_t, 0]))\right)^{1/(P-1)} & \text{ if }  0 \le \wcostly \le \beta_t \\
  \left(\min(1, c([\wcostly, f_t^{-1}(\wcostly)]))\right)^{1/(P-1)} & \text{ if } \wcostly \ge \beta_t.
\end{cases}
\]  
Now, using the specification in \Cref{example:linear}, where $c([\wcostly, \wcheap]) = \wcostly + \gamma \cdot \wcheap$, $u(w, t) = \wcostly - (\wcheap / t) + \alpha$, and $f_t(\wcheap) = \max(0, (\wcheap / t) - \alpha)$, we can simplify this expression. In particular, we see that $\beta_t = \max(0, -\alpha) = -\alpha \le 1$ (since $\alpha > -1$ by assumption). Moreover, $f_t^{-1}(\wcostly) = t \cdot (\wcostly + \alpha)$ and $c([\wcostly, f_t^{-1}(\wcostly)]) = \wcostly + \gamma \cdot t \cdot (\wcostly + \alpha)$. Together, this yields the desired expression for this case. 

\paragraph{Case 1: $\beta_t = 0$.} Even though the minimum-investment function $f_t$ is no longer strictly increasing in general, it is strictly increasing on a restricted interval. Let $\delta_t = \inf \left\{ \wcheap \mid \wcheap \ge 0,  f_t(\wcheap) > 0\right\}$ be the minimum value such that strictly positive quality is required to maintain nonnegative utility. We see that $f_t(\wcheap)$ is strictly increasing for $\wcheap > \delta_t$. This means that for $\wcostly > 0$, the inverse $f_t^{-1}$ exists. Using the the equilibrium characterization in \Cref{thm:onetype}, we observe that: 
\[\mathbb{P}[\Wcostly \le \wcostly] = 
\begin{cases}
 \left(\min(1, C_t(\delta_t)) \right)^{1/(P-1)} & \text{ if }  \wcostly = 0 \\
  \left(\min(1, c([\wcostly, f_t^{-1}(\wcostly)]))\right)^{1/(P-1)} & \text{ if } \wcostly \ge \delta_t.
\end{cases}
\] 
Now, using the specification in \Cref{example:linear}, where $c([\wcostly, \wcheap]) = \wcostly + \gamma \cdot \wcheap$, $u(w, t) = \wcostly - (\wcheap / t) + \alpha$, and $f_t(\wcheap) = \max(0, (\wcheap / t) - \alpha)$, we can simplify this expression. In particular, we see that $\delta_t = t \cdot \alpha$ and $C_t(\delta_t) = t \cdot \alpha \cdot \gamma$. Moreover, as above, we see that $f_t^{-1}(\wcostly) = t \cdot (\wcostly + \alpha)$ and $c([\wcostly, f_t^{-1}(\wcostly)]) = \wcostly + \gamma \cdot t \cdot (\wcostly + \alpha)$. 
Together, this yields the desired expression. 
\end{proof}
Next, we analyze the marginal distribution of quality of the symmetric mixed equilibrium for investment-based optimization. 
\begin{proposition}
\label{prop:equilibriumhomogenousinvestment} 
Consider \Cref{example:linear} with bounded gaming costs $\gamma \in [0, 1)$ and sufficiently high baseline utility $\alpha > -1$. Furthermore, suppose that either (a) users are homogeneous ($\mathcal{T} = \left\{t\right\}$), or (b) the baseline utility satisfies $\alpha \ge 0$. The distribution $(\Wcostly, \Wcheap) \sim \muideal(P, c, u, \mathcal{T})$ (where $\muideal(P, c, u, \mathcal{T})$ is specified as in Theorem \ref{thm:investmentbased}) satisfies:
\[\mathbb{P}[\Wcostly \le \wcostly] = 
\begin{cases}
(-\alpha)^{1/(P-1)} & \text{ if }  0 \le \wcostly \le -\alpha \\
  \left(\min(1, \wcostly)\right)^{1/(P-1)} & \text{ if } \wcostly \ge \max(0, -\alpha).
\end{cases}
\]
\end{proposition}
\begin{proof}
Let $\beta_t = \min \left\{ \wcostly \mid u([\wcostly, 0]) \ge 0 \right\}$ be the minimum investment level. We apply the equilibrium characterization in \Cref{thm:investmentbased}. We observe that: 
\[\mathbb{P}[\Wcostly \le \wcostly] = 
\begin{cases}
 \left(\min(1, c([\beta_t, 0]))\right)^{1/(P-1)} & \text{ if }  0 \le \wcostly \le \beta_t \\
  \left(\min(1, c([\wcostly, 0]))\right)^{1/(P-1)} & \text{ if } \wcostly \ge \beta_t.
\end{cases}
\]
Now, using the specification in \Cref{example:linear}, where $c([\wcostly, 0]) = \wcostly$ and $u(w, t) = \wcostly - (\wcheap / t) + \alpha$, we can simplify these expressions. In particular, we see that $\beta_t = \max(0, -\alpha) \le 1$ (since $\alpha > -1$ by assumption). Together, this yields the desired expression. 
\end{proof}

Finally, we analyze the marginal distribution over $T$ of $(\V, T)$ in \Cref{def:2types} for Cases 2-3. 
\begin{lemma}
\label{lemma:property2types}
Consider the setup of \Cref{def:2types}. If $1 \le a_{t_1} / a_{t_2} \le 1.5$, then it holds that:
\begin{align*}
\mathbb{P}[T = t_1]  &= 2 - \frac{a_{t_1}}{\cdot a_{t_2}} \\
\mathbb{P}[T = t_2] &= \frac{a_{t_1}}{\cdot a_{t_2}} - 1
\end{align*}
\end{lemma}
\begin{proof}
We separately analyze Case 2 and Case 3 in \Cref{def:2types}.

\paragraph{Case 2: $(5 - \sqrt{5}) / 2 \le a_{t_1} / a_{t_2} \le 1.5$.} We observe that: 
\begin{align*}
\mathbb{P}[T = t_2] 
&= 2 a_{t_2} \cdot \left( \frac{1}{2 a_{t_2} \left(\frac{a_{t_1}}{a_{t_2}} - 1 \right)} - \frac{1}{a_{t_2}} \right) \cdot  \left(2 - \frac{a_{t_1}}{a_{t_2}}  \right) \\
&+  2 a_{t_2} \cdot \left(\frac{1}{a_{t_2}} \cdot \left(2 - \frac{a_{t_1}}{2 \cdot a_{t_2}}\right) -  \frac{1}{2 a_{t_2} \left(\frac{a_{t_1}}{a_{t_2}} - 1 \right)}\right) \\
&= \left( \frac{1}{\left(\frac{a_{t_1}}{a_{t_2}} - 1 \right)} - 2 \right) \cdot  \left(2 - \frac{a_{t_1}}{a_{t_2}}  \right) +  \left(4 - \frac{a_{t_1}}{\cdot a_{t_2}} -  \frac{1}{ \left(\frac{a_{t_1}}{a_{t_2}} - 1 \right)}\right) \\ 
&= \frac{a_{t_1}}{a_{t_2}} - \left( \frac{a_{t_1}}{a_{t_2}} - 1 \right) \frac{1}{\left(\frac{a_{t_1}}{a_{t_2}} - 1 \right)} \\
&= \frac{a_{t_1}}{\cdot a_{t_2}} - 1. 
\end{align*}
This also implies that:
\[ \mathbb{P}[T = t_1] =1 -\mathbb{P}[T = t_2] = 2 - \frac{a_{t_1}}{a_{t_2}} \]
as desired. 

\paragraph{Case 3: $1 \le a_{t_1} / a_{t_2} \le (5 - \sqrt{5}) / 2$.}
We observe that: 
\begin{align*}
\mathbb{P}[T = t_2] 
&= 2 a_{t_2} \cdot \left( \frac{3 - \frac{a_{t_1}}{a_{t_2}}}{2 a_{t_2} \cdot \left(2 - \frac{a_{t_1}}{a_{t_2}} \right)} - \frac{1}{a_{t_2}} \right)  \cdot \left( 2 - \frac{a_{t_1}}{a_{t_2}} \right)\\
&= \left(\frac{3 - \frac{a_{t_1}}{a_{t_2}}}{2 - \frac{a_{t_1}}{a_{t_2}}} - 2 \right)  \cdot \left( 2 - \frac{a_{t_1}}{a_{t_2}} \right)\\
&= \left(\frac{\frac{a_{t_1}}{a_{t_2}} - 1}{2 - \frac{a_{t_1}}{a_{t_2}}} \right)  \cdot \left( 2 - \frac{a_{t_1}}{a_{t_2}} \right)\\
&= \frac{a_{t_1}}{a_{t_2}} - 1.
\end{align*}
This also implies that:
\[ \mathbb{P}[T = t_1] =1 -\mathbb{P}[T = t_2] = 2 - \frac{a_{t_1}}{a_{t_2}} \]
as desired. 

\end{proof}

\section{Proofs for \Cref{sec:model}}\label{appendix:proofsmodel}

We prove Theorem \ref{thm:equilibriumexistence}. The proof follows similarly to existence of equilibrium proof in \citep[Proposition 2,][]{JGS22}, and we similarly leverage equilibrium existence technology for discontinuous games \citep{reny99}.

The proof will use the following lemma, which is a simple fact about continuously differentiable functions that we reprove for completeness. 
\begin{lemma}
\label{lemma:lipschitz}
Let $Q \subseteq \mathbb{R}_{\ge 0}^2$ be a compact set. Any continuously differentiable function $f: Q \rightarrow \mathbb{R}_{\ge 0}$ is Lipschitz in the metric $d(x,y) = \|x - y\|_2$.   
\end{lemma}
\begin{proof}
By assumption, the gradient mapping $G: \mathbb{R}_{\ge 0}^2 \rightarrow \mathbb{R}^2_{\ge 0}$ given by $G(w) = \nabla(w)$ is continuous. Any continuous function on a compact set is bounded, so we know that for some constant $B$ it holds that $\|\nabla(w)\|_2 \le B$ for all $w \in Q$. Since the gradient is bounded, this means that the function is Lipschitz as desired. 
\end{proof}

We are now ready to prove Theorem \ref{thm:equilibriumexistence}.
\begin{proof}[Proof of Theorem \ref{thm:equilibriumexistence}]

We leverage a standard result about the existence of symmetric, mixed strategy equilibria in discontinuous games (see Corollary 5.3 of \citep{reny99}). We adopt the terminology of \citet{reny99}  and refer the reader to \citet{reny99} for a formal definition of the conditions. 

First, the game is symmetric by construction. In particular, the creators have symmetric utility functions. Even though the functions $U_i$ as written are not explicitly symmetric, the fact that we break ties uniformly at random means that $U_i(w_i; \textbf{w}_{-i}) = U_j(w_i; \textbf{w}_{-i})$ for all $i, j \in [P]$. We thus let $U(w_1, w_{-1}) = U_1(w_1, w_{-1})$ denote this utility function for the remainder of the analysis. 

To show the existence of a symmetric mixed equilibrium, it suffices to show that: (1) the action space is convex and compact and (2) the game is diagonally better-reply secure.

\paragraph{Creator action space is convex and compact.} In the current game, the action space $\mathbb{R}_{\ge 0} \times \mathbb{R}_{\ge 0}$ is not compact. However,  we show that we can define a modified game, where the action space is convex and compact, and where an equilibrium in this modified game is also an equilibrium in the original game. For the remainder of the proof, we analyze this modified game. 

We define the modified action space as follows. 
Let $\wcostly^{\text{max}}$ be defined to be:
\[\wcostly^{\text{max}} = 1 +  \sup \left\{ \wcostly \ge 0 \mid c([\wcostly, 0]) \le 1 \right\} .\] Let $\wcheap^{\text{max}}$ be defined to be:
\[\wcheap^{\text{max}} := 1 + \sup \left\{
\wcheap \ge 0 \mid \text{there exists } t \in \mathcal{T} \text{ such that } u([\wcostly^{\text{max}}, \wcheap], t) \ge 0 \right\}. \]
(We add an additive factor of $1$ slack to guarantee that there exists a best-response by a creator will be in the \textit{interior} of the action space and not on the boundary.) We take the action space to be \[\mathcal{W} := [0, \wcostly^{\text{max}}] \times [0, \wcheap^{\text{max}}],\]
which is compact and convex by construction. 

We show that for any distribution $\mu$ over $\mathcal{W}$, there exists a best-response $w^* \in \mathbb{R}_{\ge 0}^2$ to: 
\[\argmax_{w \in \mathbb{R}_{\ge 0}^2} \mathbb{E}_{\textbf{w}_{-i} \sim \mu^{P-1}} [U_i(w_i; \textbf{w}_{-i})]  \]
such that $w^*$ is in the interior of $\mathcal{W}$. To show this, let $w^*$ be any best-response to the above optimization program. First we show that $\wcostly^* < \wcostly^{\text{max}}$. Assume for sake of contradiction that $\wcostly^* \ge \wcostly^{\text{max}}$. Then it must hold that 
\[c(w^*) \ge c([\wcostly^*, 0]) \ge c([\wcostly^{\text{max}}, 0]) > 1.\] 
This means that 
\[\mathbb{E}_{\textbf{w}_{-i} \sim \mu^{P-1}} [U_i(w^*; \textbf{w}_{-i})] < 0 \le \mathbb{E}_{\textbf{w}_{-i} \sim \mu^{P-1}} [U_i([0, 0]; \textbf{w}_{-i})],\]
which is a contradiction. This proves that $\wcostly^* < \wcostly^{\text{max}}$ as desired. We next show that we can construct a best-response $w'$ such that $\wcostly' = \wcostly^*$ and $\wcheap' < \wcheap^{\text{max}}$. If $u(w', t) \ge 0$ for some $t \in \mathcal{T}$, then it must hold that $\wcheap^* < \wcheap^{\text{max}}$, so we can take $\wcheap' = \wcheap^*$. If $u(w', t) < 0$ for all $t \in \mathcal{T}$, then we see that:
\[\mathbb{E}_{\textbf{w}_{-i} \sim \mu^{P-1}} [U_i(w^*; \textbf{w}_{-i})] \le 0 \le \mathbb{E}_{\textbf{w}_{-i} \sim \mu^{P-1}} [U_i([\wcostly^*, 0]; \textbf{w}_{-i})],\] so we can take $\wcheap' = 0 < \wcheap^{\text{max}}$. 
Altogether, this proves that there exists a best-response $w'$ satisfying $\wcostly' < \wcostly^{\text{max}}$ and $\wcheap' < \wcheap^{\text{max}}$, which means that $w'$ is in the interior of $\mathcal{W}$

This proves that any symmetric mixed equilibrium of the game with restricted action space $\mathcal{W}$ will also be a symmetric mixed equilibrium of the new game. 

\paragraph{Establishing diagonal better reply security.} In this analysis, we slightly abuse notation and implicitly extend the definition of each utility function $U$ to mixed strategies by considering expected utility. 

First, we show  the payoff function $U(\mu; [\mu, \ldots, \mu])$ (where $\mu$ is a distribution over the action space $\mathcal{W}$) is upper semi-continuous in $\mu$ with respect to the weak* topology. Using the fact that each creator receives a $1/P$ fraction of users in expectation at a symmetric solution, we see that:
\[ U(\mu; [\mu,\ldots, \mu]) = \frac{1}{P} \cdot \frac{1}{|\mathcal{T}|} \sum_{t \in \mathcal{T}} \left(1 - \left(\int_{w} \mathbbm{1}[u(w,t) < 0] d\mu\right)^{P-1}\right) - \int_w c(w) d\mu. \]
Since $c$ is a continuous function, we see immediately that $\int_w c(w) d\mu$ is continuous in $\mu$. Moreover, for each $t \in \mathcal{T}$, we see that $\int_{w} \mathbbm{1}[u(w,t) < 0] d\mu$ is lower semi-continuous in $\mu$.  This proves that $U(\mu; [\mu, \ldots, \mu])$ is upper upper semi-continuous in $\mu$ as desired.

For each relevant payoff in the closure of the graph of the game's diagonal payoff function, we construct an action that secures that payoff along the diagonal. More formally, let $(\mu^*, \alpha^*)$ be in the closure of the graph of the game's diagonal payoff function, and suppose that $(\mu^*, \ldots, \mu^*)$ is not an equilibrium; it suffices to show that a creator can secure a payoff of $\alpha > \alpha^*$ along the diagonal at $(\mu^*, \ldots, \mu^*)$. Since $U$ is upper semi-continuous, it actually suffices to show the statement for $(\mu^*, \alpha^*)$ where $\alpha^* = U(\mu^*; [\mu^*, \ldots, \mu^*])$ and $(\mu^*, \ldots, \mu^*)$ is not an equilibrium. For each such $(\mu^*, \alpha^*)$, we construct $\mu^{\text{sec}}$ that secures a payoff of $\alpha > \alpha^*$ along the diagonal at $(\mu^*, \ldots, \mu^*)$ as follows.

Since $(\mu^*, \ldots, \mu^*)$ is not an equilibrium and since there exists a best-response in the interior of $\mathcal{W}$ as shown above, we know that there exists $w$ in the interior of $\mathcal{W}$ such that:
\[U(w; [\mu^*, \ldots, \mu^*]) > U(w; [\mu^*, \ldots, \mu^*]) = \alpha^*.\]
Since we want to find $w$ that achieves high profit in an open neighborhood of $\mu^*$, we need to strengthen the above statement; we can achieve by this by appropriately perturbing $w$ (which we can do since $w$ is in the interior of $\mathcal{W}$). First, we can perturb $w$ to $\tilde{w}$ such that the distribution $\MPlatform(w')$ where $w' \sim \mu^*$ does not have a point mass at $\MPlatform(\tilde{w})$, and such that:
\[ U(\tilde{w}; [\mu^*, \ldots, \mu^*]) = \frac{1}{|\mathcal{T}|} \sum_{t \in \mathcal{T}} \mathbbm{1}[u(\tilde{w}, t) \ge 0] \cdot \left(\mathbb{P}_{w' \sim \mu^*} [\MPlatform(\tilde{w}) > \MPlatform(w') \text{ or } u(w', t) < 0]\right)^{P-1} - c(\tilde{w}) > \alpha^*.\]
Now, we construct $w^{\text{sec}}$ as a perturbation of $\tilde{w}$ along the costly dimension $\wcheap$ to add $\epsilon$ slack to the constraint $\MPlatform(\tilde{w}) > \MPlatform(w')$. Since $\MPlatform$ is strictly increasing in the expensive component and since $\mathbb{P}_{w' \sim \mu^*}[u(w', t) \in (0, \epsilon)] \rightarrow 0$ as $\epsilon \rightarrow 0$\footnote{To see this, let $S_i = (0, 2^{-n})$ and use that $0 = \mu^*(\cap_{i \ge 1} S_i) = \lim_{i \rightarrow \infty} \mu^*(S_i)$.}, we observe that there exists $\epsilon^* > 0$ and $w^{\text{sec}} \in \mathcal{W}$ (constructed as $w^{\text{sec}} =  \tilde{w} + [\epsilon', 0]$ for some $\epsilon' > 0$) such that 
\begin{equation}
\label{eq:construction} 
\begin{split}
\alpha^* 
&< \frac{1}{|\mathcal{T}|} \sum_{t \in \mathcal{T}} \mathbbm{1}[u(w^{\text{sec}}, t) \ge 0] \cdot \left(\mathbb{P}_{w' \sim \mu^*} [\MPlatform(w^{\text{sec}}) > \MPlatform(w') + \epsilon^* \text{ or } u(w', t) < -\epsilon^* ]\right)^{P-1} - c(w^{\text{sec}})  \\
&\le \frac{1}{|\mathcal{T}|} \sum_{t \in \mathcal{T}} \mathbbm{1}[u(w^{\text{sec}}, t) \ge 0] \cdot \left(\mathbb{P}_{w' \sim \mu^*} [\MPlatform(w^{\text{sec}}) > \MPlatform(w') \text{ or } u(w', t) < 0 ]\right)^{P-1} - c(w^{\text{sec}}) \\
&= U(w^{\text{sec}}; [\mu^*, \ldots, \mu^*]).
\end{split}
\end{equation}

We claim that $\mu^{\text{sec}}$ taken to be the point mass at $w^{\text{sec}}$ will secure a payoff of 
\[\alpha = \frac{\frac{1}{|\mathcal{T}|} \sum_{t \in \mathcal{T}} \mathbbm{1}[u(w^{\text{sec}}, t) \ge 0] \cdot \left(\mathbb{P}_{w' \sim \mu^*} [\MPlatform(w^{\text{sec}}) > \MPlatform(w') + \epsilon^* \text{ or } u(w', t) < -\epsilon^*]\right)^{P-1} - c(w^{\text{sec}}) + \alpha^*}{2} > \alpha^* \] along the diagonal at $(\mu^*, \ldots, \mu^*)$. For  each $t \in \mathcal{T}$, we define the event $A_t$ to be:
\[A_t = \left\{w' \mid \MPlatform(w^{\text{sec}}) > \MPlatform(w') \text{ or } u(w', t) < 0 \right\}\]
and for $\epsilon > 0$, we define the event $A_t^{\epsilon}$ as:
\[A_t^{\epsilon} = \left\{w' \mid \MPlatform(w^{\text{sec}}) > \MPlatform(w') + \epsilon \text{ or } u(w', t) < -\epsilon \right\} .\]
In this notation, we can rewrite equation \eqref{eq:construction} as:
\[U(w^{\text{sec}}; [\mu^*, \ldots, \mu^*]) \ge \frac{1}{|\mathcal{T}|} \sum_{t \in \mathcal{T}} \mathbbm{1}[u(w^{\text{sec}}, t) \ge 0] \cdot \left(\mu^*(A_t^{\epsilon^*})\right)^{P-1} - c(w^{\text{sec}}) > \alpha^*\]
and $\alpha$ as:
\[\alpha = \frac{\frac{1}{|\mathcal{T}|} \sum_{t \in \mathcal{T}} \mathbbm{1}[u(w^{\text{sec}}, t) \ge 0] \cdot \left(\mu^*(A_t^{\epsilon^*})\right)^{P-1} - c(w^{\text{sec}}) + \alpha^*}{2} > \alpha^* \]

We define a metric $d$ on $\mathbb{R}_{\ge 0}^2$  as follows. Using Lemma \ref{lemma:lipschitz}, we know that $\MPlatform(\cdot)$ and $u(\cdot, t)$ for each $t \in \mathcal{T}$ are Lipschitz in $\|\cdot\|_2$. Let the Lipschitz constants be $L_M$ and $L_t$  for each $t \in \mathcal{T}$, respectively. Consider the metric on $\mathbb{R}_{\ge 0}^2$ given by 
\[d(w,w') = \max(L_M, \max_{t \in \mathcal{T}} L_t) \cdot \|w-w'\|_2.\] 
For $\epsilon > 0$ let $B_{\epsilon}(\mu^*)$ denote the $\epsilon$-ball with respect to the Prohorov metric; using the definition of the weak* topology, we see that $B_{\epsilon}(\mu^*)$ is an open set with respect to the weak* topology. 
For every $w' \in A_t^{\epsilon}$, we see that $A_t$ contains the open neighborhood $B_{\epsilon}(w')$ with respect to $d$. 
By the definition of the Prohorov metric, we know that for all $\mu' \in B_{\epsilon}(\mu^*)$, it holds that
\[\mu'(A_i) \ge \mu^*(A_i^{\epsilon}) - \epsilon.\] 
This implies that
\begin{align*}
 U(w^{\text{sec}}; [\mu', \ldots, \mu']) &\ge 
\frac{1}{|\mathcal{T}|} \sum_{t \in \mathcal{T}} \mathbbm{1}[u(w^{\text{sec}}, t) \ge 0] \cdot \left(\mu'(A_t)\right)^{P-1} - c(w^{\text{sec}}) \\
&\ge 
\frac{1}{|\mathcal{T}|} \sum_{t \in \mathcal{T}} \mathbbm{1}[u(w^{\text{sec}}, t) \ge 0] \cdot \left(\mu^*(A_t^{\epsilon}) - \epsilon \right)^{P-1} - c(w^{\text{sec}})  \\
&\ge \left(\frac{1}{|\mathcal{T}|} \sum_{t \in \mathcal{T}} \mathbbm{1}[u(w^{\text{sec}}, t) \ge 0] \cdot \left(\mu^*(A_t^{\epsilon})\right)^{P-1} - c(w^{\text{sec}})\right) \\
&- \frac{1}{|\mathcal{T}|} \sum_{t \in \mathcal{T}}  \left(\mathbbm{1}[u(w^{\text{sec}}, t) \ge 0] \cdot \underbrace{\left(\mu^*(A_t^{\epsilon})\right)^{P-1} -  \left(\mu^*(A_t^{\epsilon}) - \epsilon \right)^{P-1}}_{(A)}  \right).
\end{align*}
Using that (A) goes to $0$ as $\epsilon$ goes to $0$, we see that for sufficiently small $\epsilon$, it holds that:
\begin{align*}
&\frac{1}{|\mathcal{T}|} \sum_{t \in \mathcal{T}}  \left(\mathbbm{1}[u(w^{\text{sec}}, t) \ge 0] \cdot \underbrace{\left(\mu^*(A_t^{\epsilon})\right)^{P-1} -  \left(\mu^*(A_t^{\epsilon}) - \epsilon \right)^{P-1}}_{(A)}\right)\\
&\le \frac{\frac{1}{|\mathcal{T}|} \sum_{t \in \mathcal{T}} \mathbbm{1}[u(w^{\text{sec}}, t) \ge 0] \cdot \left(\mu^*(A_t^{\epsilon^*})\right)^{P-1} - c(w^{\text{sec}}) - \alpha^*}{3}.
\end{align*}
As long as $\epsilon$ is also less than $\epsilon^*$, this means that: 
\begin{align*}
  U(w^{\text{sec}}; [\mu', \ldots, \mu'])
  &\ge \left(\frac{1}{|\mathcal{T}|} \sum_{t \in \mathcal{T}} \mathbbm{1}[u(w^{\text{sec}}, t) \ge 0] \cdot \left(\mu^*(A_t^{\epsilon})\right)^{P-1} - c(w^{\text{sec}})\right) \\
  &- \frac{\frac{1}{|\mathcal{T}|} \sum_{t \in \mathcal{T}} \mathbbm{1}[u(w^{\text{sec}}, t) \ge 0] \cdot \left(\mu^*(A_t^{\epsilon^*})\right)^{P-1} - c(w^{\text{sec}}) - \alpha^*}{3}\\
  &\ge \left(\frac{1}{|\mathcal{T}|} \sum_{t \in \mathcal{T}} \mathbbm{1}[u(w^{\text{sec}}, t) \ge 0] \cdot \left(\mu^*(A_t^{\epsilon^*})\right)^{P-1} - c(w^{\text{sec}})\right) \\
  &- \frac{\frac{1}{|\mathcal{T}|} \sum_{t \in \mathcal{T}} \mathbbm{1}[u(w^{\text{sec}}, t) \ge 0] \cdot \left(\mu^*(A_t^{\epsilon^*})\right)^{P-1} - c(w^{\text{sec}}) - \alpha^*}{3}\\ 
  &= \frac{2 \left(\frac{1}{|\mathcal{T}|} \sum_{t \in \mathcal{T}} \mathbbm{1}[u(w^{\text{sec}}, t) \ge 0] \cdot \left(\mu^*(A_t^{\epsilon^*})\right)^{P-1} - c(w^{\text{sec}})\right) + \alpha^*}{3} \\
  &> \alpha  
\end{align*}
for all $\mu' \in B_{\epsilon}(\mu^*)$, as desired. 
\end{proof}

\section{Proofs for \Cref{sec:positivecorrelation}}\label{appendix:proofssecpositivecorrelation}

To prove Proposition \ref{prop:positivecorrelation}, we first show that the support of the equilibrium is contained within union of curves of the following form in $\mathbb{R}_{\ge 0}^2$:
\[\Caugt = \left\{[\wcostly, \wcheap] \mid u([\wcostly, \wcheap], t) = 0 \right\} \cup \left\{(0, \wcheap) \mid \wcheap \in [0, \min_{w' \mid u(w', t) = 0} \wcheap')] \right\}, \]
as formalized in the following lemma. 
\begin{lemma}
\label{lemma:ctaugunion}
%Let $P \ge 2$ and 
Let $\mathcal{T} \subseteq \mathbb{R}_{\ge 0}$ be any finite type space. Suppose that $(\mu_1, \ldots, \mu_P)$ is an equilibrium in the game with $M = \MPlatform$. Then $\text{supp}(\mu_i) \subseteq \cup_{t \in \mathcal{T}} \Caugt \cup \left\{w \mid c(w) = 0 \right\}$ for all $i \in [P]$. 
\end{lemma}

We prove this lemmas as a corollary to the following sublemma:
\begin{lemma}
\label{lemma:bestresponse}
Let $\mathcal{T} \subseteq \mathbb{R}_{\ge 0}$ be any finite type space. Let $(\mu_1, \ldots, \mu_P)$ be any mixed strategy profile. Then, in the game with $M = \MPlatform$, any best-response $w \in \mathbb{R}_{\ge 0}^2$ to:
\[
\argmax_{w \in \mathbb{R}_{\ge 0}^2} \mathbb{E}_{\textbf{w}_{-i} \sim \mathbf{\mu}_{-i})} [U_i(w; \textbf{w}_{-i})]  
\]
satisfies $w \in \cup_{t \in \mathcal{T}} \Caugt \cup \left\{w \mid c(w) = 0 \right\}$. 
\end{lemma}
\begin{proof}[Proof of Lemma \ref{lemma:bestresponse}]
Assume for sake of contradiction that a best-response $w'$ satisfies $w' \not\in \cup_{t \in \mathcal{T}} \Caugt \cup \left\{w \mid c(w) = 0 \right\}$. Let $m = \MPlatform(w)$. We will show that $w'$ is not a best response. 

Let $T := \left\{t \in \mathcal{T} \mid u(w', t) \ge 0\right\}$ be the set of types for which $w'$ incurs nonnegative user utility. For the remainder of the analysis, we split into two cases: $T = \emptyset$ and $T \neq \emptyset$. 

First, suppose that $T = \emptyset$. Then no user will never consume content, so since $c(w') > 0$, it holds that $\mathbb{E}_{\mathbf{w}_{-i} \sim \mathbf{\mu}_{-i}}[U_i(w'; \mathbf{w}_{-i})] < 0$. However, if the creator were to instead choose $[0,0]$ which incurs $0$ cost, they would get nonnegative utility 
\[\mathbb{E}_{\mathbf{w}_{-i} \sim \mathbf{\mu}_{-i}}[U_i([0,0]; \mathbf{w}_{-i})] \ge 0 > \mathbb{E}_{\mathbf{w}_{-i} \sim \mathbf{\mu}_{-i}}[U_i(w'; \mathbf{w}_{-i})] .\] Thus, $w'$ is not a best response, which is a contradiction.

Next, suppose that $T \neq \emptyset$. By the assumptions on $u$, it holds that $t = \min_{t' \in \mathcal{T}} t' \in T$. Then, $w$ is a feasible solution to the optimization program $\min_w c(w) \text{ s.t. } u(w, t) \ge 0, \MPlatform(w) \ge m$. By Lemma \ref{lemma:uniqueness}, there exists $w^* \in \Caugt$ that is a feasible solution to $\min_w c(w) \text{ s.t. } u(w, t) \ge 0, \MPlatform(w) \ge m$ such that $c(w^*) < c(w')$. Moreover, by assumption (A3), since $u(w, t) \ge 0$, we see that $u(w^*, t') \ge 0$ for all $t' \in T$. Thus, 
\[\mathbb{E}_{\mathbf{w}_{-i} \sim \mathbf{\mu}_{-i}}[U_i(w^*; \mathbf{w}_{-i})] > \mathbb{E}_{\mathbf{w}_{-i} \sim \mathbf{\mu}_{-i}}[U_i(w'; \mathbf{w}_{-i})],\]
so $w'$ is not a best response which is a contradiction. 
\end{proof}

We can now deduce Lemma \ref{lemma:ctaugunion}.
\begin{proof}[Proof of Lemma \ref{lemma:ctaugunion}]
Assume for sake of contradiction that $w \in \text{supp}(\mu_i)$ satisfies $w \not\in \cup_{t \in \mathcal{T}} \Caugt \cup \left\{w \mid c(w) = 0 \right\}$. By Lemma \ref{lemma:bestresponse}, $w$ is not a best response, which is a contradiction.  
\end{proof}

Proposition \ref{prop:positivecorrelation} follows immediately from Lemma \ref{lemma:ctaugunion} along with the lemmas in \Cref{appendix:auxlemmas}. 
\begin{proof}[Proof of Proposition \ref{prop:positivecorrelation}]
Applying Lemma \ref{lemma:ctaugunion} and Lemma \ref{lemma:one-to-one}, we know that:
\[\text{supp}(\mu_i) \subseteq \cup_{t \in \mathcal{T}} \Caugt \cup \left\{w \mid c(w) = 0 \right\} = \cup_{t \in \mathcal{T}} \left\{ [f_t(\wcheap), \wcheap] \mid \wcheap \ge 0 \right\} \cup \left\{w \mid c(w) = 0 \right\}.\]
Using the fact that gaming is not costless, we see that:
\begin{align*}
 \text{supp}(\mu_i) &\subseteq \cup_{t \in \mathcal{T}} \left\{ [f_t(\wcheap), \wcheap] \mid \wcheap \ge 0 \right\} \cup \left\{w \mid c(w) = 0 \right\} \\
 &= \cup_{t \in \mathcal{T}} \left\{ [f_t(\wcheap), \wcheap] \mid \wcheap \ge 0 \right\} \cup \left\{[0,0] \right\}   
\end{align*}
The fact that the functions $f_t$ are weakly increasing follows from Lemma \ref{lemma:increasing}.
\end{proof}

Theorem \ref{thm:positivecorrelationhomogeneous} follows as a consequence of Proposition \ref{prop:positivecorrelation}.
\begin{proof}[Proof of Theorem \ref{thm:positivecorrelationhomogeneous}]
Let $\mu_1, \ldots, \mu_P$ be an equilibrium, and suppose that $w^1, w^2 \in \cup_{i \in [P]} \text{supp}(\mu_i)$. By Proposition \ref{prop:positivecorrelation}, we see that:
\[\cup_{i \in [P]} \text{supp}(\mu_i) \subseteq \left\{ [f_t(\wcheap), \wcheap] \mid \wcheap \ge 0 \right\} \cup \left\{[0,0] \right\}.\]
Using the fact that $f_t$ is weakly increasing and $f_t(0) \ge 0$, we see that if $\wcheap^2 \ge \wcheap^1$, then $\wcostly^2 \ge \wcostly^1$. 
\end{proof}

\section{Proofs for \Cref{sec:equilibriumcharacterizations}}\label{appendix:proofsequilibriumcharacterizations}

\subsection{Proof of \Cref{thm:onetype}}
Before proving Theorem \ref{thm:onetype}, we prove the following properties of $\muengagement(P, c, u, \mathcal{T})$.
\begin{lemma}
\label{lemma:propertiesdistonetype}
Let $\mathcal{T} = \left\{t\right\}$. The distribution $\muengagement(P, c, u, \mathcal{T})$ satisfies the following properties:
\begin{enumerate}[label=(P\arabic*)]
    \item The only possible atom in the distribution $\muengagement(P, c, u, \mathcal{T})$ is at $(0,0)$, and moreover that $(0,0)$ is an atom when $f_t(0) > 0$.
    \item Suppose that $(\wcheap, \wcostly) \in \text{supp}(\muengagement(P, c, u, \mathcal{T}))$. If $(\wcheap, \wcostly) \neq (0,0)$ or if $f_t(0) = 0$, then it holds that $u([\wcheap, \wcostly], t) \ge 0$.
    \item If $(0,0)$ is an atom of $\muengagement(P, c, u, \mathcal{T})$, then $u([0,0], t) < 0$.
\end{enumerate}   
\end{lemma}
\begin{proof}
 To prove (P1), note that if $(\wcheap, \wcostly) \in \text{supp}(\muengagement(P, c, u, \mathcal{T}))$ is an atom, then $\wcheap$ must be an atom in the marginal distribution $\Wcheap$. The specification of the cdf shows that the only possible atom is at $\Wcheap = 0$. Moreover, $0$ is an atom of $\Wcheap$ if and only if $c(f_t(0), 0) > 0$, which occurs if and only if $f_t(0) > 0$. When $f_t(0) > 0$, we further see that the conditional distribution $\Wcostly$ is a point mass at $0$, as desired. 

To prove (P2), note that $\Wcostly$ is a point mass at $f_t(\wcheap)$. By the definition of $f_t$, it holds that $u([\wcheap, \wcostly], t) \ge 0$. 

To prove (P3), note that the first property showed that $(0,0)$ is an atom if and only if $f_t(0) > 0$. By the definition of $f_t$, we see that $u([0,0], t) < 0$ as desired.
   
\end{proof}
We prove Theorem \ref{thm:onetype}. 
\begin{proof}[Proof of Theorem \ref{thm:onetype}]
Let $\mu = \muengagement(P, c, u, \mathcal{T})$ for notational convenience. We analyze the expected utility of  $H(w) = \mathbb{E}_{\mathbf{w}_{-i} \sim \mu_{-i}}[U_i(w; w_{-i})]$ of a content creator if all of the creators choose the strategy $\mu$. It suffices to show that any $w^* \in \text{supp}(\mu)$ is a best response $w^* \in \argmax_{w} H(w)$. We use the properties (P1)-(P4) in Lemma \ref{lemma:propertiesdistonetype}. 

First, we observe that we can write $H(w)$ as:
\begin{align*}
H(w) &= \mathbb{E}_{\mathbf{w}_{-i} \sim \mu_{-i}}[U_i(w; w_{-i})] \\
&= \mathbbm{1}[u(w, t) \ge 0] \cdot \mathbb{P}_{\Wcheap}[\MPlatform(w) > \MPlatform([f_t(\Wcheap), \Wcheap])]^{P-1} - c(w),
\end{align*}
because (P1) implies that the only possible atom occurs at $[0,0]$, (P3) implies that $u([0,0], t) < 0$ if $[0,0]$ is an atom, and (P2) implies that $\mathbbm{1}[u(w, t) \ge 0] = 0$ for $(\wcheap, \wcostly) \neq (0,0)$.  

If $w \in \text{supp}(\mu)$, then we claim that $H(w) = 0$. If $\wcheap = 0$ and $f_t(\wcheap) = 0$, it is immediate that $H(w) = 0$. If $\wcheap > 0$, then $w = [f_t(\wcheap), \wcheap]$. By (P2), it holds that $u(w, t) \ge 0$. This means that: 
\begin{align*}
H(w) &= \mathbb{P}_{\Wcheap}[\MPlatform([f_t(\wcheap), \wcheap]) > \MPlatform([f_t(\Wcheap), \Wcheap])]^{P-1} - c(w) \\
&=_{(1)} \mathbb{P}_{\Wcheap}[\wcheap > \Wcheap]^{P-1} - c(w) \\ 
&= 0,    
\end{align*}
where (1) uses the fact that $\MPlatform([f_t(\wcheap), \wcheap])$ is strictly increasing in $\wcheap$ (Lemma \ref{lemma:increasing}). 

The remainder of the proof boils down to showing that $H(w) \le 0$ for any $w$. If $u(w, t) < 0$, then $H(w) \le 0$. If $u(w, t) \ge 0$, then 
\[ H(w) = \mathbb{P}_{\Wcheap}[\MPlatform(w) > \MPlatform([f_t(\Wcheap), \Wcheap])]^{P-1} - c(w).\]
It suffices to show that $H(w) \le 0$ at any best-response $w$ such that $u(w, t) \ge 0$. If $w$ is a best response and $u(w, t) \ge 0$, then it must be true that $w$ is a solution to \eqref{eq:inducedcost}. By Lemma \ref{lemma:uniqueness}, this means that $w \in \Caugt$, and by Lemma \ref{lemma:one-to-one}, this means that $w$ is of the form $[f_t(\wcheap), \wcheap]$, which means that:
\[H(w) = \mathbb{P}_{\Wcheap}[\wcheap > \Wcheap]^{P-1} - c(w) \le 0, \]
which proves the desired statement. 
    
\end{proof}

\subsection{Proof of Theorem \ref{thm:onetypeunique}}

We prove Theorem \ref{thm:onetypeunique}.
\begin{proof}[Proof of Theorem \ref{thm:onetypeunique}]

The high-level idea of the proof is to define a new game with a restricted action space that is easier to analyze (we define slightly different variant games for Part 2 and Part 3). Suppose that $\mu$ is a symmetric mixed equilibrium. By Lemma \ref{lemma:ctaugunion} and Lemma \ref{lemma:one-to-one}, we see that: 
\[\text{supp}(\mu) \subseteq \Caugt \cup \left\{w \mid c(w) = 0 \right\} = \left\{ [\wcheap, f_t(\wcheap)] \mid \wcheap \ge 0 \right\} \cup \left\{w \mid c(w) = 0 \right\}.\]
Using the assumption that gaming tricks are costly (i.e. $(\nabla c(w))_2 > 0$ for all $w \in \mathbb{R}_{\ge 0}^2$), we obtain that:
\[\text{supp}(\mu) \subseteq \left\{ [\wcheap, f_t(\wcheap)] \mid \wcheap \ge 0 \right\} \cup \left\{[0, 0] \right\}.\]

This simplification will ultimately enable us to convert the 2-dimensional action space to the 1-dimensional space specified by engagement. By Lemma \ref{lemma:increasing} and since $\MPlatform$ is strictly increasing in both arguments, it holds that $\MPlatform([f_t(\wcheap), \wcheap])$ is a strictly increasing function of $\wcheap$. Let 
\[m_{\text{min}} := \inf_{\wcheap \ge 0} \MPlatform([f_t(\wcheap), \wcheap])\]
and 
\[m_{\text{max}} := \sup_{\wcheap \ge 0} \MPlatform([f_t(\wcheap), \wcheap]).\]
(Note that $m_{\text{max}}$ might be equal to $\infty$.) This means that for each value $m \in [m_{\text{min}}, m_{\text{max}})$, there is exactly one value $w = [f_t(\wcheap), \wcheap]$ such that $\MPlatform(w) = m$. 

We break into two cases: $f_t(\wcheap) = 0$ and $f_t(\wcheap) \ge 0$ for the remainder of the analysis. 

\paragraph{Case 1: $f_t(\wcheap) = 0$.} In the case that $f_t(\wcheap) = 0$, we can further simplify:
\[\text{supp}(\mu) \subseteq \left\{ [\wcheap, f_t(\wcheap)] \mid \wcheap \ge 0 \right\} \cup \left\{[0, 0] \right\} = \left\{ [\wcheap, f_t(\wcheap)] \mid \wcheap \ge 0 \right\}.\]

Let's consider a different game where the action set is $A = [m_{\text{min}}, m_{\text{max}})$. The action $m \in [m_{\text{min}}, \infty)$ corresponds to the unique value $w = [f_t(\wcheap), \wcheap, ]$ such that $\MPlatform(w) = m$. In this new game, the utility function is 
\begin{equation}
\label{eq:tildeU}
  \tilde{U}(a_i; a_{-i}) := 1[a_i = \max_{j \in [P]} a_j] - \tilde{c}(a)  
\end{equation}
where ties are broken uniformly at random and where $\tilde{c}(m) = c(w)$ where $w$ is the unique value in $\left\{ [\wcheap, f_t(\wcheap)] \mid \wcheap \ge 0 \right\}$ such that $\MPlatform(w) = m$. We see that $\tilde{c}$ is strictly increasing in $m$ since $f_t(\wcheap)$ is weakly increasing in $\wcheap$ by Lemma \ref{lemma:increasing} and since gaming tricks are costly by assumption. (In the remainder of the proof, we also slightly abuse notation and implicitly extend the definition of $\tilde{U}$ to mixed strategies by considering expected utility.) 

We first show that there exists a symmetric equilibrium in the new game with action set $A$. We use the fact that we have constructed a symmetric equilibrium in the original game in Theorem \ref{thm:onetype}. The transformed distributions $\tilde{\mu}$ is a symmetric equilibrium in the new game; thus, there exists an equilibrium in the new game. 

We claim that it suffices to show that there is at most one symmetric equilibrium in the new game with action set $A$. Every action $a \in A$ corresponds to a unique $w$. Thus, uniqueness of equilibrium in the transformed game guarantees uniqueness of equilibrium in the original game. 

The remainder of the proof of this case boils down to showing that there is at most one symmetric mixed equilibrium $\tilde{\mu}$ in the new game. Let $\tilde{\mu}$ be a symmetric mixed Nash equilibrium of the transformed game. 

First, we claim that $\mathbb{P}_{M \sim \tilde{\mu}}[M = m'] = 0$ for all $m'$ (no point masses). If $\mathbb{P}[M = m'] > 0$, then because of uniform-at-random tiebreaking, it holds that $m' + \epsilon$ for some $\epsilon$ performs strictly better than $m'$ for some $\epsilon > 0$.

Next, we claim that $m_{\text{min}} \in \text{supp}(\tilde{\mu})$. Assume for sake of contradiction that $m' = \inf_{m \in \text{supp}(\tilde{\mu})} > m_{\text{min}}$. At $m = m'$, the creator gets a utility of $-\tilde{c}(m') < 0$, which means that the creator would get higher utility from the deviation $m = m_{\text{min}}$ where the utility would be $0$. (Since $f_t(\wcheap) = 0$, it holds that $\tilde{c}(m_{\text{min}}) = 0$.) This is a contradiction.  

For $m \in \text{supp}(\tilde{\mu})$, we claim that $\mathbb{P}_{M \sim \tilde{\mu}}[M \le m] = (\tilde{c}(m))^{1/(P-1)}$. This is because $m_{\text{min}} \in \text{supp}(\tilde{\mu})$ and the utility at $m_{\text{min}}$ is $0$, so the utility at any $m \in \text{supp}(\tilde{\mu})$ must be $0$. This implies that 
$\mathbb{P}_{M \sim \tilde{\mu}}[M \le m] = (\tilde{c}(m))^{1/(P-1)}$. 

Finally, it suffices to show that the support is a closed interval of the form $[m_{\text{min}}, m^*]$ where $m^*$ is the unique value such that $\tilde{c}(m^*) = 1$. We first see that since $\mathbb{P}_{M \sim \tilde{\mu}}[M \le m] \le 1$, it must hold that $m \le m^*$ for any $m \in \text{supp}(\tilde{\mu})$. 
To see this, let $Q = \text{supp}(\tilde{\mu}) \cup [m^*, m_{\text{max}}]$. Since $Q$ is a finite union of closed sets, it is a closed set. This means that $\bar{Q} = [m_{\text{min}}, m_{\text{max}}] \setminus Q$ is an open set. It suffices to prove that $\bar{Q} = \emptyset$. Assume for sake of contradiction $\bar{Q} \neq \emptyset$. If $m' \in \bar{Q}$, let $m_1 = \inf \left\{ m < m' \mid m \in \text{supp}(\tilde{\mu}) \right\}$ and $m_2 = \sup \left\{ m > m' \mid m \in \text{supp}(\tilde{\mu}) \right\}$. Since $\bar{Q}$ is open, there is an open neighborhood such that $B_{\epsilon}(m') \subseteq \bar{Q}$, which means that $m_2 > m_1$. However, this means that $\mathbb{P}_{M \sim \tilde{\mu}}[M = m_2] -= \mathbb{P}_{M \sim \tilde{\mu}}[M \le m_2] - \mathbb{P}_{M \le \tilde{\mu}}[M = m_2] > 0$, which is a contradiction.  

This proves the desired statement for Case 1. 

\paragraph{Case 2: $f_t(\wcheap) > 0$.} For this case, let's consider a different game where the action set is $A = \left\{\bot\right\} \cup [m_{\text{min}}, \infty)$. The action $m \in [m_{\text{min}}, \infty)$ corresponds to the unique value $w = [f_t(\wcheap), \wcheap, ]$ such that $\MPlatform(w) = m$. In this new game, the utility function is 
\begin{equation}
\label{eq:tildeUbot}
  \tilde{U}(a_i; a_{-i}) := 1[a_i \neq \bot] \cdot 1[a_i = \max_{j \in [P]} a_j] - \tilde{c}(a)  
\end{equation}
where we use the ordering that $\bot < m_{\min}$, where ties are broken uniformly at random and where $\tilde{c}(\bot) = 0$ and $\tilde{c}(m) = c(w)$ where $w$ is the unique value in $\left\{ [f_t(\wcheap), \wcheap] \mid \wcheap \ge 0 \right\}$ such that $\MPlatform(w) = m$. We similarly see that $\tilde{c}$ is strictly increasing in $m$ since $f_t(\wcheap)$ is weakly increasing in $\wcheap$ by Lemma \ref{lemma:increasing} and since gaming tricks are costly by assumption. (As before, in the remainder of the proof, we also slightly abuse notation and implicitly extend the definition of $\tilde{U}$ to mixed strategies by considering expected utility.)

By an analogous argument to the previous case, we know that there exists a symmetric equilibrium in the new game with action set $A$, and it suffices to show that there is at most one symmetric equilibrium in the new game with action set $A$. The remainder of the proof of this case boils down to showing that there is at most one symmetric mixed equilibrium $\tilde{\mu}$ in the new game. Let $\tilde{\mu}$ be a symmetric mixed Nash equilibrium of the transformed game. We split into several subcases: (1) $\tilde{c}(m_{\text{min}}) > 1$, (2) $\tilde{c}(m_{\text{min}})  = 1$, and (3) $0 < \tilde{c}(m_{\text{min}}) > 1 < 1$. 

\paragraph{Subcase 2a: $\tilde{c}(m_{\text{min}}) > 1$.} Choosing $\bot$ is a strictly dominant strategy. This means that the unique equilibrium is where each $\tilde{\mu}$ is a point mass at $\bot$. 

\paragraph{Subcase 2b: $\tilde{c}(m_{\text{min}}) = 1$.} We  claim that $\tilde{\mu}$ as a point mass $\bot$ is still the unique symmetric mixed Nash equilibrium. Assume for sake of contradiction that there is an equilibrium $\tilde{\mu}$ where $\mathbb{P}_{m \sim \tilde{\mu}}[a \neq \bot] > 0$. It must be true that $\text{supp}(\tilde{\mu}) \subseteq \left\{\bot, m_{\text{min}} \right\}$ (otherwise, $\bot$ would be a better response). Thus,  $\mathbb{P}_{a^* \sim \tilde{\mu}}[a \neq \bot] > 0$. However, this means that: 
\[\tilde{U}(a^*; [\tilde{\mu}, \ldots, \tilde{\mu}]) < 0 = \tilde{U}(\bot; [\tilde{\mu}, \ldots, \tilde{\mu}])\] because of uniform-at-random tiebreaking. This is a contradiction since $a^*$ needs to be a best response. 

\paragraph{Subcase 2b: $0 < \tilde{c}(m_{\text{min}}) < 1$.}
First, we claim that $\mathbb{P}[m = m'] = 0$ for all $m' \in [m_{\text{min}}, \infty)$. If $\mathbb{P}[m = m'] > 0$, then because of uniform-at-random tiebreaking, it holds that $m' + \epsilon$ for some $\epsilon$ performs strictly better than $m'$ for some $\epsilon > 0$.

Next, we claim that $\mathbb{P}_{\tilde{\mu}}[\bot]  = \tilde{c}(m_{\text{min}})^{1/(P-1)}$. First, assume for sake of contradiction that $\mathbb{P}_{\tilde{\mu}}[\bot] > \tilde{c}(m_{\text{min}})^{1/(P-1)}$: in this case, 
\[U(m_{\text{min}}; [\tilde{\mu}, \ldots, \tilde{\mu}]) = \alpha^{P-1} - \tilde{c}(m_{\text{min}}) > 0 = U(\bot; [\tilde{\mu}, \ldots, \tilde{\mu}])\] which is a contradiction. Next, assume for sake of contradiction that $\mathbb{P}_{\tilde{\mu}}[\bot] < \tilde{c}(m_{\text{min}})^{1/(P-1)}$. Since $\tilde{c}(m_{\text{min}})^{1/(P-1)} < 1$ by assumption, this means that that $\mathbb{P}_{\tilde{\mu}}[\bot] < 1$. Let $m' = \min \left\{m \ge m_{\text{min}} \mid m \in \text{supp}(\mu)\right\}$. 
Using that $\mathbb{P}_{\tilde{\mu}}[m = m_{\text{min}}] = 0$, it holds that:
\[\tilde{U}(m'; [\tilde{\mu}, \ldots, \tilde{\mu}]) = \alpha^{P-1} - \tilde{c}(m') \le \alpha^{P-1} - \tilde{c}(m_{\text{min}}) < 0 = U(\bot; [\tilde{\mu}, \ldots, \tilde{\mu}]),\]
which is a contradiction. 

For $m \in \text{supp}(\tilde{\mu})$, we claim that $\mathbb{P}_{M \sim \tilde{\mu}}[M \le m] = (\tilde{c}(m))^{1/(P-1)}$. This is because $\bot \in \text{supp}(\tilde{\mu})$ and the utility at $\bot$ is $0$, so the utility at any $m \in \text{supp}(\tilde{\mu})$ must be $0$. This implies that 
$\mathbb{P}_{M \sim \tilde{\mu}}[M \le m] = (\tilde{c}(m))^{1/(P-1)}$. 

Finally, it suffices to show that the support is a closed interval of the form $[m_{\text{min}}, m^*]$ where $m^*$ is the unique value such that $\tilde{c}(m^*) = 1$. We first see that since $\mathbb{P}_{M \sim \tilde{\mu}}[M \le m] \le 1$, it must hold that $m \le m^*$ for any $m \in \text{supp}(\tilde{\mu})$. 
To see this, let $S = \text{supp}(\tilde{\mu}) \cup [m^*, m_{\text{max}}]$. Since $Q$ is a finite union of closed sets, it is a closed set. This means that $\bar{Q} = [m_{\text{min}}, m_{\text{max}}] \setminus Q$ is an open set. It suffices to prove that $\bar{Q} = \emptyset$. Assume for sake of contradiction $\bar{Q} \neq \emptyset$. If $m' \in \bar{Q}$, let $m_1 = \inf \left\{ m < m' \mid m \in \text{supp}(\tilde{\mu}) \right\}$ and $m_2 = \sup \left\{ m > m' \mid m \in \text{supp}(\tilde{\mu}) \right\}$. Since $\bar{Q}$ is open, there is an open neighborhood such that $B_{\epsilon}(m') \subseteq \bar{Q}$, which means that $m_2 > m_1$. However, this means that $\mathbb{P}_{M \sim \tilde{\mu}}[M = m_2] -= \mathbb{P}_{M \sim \tilde{\mu}}[M \le m_2] - \mathbb{P}_{M \le \tilde{\mu}}[M = m_2] > 0$, which is a contradiction.

\end{proof}

\subsection{Useful setup and lemmas for heterogeneous users}\label{subsec:tools}

In our analysis, it is cleaner to work in the reparametrized space 
\[S := \left\{(\MPlatform([\wcostly, \wcheap]) - s, t) \mid t \in \mathcal{T}, u([\wcostly, \wcheap], t) = 0 \right\}\]
than directly over the content space $\mathbb{R}_{\ge 0}$. (This is the same reparameterized space described in \Cref{appendix:characterization2types}.) Recall that we map each $(v, t) \in S$ to the unique content $h(v, t) \in \mathbb{R}_{\ge 0}^2$ of the form $h(v, t) = [f_{t}(\wcheap), \wcheap]$ such that $\MPlatform([f_{t}(\wcheap), \wcheap]) = v - s$. Conceptually, $h(v,t)$ captures content with engagement $v - s$ optimized for winning type $t$. 

Using that the coefficients $a_t$ are strictly decreasing as given by Assumption \ref{assumption:linearity}, it is easy to see that $h$ is an one-to-one function mapping $S$ to 
\[\cup_{t \in \mathcal{T}} \Ct = \left\{ [\wcostly, \wcheap] \mid t \in \mathcal{T}, u([\wcostly, \wcheap], t) = 0 \right\}.\] We let $h^{-1}$ denote its inverse which is defined on the image $\cup_{t \in \mathcal{T}} \Ct $.

We show that if the support of $\mu$ is contained in $\cup_{t \in \mathcal{T}} \Ct$, then there exists a best response $w$ in $\Ct$.
\begin{lemma}
\label{lemma:bestresponsestronger}
Suppose that gaming is costless (i.e., $\nabla(c(w))_2 = 0$ for all $w \in \mathbb{R}_{\ge 0}^2$), and suppose that $u([0, 0], t) \ge 0$ for all $t \in \mathcal{T}$. Let $\mu$ be a distribution supported on $\cup_{t \in \mathcal{T}} \Ct$. Then, in the game with $M = \MPlatform$, there exists a best response $w \in \mathbb{R}_{\ge 0}^2$ to:
\[\argmax_{w \in \mathbb{R}_{\ge 0}^2} \mathbb{E}_{\textbf{w}_{-i} \sim \mu^{P-1}} [U_i(w; \textbf{w}_{-i})]  \]
such that $w \in \cup_{t \in \mathcal{T}} \Ct$.
\end{lemma}
\begin{proof}
Let $w$ be any best-response. By Lemma \ref{lemma:bestresponse}, we know that $w \in \cup_{t \in \mathcal{T}} \Caugt \cup \left\{w \mid c(w) = 0\right\}$. We split into two cases: (1) $u([\wcostly, \wcheap], t) < 0$ for all $t \in \mathcal{T}$, and (2) $u([\wcostly, \wcheap], t) \ge 0$ for some $t \in \mathcal{T}$. 

\paragraph{Case 1: $u([\wcostly, \wcheap], t) < 0$ for all $t \in \mathcal{T}$.} In this case, no users will consume the content $w$. Letting $t_{\text{min}} = \min(\mathcal{T})$ and $\wcheap'$ be such that $u([0, \wcheap'], t_{\text{min}}) = 0$, since gaming is costless, it holds that:
\[\mathbb{E}_{\textbf{w}_{-i} \sim (\muengagement(P, c, u, \mathcal{T}))^{P-1}} [U_i(w; \textbf{w}_{-i})] \le \mathbb{E}_{\textbf{w}_{-i} \sim (\muengagement(P, c, u, 
\mathcal{T}))^{P-1}} [U_i([0, \wcheap']; \textbf{w}_{-i})]. \]
This means that there exists a best-response in $\cup_{t \in \mathcal{T}} \Ct$. 

\paragraph{Case 2: $u([\wcostly, \wcheap], t) \ge 0$ for some $t \in \mathcal{T}$.} 
For the second case, where  $u([\wcostly, \wcheap], t) \ge 0$ for some $t \in \mathcal{T}$, let: 
\[\wcheap' = \inf \left\{\wcheap'' \mid u([\wcostly, \wcheap''], t) = 0 \text{ for some } t \in \mathcal{T}, \wcheap'' \ge \wcheap \right\}.\] We see that by construction, it holds that $\left\{t \in \mathcal{T} \mid u(w, t) \ge 0 \right\} = \left\{t \in \mathcal{T} \mid u([\wcostly, \wcheap']) \ge 0 \right\}$ and it also holds that $\MPlatform(w) \le \MPlatform([\wcostly, \wcheap'])$. Since gaming is costless, we see that 
\[\mathbb{E}_{\textbf{w}_{-i} \sim (\muengagement(P, c, u, \mathcal{T}))^{P-1}} [U_i(w; \textbf{w}_{-i})] \le \mathbb{E}_{\textbf{w}_{-i} \sim (\muengagement(P, c, u, 
\mathcal{T}))^{P-1}} [U_i([\wcostly, \wcheap']; \textbf{w}_{-i})]. \] This means that there exists a best-response in $\cup_{t \in \mathcal{T}} \Ct$. 
   
\end{proof}

We translate the costs into the reparameterized space. 
\begin{lemma}
\label{lemma:costs}
Let $P = 2$, suppose that gaming tricks are costless (that is, $(\nabla(c(w)))_2 = 0$ for all $w \in \mathbb{R}_{\ge 0}^2$), and suppose that $u([0, 0], t) \ge 0$ for all $t \in \mathcal{T}$. Suppose that Assumption \ref{assumption:linearity} holds. Then for any $w \in \cup_{t \in \mathcal{T}} \Ct$, it holds that:
\[ c(w) = \max(0, a_{t} \cdot v - 1),\]
where $(v, t) = h^{-1}(w)$. 
\end{lemma}
\begin{proof}
We first claim that 
\[ c(w) = \CEngagement_{t}(\MPlatform(w)).\]
To see this, observe that by Lemma \ref{lemma:uniqueness}, there exists $w' \in \Caugt$ such that $\MPlatform(w') \ge \MPlatform(w)$ and $\CEngagement_{t}(\MPlatform(w)) = c(w') \le c(w)$. By Lemma \ref{lemma:increasing}, this implies that $w' = w$, so $c(w) = \CEngagement_{t}(\MPlatform(w))$ as desired.

Next, we observe that by Assumption \ref{assumption:linearity}, it holds that:
\[\CEngagement_{t}(\MPlatform(w)) =  \max(0, a_{t} \cdot v - 1).\]

Putting this all together gives the desired statement. 

\end{proof}

\subsection{Proof of Theorem \ref{thm:Ntypes} and Proposition \ref{prop:2typeswellseparated}}

We prove Theorem \ref{thm:Ntypes} and then deduce Proposition \ref{prop:2typeswellseparated} as a corollary.

We first apply the reparametrization from \Cref{subsec:tools} to the random variable $(\V, T)$ to Definition \ref{def:Ntypes}. Using the specification in in Definition \ref{def:Ntypes}, it is easy to see that every $w \in \text{supp}(\muengagement(P, c, u, \mathcal{T}))$ is $\Ct$, which means that the distribution $h^{-1}(W)$ for $W \sim \muengagement(P, c, u, \mathcal{T})$ is well-defined. In the following proposition, we simplify the distribution $(\V, T)$. 
\begin{proposition}
\label{prop:NtypesMdistribution}
Consider the setup of Definition \ref{def:Ntypes}. Let $(\V, T) = h^{-1}(W)$ where $W \sim \muengagement(P, c, u, \mathcal{T})$. Then it holds that $\mathbb{P}[T = t_i]$ satisfies: 
\[\mathbb{P}[T = t_i] =
\begin{cases}
  \frac{1}{N-i+1} & \text{ if }  1 \le i \le N' -1\\
  1 - \sum_{i'=1}^{N'-1} \frac{1}{N-i'+1} & \text{ if }  i = N' \\
  0 & \text{ if }  i > N'. \\
\end{cases}
\] 
Moreover, the conditional distribution $\V \mid T = t_{i}$ is distributed as a uniform distribution over $\left[\frac{1}{a_{t_i}}, \frac{1}{a_{t_i}} \cdot \left(1 + \frac{1}{N} \right) \right]$ for $1 \le i \le N' - 1$ and a uniform distribution over 
\[\left[\frac{1}{a_{t_{N'}}},  \frac{1}{a_{t_{N'}}} \cdot \left( 1+ \frac{N - N' + 1}{N} \cdot \left(1 - \sum_{j=1}^{N' - 1} \frac{1}{N- j + 1}\right)\right) \right]\]
for $i = N'$. 
\end{proposition}
\begin{proof}
It suffices to show that for each $1 \le i \le N'$, it holds that $\V \mid T = t_{i}$ is distributed according to the uniform distribution specified in the proposition statement. 

The remainder of the analysis boils down to analyzing the distribution $\V \mid T = t_{i}$. Let the distributions $W^i = (\Wcheap^i, \Wcostly^i)$ for $1 \le i \le N'$, and the mixture weights $\alpha^i$ be defined as in Definition \ref{def:Ntypes}. Observe that $\V \mid T = t_{i}$ is distributed as $(\MPlatform([\Wcostly^i, \Wcheap^i] + s, t_{i})$. Note that every $w \in \text{supp}(\Wcostly^i, \Wcheap^i)$ is in the image $\cup_{t \in \mathcal{T}} \Ct$, which means that it can be uniquely expressed as $w = h(v, t)$. To translate the costs, we observe that: 
\begin{align*}
 \CCheap_t(\wcheap) &= c([f_{t_i}(\wcheap), \wcheap]) \\
  &=_{(A)} \max(0, a_{t_i} \cdot v - 1) \\
  &=_{(B)} a_{t_i} \cdot v - 1
\end{align*}
where (A) follows from Lemma \ref{lemma:costs}, and (B) follows the fact that $\wcheap \in \text{supp}(\Wcheap^i)$ implies that $\CCheap_{t_i}(\wcheap) > 0$.  Putting this all together, we see that $\mathbb{P}[\V \le v \mid T = t_i] = 
\mathbb{P}[\Wcheap^i \le \wcheap]$ for any $\wcheap \in \text{supp}(\Wcheap^i)$. Using the specification in Definition \ref{def:Ntypes}, this means that if $1 \le i \le N' - 1$, then it holds that
\begin{align*}
 \mathbb{P}[\V \le v \mid T = t_i] &= 
\mathbb{P}[\Wcheap^i \le \wcheap] \\
&=
 \min\left(N \cdot \CCheap_{t_i}(\wcheap), 1\right)  \\
 &= \min\left( N \cdot (a_{t_i} \cdot v - 1), 1 \right),   
\end{align*}
which implies that $V \mid T = {t_i}$ is a uniform distribution over $\left[\frac{1}{a_{t_i}}, \frac{1}{a_{t_i}} \cdot \left(1 + \frac{1}{N} \right) \right]$ as desired. Similarly, if $i = N'$, then 
\begin{align*}
 \mathbb{P}[\V \le v \mid T = t_i] &= 
\mathbb{P}[\Wcheap^i \le \wcheap] \\
& = \min\left(\frac{N}{N-N'+1} \cdot \left(1 - \sum_{j=1}^{N'-1} \frac{1}{N-j+1} \right)^{-1} \cdot \CCheap_{t_i}(\wcheap), 1\right) \\
&=  \min\left(\frac{N}{N-N'+1} \cdot \left(1 - \sum_{j=1}^{N'-1} \frac{1}{N-j+1} \right)^{-1} \cdot (a_{t_i} \cdot v - 1), 1\right),    
\end{align*}
which implies that $V \mid T = {t_i}$ is a uniform distribution over  
\[\left[\frac{1}{a_{t_{N'}}},  \frac{1}{a_{t_{N'}}} \cdot \left( 1+ \frac{N - N' + 1}{N} \cdot \left(1 - \sum_{j=1}^{N' - 1} \frac{1}{N- j + 1}\right)\right) \right]\] as desired. 
\end{proof}

We prove Theorem \ref{thm:Ntypes}. The main ingredient of the proof is the following lemma which computes the expected utility for different choices of $(v, t) \in S$ when the other creator's actions are selected according to the distribution $(\V, T) = h^{-1}(W)$ for $W \sim \muengagement(P, c, u, \mathcal{T})$. 
\begin{lemma}
\label{lemma:Ntypesanalysis}
Consider the setup of Definition \ref{def:Ntypes}.  Let $(\V, T) = h^{-1}(W)$ be the reparameterized distribution of $W \sim \muengagement(P, c, u, \mathcal{T})$. For any $(v, t) \in S$, the expected utility satisfies:
\[\mathbb{E}_{\V, T} [U_1(h^{-1}(v, t); (\V, T))] \le \frac{N - N' + 1}{N} \cdot \sum_{j = 1}^{N' - 1} \frac{1}{N - j + 1}. \]
Moreover, for $(v, t) \in \text{supp}((\V, T))$, it holds that:
\[\mathbb{E}_{\V, T} [U_1(h^{-1}(v, t); (\V, T))] = \frac{N - N' + 1}{N} \cdot \sum_{j = 1}^{N' - 1} \frac{1}{N - j + 1}. \]
\end{lemma}
\begin{proof}[Proof of Lemma \ref{lemma:Ntypesanalysis}]
We first simplify the expected utility.
\begin{align*}
&\mathbb{E}_{\V, T} [U_1(h^{-1}(v, t); h^{-1}(\V, T)] \\
&= \frac{1}{|\mathcal{T}|}\sum_{t' \in \mathcal{T}} (\mathbb{P}_{\V, T}[v \ge V] \cdot \mathbbm{1}[t \le t'] + \mathbb{P}_{\V, T}[V > v, T > t'] \cdot \mathbbm{1}[t \le t']) - c(h^{-1}(v, t)) \\
&=_{(A)} \frac{|\left\{t' \in \mathcal{T} \mid t' \ge t \right\}|}{|\mathcal{T}|} \cdot \mathbb{P}_{\V, T}[V \le v] - \max(0, a_{t} \cdot v - 1) + \frac{1}{|\mathcal{T}|}\sum_{t' \in \mathcal{T} \mid t' \ge t} \mathbb{P}_{\V, T}[V > v, T > t'].
\end{align*}
where (A) uses Lemma \ref{lemma:costs}.

To analyze this expression, we apply Proposition \ref{prop:NtypesMdistribution} to obtain the reparameterized joint distribution $(\V, T)$ of $\muengagement(P, c, u, \mathcal{T})$. By the well-separatedness assumption on the type space, we see that the support of the distributions $\text{supp}(V \mid T \mid t = i)$ are disjoint and that $\max(\text{supp}(V)) = \frac{1}{a_{t_{N'}}} \cdot \left( 1+ \frac{N - N' + 1}{N} \cdot \left(1 - \sum_{j=1}^{N' - 1} \frac{1}{N- j + 1}\right)\right)$. 

We split the remainder of the analysis into several cases: (1) $t \le t_{N' - 1}$ and $\frac{1}{a_t} \le v < \frac{1}{a_t} \cdot (1 + \frac{1}{N})$, (2) $t \le t_{N' - 1}$, $v \ge \frac{1}{a_t} \cdot (1 + \frac{1}{N})$, and $v \le \max(\text{supp}(\V))$, (3) $t = t_{N'}$ and $\frac{1}{a_{t_{N'}}} \le v \le \max(\text{supp}(\V))$, (4) $t \le N'$ and $v \ge \max(\text{supp}(\V))$, and and $v \le \max(\text{supp}(\V))$, and (5) $t > t_{N'}$.

\paragraph{Case 1: $t \le t_{N' - 1}$ and $\frac{1}{a_t} \le v < \frac{1}{a_t} \cdot (1 + \frac{1}{N})$.} We know that $t = t_i$ for some $1 \le i \le N' - 1$. This means that for $(v', t') \in \text{supp}((V, T))$, if $t' > t$, then $v' > v$; moreover, if $t' < t$, then $v' < v$. This implies that 
\begin{align*}
&\mathbb{E}_{\V, T} [U_1(h^{-1}(v, t); (\V, T))] \\
&= \frac{|\left\{t' \in \mathcal{T} \mid t' \ge t_i \right\}|}{|\mathcal{T}|} \cdot \mathbb{P}[V \le v]  - \max(0, a_{t} \cdot v - 1) + \frac{1}{|\mathcal{T}|}\sum_{t' \in \mathcal{T} \mid t' \ge t_i} \mathbb{P}[V > v, T > t']  \\
&= \underbrace{\frac{N - i + 1}{N} \cdot \mathbb{P}[V \le v, T = t_i] - \max(0, a_{t} \cdot v - 1)}_{(A)} + \underbrace{\frac{N - i + 1}{N} \cdot \mathbb{P}[T < t_i] + \frac{1}{N} \sum_{t' \in \mathcal{T} \mid t' \ge t_i} \mathbb{P}[T > t']}_{(B)}.
\end{align*}
For term (A), we apply Proposition \ref{prop:NtypesMdistribution} to see that:
\begin{align*}
 &\frac{N - i + 1}{N} \cdot \mathbb{P}[V \le v, T = t_i] - \max(0, a_{t} \cdot v - 1) \\
 &= \frac{N - i + 1}{N} \cdot \mathbb{P}[V \le v \mid T = t_i] \cdot \mathbb{P}[T = t_{i}] - \max(0, a_{t} \cdot v - 1) \\
 &= \frac{(N - i + 1) \cdot N \cdot (a_{t_i} \cdot v - 1)}{N \cdot (N-i + 1)} - (a_{t} \cdot v - 1) \\ &= 0.  
\end{align*}

For term (B), we apply Proposition \ref{prop:NtypesMdistribution} to see that:
\begin{align*}
 &\frac{N - i + 1}{N} \cdot \mathbb{P}[T < t_i] + \frac{1}{|\mathcal{T}|}\sum_{t' \in \mathcal{T} \mid t' \ge t_i} \mathbb{P}[T > t'] \\
 &= 
 \frac{N - i + 1}{N} \cdot \mathbb{P}[T < t_i] + \frac{1}{N}\sum_{t' \in \mathcal{T} \mid t' \ge t_i} (1 - \mathbb{P}[T \le t']) \\  
 &= \frac{N - i + 1}{N} \cdot \sum_{j = 1}^{i-1} \frac{1}{N - j + 1} + \frac{1}{N} \sum_{j = i}^{N'-1} \left(1 - \sum_{j' = 1}^{j}  \frac{1}{N - j' + 1} \right) \\
 &= \frac{1}{N} \left(\sum_{j = 1}^{i-1} \frac{N - i + 1}{N - j + 1} - \sum_{j' = 1}^{i-1}  \frac{N' - i}{N - j' + 1} + N' - i - \sum_{j' = i}^{N' - 1}  \frac{N' - j'}{N - j' + 1} \right) \\
  &= \frac{1}{N} \left(\sum_{j = 1}^{i-1} \frac{N - N' + 1}{N - j + 1} + \sum_{j' = i}^{N' - 1}  \left( 1- \frac{N' - j'}{N - j' + 1} \right) \right) \\
    &= \frac{1}{N} \left(\sum_{j = 1}^{i-1} \frac{N - N' + 1 }{N - j + 1} + \sum_{j' = i}^{N' - 1}  \frac{N - N' + 1}{N - j' + 1} \right) \\
    &= \frac{N - N' + 1}{N} \cdot \sum_{j = 1}^{N' - 1} \frac{1}{N - j + 1}. 
    \end{align*}

Altogether, this proves that:
\[\mathbb{E}_{\V, T} [U_1(h^{-1}(v, t); (\V, T))] = \frac{N - N' + 1}{N} \cdot \sum_{j=1}^{N' - 1} \frac{1}{N- j + 1}\]
as desired. 

\paragraph{Case 2: $t \le t_{N' - 1}$, $v \ge \frac{1}{a_t} \cdot (1 + \frac{1}{N})$, and $v \le \max(\text{supp}(\V))$.} We know that $t = t_i$ for some $1 \le i \le N' - 1$. Let $i^* \in [i, N']$ be the \textit{maximum} value such that $v > \frac{1}{a_{t_{i^*}}} \cdot (1 + \frac{1}{N})$. (We immediately see that $i^* \neq N'$, because $v \le \max(\text{supp}(\V)) \le \frac{1}{a_{t_{N'}}} \cdot (1 + \frac{1}{N})$.) This means that for $(v', t') \in \text{supp}((V, T))$, if $t' > t_{i^*}$, then $v' > v$; moreover, if $t' < t_{i^*}$, then $v' < v$. This implies that 
\begin{align*}
&\mathbb{E}_{\V, T} [U_1(h^{-1}(v, t); (\V, T))] \\
&= \frac{|\left\{t' \in \mathcal{T} \mid t' \ge t_i \right\}|}{|\mathcal{T}|} \cdot \mathbb{P}_{\V, T}[V \le v]  - \max(0, a_{t} \cdot v - 1) + \frac{1}{|\mathcal{T}|}\sum_{t' \in \mathcal{T} \mid t' \ge t_i} \mathbb{P}_{\V, T}[V > v, T > t']  \\
&= \underbrace{\frac{N - i + 1}{N} \cdot \left(\mathbb{P}[V \le v, T = t_{i^* + 1}] + \sum_{j=i}^{i^*} \mathbb{P}[T = t_{j}] \right)  - \frac{1}{N} \sum_{t' \in \mathcal{T} \mid t' \ge t_i} \mathbb{P}[T > t', V \le v] - \max(0, a_{t} \cdot v - 1)}_{(A)}  \\ 
&+ \underbrace{\frac{N - i + 1}{N} \cdot \mathbb{P}[T < t_i] + \frac{1}{N} \sum_{t' \in \mathcal{T} \mid t' \ge t_i} \mathbb{P}[T > t']}_{(B)}.
\end{align*}
For term (A), we apply Proposition \ref{prop:NtypesMdistribution} to see that:
\begin{align*}
&\frac{N - i + 1}{N} \cdot \left(\mathbb{P}[V \le v, T = t_{i^* + 1}]  + \sum_{j=i}^{i^*} \mathbb{P}[T = t_{j}] \right)  - \frac{1}{N} \sum_{t' \in \mathcal{T} \mid t' \ge t_i} \mathbb{P}[T > t', V \le v] - \max(0, a_{t} \cdot v - 1) \\
&= \frac{N - i + 1}{N} \cdot \left( \mathbb{P}[V \le v, T = t_{i^* + 1}] + \sum_{j=i}^{i^*} \mathbb{P}[T = t_{j}]  \right) \\
&
-\frac{1}{N} \sum_{t' \in \mathcal{T} \mid t' \ge t_i} \left(\mathbb{P}[V \le v, T = t_{i^*+ 1}, t_{i^* + 1} > t'] + \sum_{j=i}^{i^*} \mathbb{P}[T = t_{j}, t_j > t']  \right)
 - (a_{t} \cdot v - 1) \\
 &= \frac{N - i + 1}{N} \cdot \left(\mathbb{P}[V \le v, T = t_{i^* + 1}] + \sum_{j=i}^{i^*} \mathbb{P}[T = t_{j}]  \right) \\
 &- \frac{i^* - i + 1}{N} \cdot \mathbb{P}[V \le v, T = t_{i^* + 1}] - \sum_{j=i}^{i^*}  \frac{j - i}{N}  \cdot \mathbb{P}[T = t_{j}]  
 - (a_{t} \cdot v - 1) \\
  &= \frac{N - i^*}{N} \cdot \mathbb{P}[V \le v, T = t_{i^* + 1}] + \sum_{j=i}^{i^*} \frac{N - j + 1}{N} \cdot \mathbb{P}[T = t_{j}] 
 - (a_{t} \cdot v - 1) \\
 &\le \frac{N - i^*}{N} \cdot \mathbb{P}[T = t_{i^* + 1}] \cdot \mathbb{P}[V \le v | T = t_{i^* + 1}]  + \frac{i^* - i + 1}{N}
 - (a_{t} \cdot v - 1) \\
  &\le_{(1)}   (a_{t_{i^* + 1}} \cdot v -1) + \frac{i^* - i + 1}{N}
 - (a_{t} \cdot v - 1) \\
 &= - (a_t - a_{t_{i^* + 1}}) \cdot v + \frac{i^* - i +1}{N} \\
 &\le - (a_t - a_{t_{i^* + 1}}) \cdot \frac{1}{a_{t_{i^*}}} \cdot (1 + \frac{1}{N}) + \frac{i^* - i +1}{N} \\
  &\le -\frac{a_{t_i}}{a_{t_{i^*}}} \cdot (1 + \frac{1}{N}) + \frac{a_{t_{i^* + 1}}}{a_{t_{i^*}}} \cdot (1 + \frac{1}{N}) + \frac{i^* - i +1}{N} \\
  &\le -\left(1 + \frac{1}{N}\right)^{i^* - i + 1} + 1 + \frac{i^* - i +1}{N} \\
  &\le_{(2)} 0
\end{align*}
where (1) uses the fact that 
\[\mathbb{P}[V \le v \mid T = t_{i^*+1}] \le \min\left(N \cdot \frac{1}{(N -i^*) \alpha^{i^*+1}} \cdot (a_{t_{i^* + 1}} \cdot v -1), 1 \right) \le N \cdot \frac{1}{(N -i^*) \alpha^{i^* + 1}} \cdot (a_{t_{i^* + 1}} \cdot v -1). \]
where (2) uses the fact that for $m \ge 1$ and $x \ge 0$, it holds that $(1 + x)^m \ge 1 + m \cdot x$. 

For term (B), we apply the same argument as in Case 1 to see that it is equal to $\frac{N - N' + 1}{N} \cdot \sum_{j = 1}^{N' - 1} \frac{1}{N - j + 1}$. 

Altogether, this proves that:
\[ \mathbb{E}_{\V, T} [U_1(h^{-1}(v, t); (\V, T))] \le \frac{N - N' + 1}{N} \cdot \sum_{j=1}^{N' - 1} \frac{1}{N- j + 1}\]
as desired. 

\paragraph{Case 3: $t = t_{N'}$ and $\frac{1}{a_{t_{N'}}} \le v \le \max(\text{supp}(\V))$. } For $(v', t') \in \text{supp}((V, T))$, if $t' < t$, then $v' < v$. Moreover, we see that $\max(\text{supp}(\V)) = \frac{1}{a_{t_{N'}}} \cdot \left( 1+ \frac{N - N' + 1}{N} \cdot \left(1 - \sum_{j=1}^{N' - 1} \frac{1}{N- j + 1}\right)\right)$. This implies that 
\begin{align*}
&\mathbb{E}_{\V, T} [U_1(h^{-1}(v, t); (\V, T))] \\
&= \frac{|\left\{t' \in \mathcal{T} \mid t' \ge t_{N'} \right\}|}{|\mathcal{T}|} \cdot \mathbb{P}[V \le v]  - \max(0, a_{t} \cdot v - 1) + \frac{1}{|\mathcal{T}|}\sum_{t' \in \mathcal{T} \mid t' \ge t_{N'}} \mathbb{P}[V > v, T > t']  \\
&= \frac{|\left\{t' \in \mathcal{T} \mid t' \ge t_{N'} \right\}|}{|\mathcal{T}|} \cdot \mathbb{P}[V \le v]  - \max(0, a_{t} \cdot v - 1) \\
&= \underbrace{\frac{N - N' + 1}{N} \cdot \mathbb{P}[V \le v, T = t_{N'}] - \max(0, a_{t} \cdot v - 1)}_{(A)} + \underbrace{\frac{N - N' + 1}{N} \cdot \mathbb{P}[T < t_{N'}]}_{(B)}.
\end{align*}
For term (A), we apply Proposition \ref{prop:NtypesMdistribution} to see that:
\begin{align*}
 &\frac{N - N' + 1}{N} \cdot \mathbb{P}[V \le v, T = t_{N'}] - \max(0, a_{t} \cdot v - 1) \\
 &=  \frac{N - N' + 1}{N} \cdot \mathbb{P}[V \le v \mid T = t_{N'}] \cdot \mathbb{P}[T = t_{N'}] \\
 &= \frac{N - N' + 1}{N} \cdot \frac{N}{N-N'+1} \cdot \left(1 - \sum_{j=1}^{N'-1} \frac{1}{N-j+1} \right)^{-1} \cdot (a_{t_i} \cdot v - 1) \cdot \left(1 - \sum_{j=1}^{N'-1} \frac{1}{N-j+1} \right)- (a_{t} \cdot v - 1) \\
 &= 0.    
\end{align*}

For term (B), we apply Proposition \ref{prop:NtypesMdistribution} to see that:
\begin{align*}
\frac{N - N' + 1}{N} \cdot \mathbb{P}[T < t_{N'}]
=  \frac{N - N' + 1}{N} \cdot \sum_{j=1}^{N' - 1} \mathbb{P}[T  = t_{j}] = \frac{N - N' + 1}{N} \cdot \sum_{j=1}^{N' - 1} \frac{1}{N- j + 1}. 
\end{align*}

Altogether, this proves that:
\[ \mathbb{E}_{\V, T} [U_1(h^{-1}(v, t); (\V, T))] = \frac{N - N' + 1}{N} \cdot \sum_{j=1}^{N' - 1} \frac{1}{N- j + 1}\]
as desired.

\paragraph{Case 4: $t \le N'$ and $v \ge \max(\text{supp}(\V))$.} We know that $t = t_i$ for some $1 \le i \le N$. For $(v', t') \in \text{supp}((V, T))$, we see that $v' < v$. This implies that 
\begin{align*}
&\mathbb{E}_{\V, T} [U_1(h^{-1}(v, t); (\V, T))] \\
&= \frac{|\left\{t' \in \mathcal{T} \mid t' \ge t_i \right\}|}{|\mathcal{T}|} \cdot \mathbb{P}[V \le v]  - \max(0, a_{t} \cdot v - 1) + \frac{1}{|\mathcal{T}|}\sum_{t' \in \mathcal{T} \mid t' \ge t_{N'}} \mathbb{P}[V > v, T > t']  \\
&= \frac{N - i + 1}{N}  - (a_{t_i} \cdot v - 1).
\end{align*}
This expression is upper bounded by the case where $v = \max(\text{supp}(\V))$. If $1 \le i \le N' -1$, we can thus apply Case 2; if $i = N'$, we can apply Case 3. Altogether, this proves that:
\[ \mathbb{E}_{\V, T} [U_1(h^{-1}(v, t); (\V, T))] \le \frac{N - N' + 1}{N} \cdot \sum_{j=1}^{N' - 1} \frac{1}{N- j + 1}\]
as desired.

\paragraph{Case 5: $t > t_{N'}$.} In this case, $(v, t) \in S$ satisfies 
\[v \ge \frac{1}{a_{N' + 1}} \ge \frac{1}{a_{t_{N'}}} \cdot \left( 1 + \frac{1}{N} \right) \ge \frac{1}{a_{t_{N'}}} \cdot \left( 1+ \frac{N - N' + 1}{N} \cdot \left(1 - \sum_{j=1}^{N' - 1} \frac{1}{N- j + 1}\right)\right) = \max(\text{supp}(V)). \]
This means that 
\[\mathbb{E}_{\V, T} [U_1(h^{-1}(v, t); (\V, T))] =  \frac{|\left\{t' \in \mathcal{T} \mid t' \ge t \right\}|}{|\mathcal{T}|} - \max(0, a_{t} \cdot v - 1) \le \frac{|\left\{t' \in \mathcal{T} \mid t' \ge t \right\}|}{|\mathcal{T}|} \le \frac{N - N'}{N}.\]
It suffices to show that $\frac{N - N'}{N} \le \frac{N - N' + 1}{N} \cdot \sum_{j = 1}^{N' - 1} \frac{1}{N - j + 1}$, which we can be written as:
\[\frac{N- N'}{N- N' + 1} = 1 - \frac{1}{N-N'+1} \le \sum_{j = 1}^{N' - 1} \frac{1}{N - j + 1},\]
we know is true by the definition of $N'$. 

\end{proof}

Using Lemma \ref{lemma:Ntypesanalysis}, we prove Theorem \ref{thm:Ntypes}. 
\begin{proof}[Proof of Theorem \ref{thm:Ntypes}]
We first claim that we can work over the reparameterized space $(\V, T)$ defined in \Cref{subsec:tools}.
As described above, every $w \in \text{supp}(\muengagement(P, c, u, \mathcal{T}))$ is in the 
image $\cup_{t \in \mathcal{T}} \Ct$, which means that it is associated with a unique value $h^{-1}(w) = (v, t) \in S$. It thus suffices to show that there exists a best response $w \in \mathbb{R}_{\ge 0}^2$ to:
\[\argmax_{w \in \mathbb{R}_{\ge 0}^2} \mathbb{E}_{\textbf{w}_{-i} \sim (\muengagement(P, c, u, \mathcal{T}))^{P-1}} [U_i(w; \textbf{w}_{-i})]  \]
that is also in the image $\cup_{t \in \mathcal{T}} \Ct$; this follows from Lemma \ref{lemma:bestresponsestronger}. 

We thus work over the reparameterized space for the remainder of the analysis. By Lemma \ref{lemma:Ntypesanalysis}, we see that for any $(v, t) \in S$, the expected utility satisfies:
\[\mathbb{E}_{\V, T} [U_1(h^{-1}(v, t); (\V, T))] \le \frac{N - N' + 1}{N} \cdot \sum_{j = 1}^{N' - 1} \frac{1}{N - j + 1}. \]
Moreover, for $(v, t) \in \text{supp}((\V, T))$, it holds that:
\[\mathbb{E}_{\V, T} [U_1(h^{-1}(v, t); (\V, T))] = \frac{N - N' + 1}{N} \cdot \sum_{j = 1}^{N' - 1} \frac{1}{N - j + 1}. \]
This proves that $(\V, T)$ is an equilibrium in the reparameterized space $S$.

Putting this all together, this proves that  $\muengagement(P, c, u, \mathcal{T})$ is an equilibrium in the original space $\mathbb{R}_{\ge 0}^2$. 
\end{proof}

We prove \Cref{prop:2typeswellseparated}.
\begin{proof}[Proof of \Cref{prop:2typeswellseparated}]
We apply \Cref{thm:Ntypes}. It suffices to show that $N' = 2$. To see this, since $N = 2$, we see that $\sum_{i=1}^{m} \frac{1}{N - i + 1}$ is equal to $\frac{1}{2}$ if $m = 1$ and   $\sum_{i=1}^{m} \frac{1}{N - i + 1}$ is equal to $1 + \frac{1}{2} \ge 1$ if $m = 2$. This means that $N' = 2$ as desired. 
\end{proof}

\subsection{Proof of Theorem \ref{thm:2types}}

Before diving into the proof of \Cref{thm:2types},  we first verify that the distributions in \Cref{def:2types} are indeed well-defined: in particular, it suffices to verify that the ordering of $v$ values at the boundary points indeed proceeds in the order shown in \Cref{fig:2types} and that the total density of $g$ is $1$. We split into the three cases in \Cref{def:2types}. 

\paragraph{Case 1: $a_{t_1} / a_{t_2} \ge 1.5$.}
We first show that: 
\[\frac{1}{a_{t_1}} \le_{(1)} \frac{3}{2 \cdot a_{t_1}}  \le_{(2)} \frac{1}{a_{t_2}} \le_{(3)} \frac{5}{4 \cdot a_{t_2}}. \]
Inequalities (1) and (3) are trivial, and inequality (2) follows from the fact that  $a_{t_1} / a_{t_2} \ge 1.5$.

We next observe that:
\begin{align*}
\int_{v} g(v) dv &= a_{t_1} \left(\frac{3}{2 \cdot a_{t_1}}  - \frac{1}{a_{t_1}}  \right) + 2 a_{t_2} \cdot \left(\frac{5}{4 \cdot a_{t_2}} - \frac{1}{a_{t_2}} \right) \\
&= 0.5 + 0.5 \\
&= 1.
\end{align*}

\paragraph{Case 2: $(5 - \sqrt{5})/2 \le a_{t_1} / a_{t_2} \le 1.5$.}
We first show that: 
\[\frac{1}{a_{t_1}} \le_{(1)} \frac{1}{a_{t_2}} \le_{(1)} \frac{1}{2 a_{t_2} \left(\frac{a_{t_1}}{a_{t_2}} - 1 \right)} \le_{(3)} \frac{1}{a_{t_2}} \cdot \left(2 - \frac{a_{t_1}}{2 \cdot a_{t_2}}\right). \]
Inequality (1) follows from the fact that $a_{t_1} > a_{t_2}$. Inequality (2) follows from the fact that $2 \left(\frac{a_{t_1}}{a_{t_2}} - 1 \right) \le 1$. Inequality (3) can be rewritten as:
\[ 2 \left(\frac{a_{t_1}}{a_{t_2}} - 1 \right)  \left(2 - \frac{a_{t_1}}{2 \cdot a_{t_2}}\right) \ge 1,\]
which follows from the fact that $(5 - \sqrt{5})/2 \le a_{t_1} / a_{t_2} \le 1.5$.  

We next observe that:
\begin{align*}
\int_{v} g(v) dv &=
a_{t_1} \left(\frac{1}{a_{t_2}}   - \frac{1}{a_{t_1}}  \right) + 2 a_{t_2} \cdot \left(\frac{1}{a_{t_2}} \cdot \left(2 - \frac{a_{t_1}}{2 \cdot a_{t_2}}\right) - \frac{1}{a_{t_2}} \right) \\
&= \frac{a_{t_1}}{a_{t_2}}   - 1  + 4 - \frac{a_{t_1}}{\cdot a_{t_2}} - 2 \\
&= 1.
\end{align*}

\paragraph{Case 3: $1 \le a_{t_1} / a_{t_2} \le (5 - \sqrt{5})/2$.}
We first show that: 
\[\frac{1}{a_{t_1}} \le_{(1)} \frac{1}{a_{t_2}} \le_{(2)} \frac{3 - \frac{a_{t_1}}{a_{t_2}} }{2 a_{t_2} \left(2 - \frac{a_{t_1}}{a_{t_2}} \right) }  \le_{(3)} \frac{1}{a_{t_1}} + \left( \frac{1}{a_{t_1}} - \frac{1}{2 a_{t_2}} \right) \left(\frac{3 - \frac{a_{t_1}}{a_{t_2}}}{2 - \frac{a_{t_1}}{a_{t_2}} }\right)  \]
Inequality (1) follows from the fact that $a_{t_1} > a_{t_2}$. Inequality (2) can be written as:
\[2 \left(2 - \frac{a_{t_1}}{a_{t_2}} \right) \le  3 - \frac{a_{t_1}}{a_{t_2}}, \]
which can be written as:
\[1 \le \frac{a_{t_1}}{a_{t_2}}, \]
which holds because $a_{t_1} > a_{t_2}$. Inequality (3) can be written as:
\[ \left( \frac{1}{a_{t_2}} - \frac{1}{a_{t_1}} \right) \left(\frac{3 - \frac{a_{t_1}}{a_{t_2}}}{2 - \frac{a_{t_1}}{a_{t_2}} }\right) \le  \frac{1}{a_{t_1}}, \]
which can be rewritten as: 
\[ \left( \frac{a_{t_1}}{a_{t_2}} - 1 \right) \left(3 - \frac{a_{t_1}}{a_{t_2}}\right) \le  2 - \frac{a_{t_1}}{a_{t_2}} , \]
which follows from the fact that $1 \le a_{t_1} / a_{t_2} \le (5 - \sqrt{5})/2$.  

We next observe that:
\begin{align*}
&\int_{v} g(v) dv \\
&= a_{t_1} \left(\frac{1}{a_{t_2}}   - \frac{1}{a_{t_1}}  \right) + 2 a_{t_2} \cdot \left( \frac{3 - \frac{a_{t_1}}{a_{t_2}} }{2 a_{t_2} \left(2 - \frac{a_{t_1}}{a_{t_2}} \right) }  - \frac{1}{a_{t_2}} \right) \\
&+ a_{t_1} \cdot \left(\frac{1}{a_{t_1}} + \left( \frac{1}{a_{t_1}} - \frac{1}{2 a_{t_2}} \right) \left(\frac{3 - \frac{a_{t_1}}{a_{t_2}}}{2 - \frac{a_{t_1}}{a_{t_2}} }\right)  - \frac{3 - \frac{a_{t_1}}{a_{t_2}} }{2 a_{t_2} \left(2 - \frac{a_{t_1}}{a_{t_2}} \right) } \right) \\
&=  \frac{a_{t_1} }{a_{t_2}} - 1 + \frac{3 - \frac{a_{t_1}}{a_{t_2}} }{\left(2 - \frac{a_{t_1}}{a_{t_2}} \right) }  - 2  + 1 + \left( 1 - \frac{a_{t_1}}{2 a_{t_2}} \right) \left(\frac{3 - \frac{a_{t_1}}{a_{t_2}}}{2 - \frac{a_{t_1}}{a_{t_2}} }\right)  - \frac{a_{t_1}}{2 a_{t_2} } \cdot \frac{3 - \frac{a_{t_1}}{a_{t_2}} }{\left(2 - \frac{a_{t_1}}{a_{t_2}} \right) }  \\
&=  \frac{a_{t_1} }{a_{t_2}} - 2 + \left( 2 - \frac{a_{t_1}}{a_{t_2}} \right) \left(\frac{3 - \frac{a_{t_1}}{a_{t_2}}}{2 - \frac{a_{t_1}}{a_{t_2}} }\right) \\
&=  \frac{a_{t_1} }{a_{t_2}} - 2 + 3 - \frac{a_{t_1}}{a_{t_2}}  \\
&= 1.
\end{align*}

Now, we turn to proving \Cref{thm:2types}. Like in the proof of \Cref{thm:Ntypes}, the main ingredient is computing the expected utility for different choices of $(v, t) \in S$ when the other creator's actions are selected according to the distribution $(\V, T) = h^{-1}(W)$ for $W \sim \muengagement(P, c, u, \mathcal{T})$. Since we have already analyzed Case 1 ($a_{t_1}/a_{t_2} \ge 1.5$) as a consequence of \Cref{thm:Ntypes}, we can focus on Cases 2 and 3. We analyze the expected utility separately for Case 2 and Case 3 in the following lemmas. 
\begin{lemma}
\label{lemma:case2}
Consider the setup of Definition \ref{def:2types}, and and assume that $(5 - \sqrt{5})/2 \le a_{t_1} / a_{t_2} \le 1.5$. Let $(\V, T) = h^{-1}(W)$ be the reparameterized distribution of $W \sim \muengagement(P, c, u, \mathcal{T})$. For any $(v, t) \in S$, the expected utility satisfies:
\[\mathbb{E}_{\V, T} [U_1(h^{-1}(v, t); (\V, T))] \le \frac{a_{t_1}}{2 \cdot a_{t_2}} - \frac{1}{2}. \]
Moreover, for $(v, t) \in \text{supp}((\V, T))$, it holds that:
\[\mathbb{E}_{\V, T} [U_1(h^{-1}(v, t); (\V, T))] = \frac{a_{t_1}}{2 \cdot a_{t_2}} - \frac{1}{2}. \]
\end{lemma}
\begin{proof}[Proof of Lemma \ref{lemma:case2}]
We first simplify the expected utility.
\begin{align*}
&\mathbb{E}_{\V, T} [U_1(h^{-1}(v, t); h^{-1}(\V, T)] \\
&= \mathbb{P}_{\V, T}[v \ge V] \cdot \mathbbm{1}[t = t_1] + \frac{1}{2} \cdot  \mathbb{P}_{\V, T}[V > v, T = t_2] \cdot \mathbbm{1}[t = t_1] + \frac{1}{2} \cdot  \mathbb{P}_{\V, T}[v \ge V] \cdot \mathbbm{1}[t = t_2] - c(h^{-1}(v, t)) \\
&=_{(1)} \left(\mathbbm{1}[t = t_1] + \mathbbm{1}[t = t_2] \cdot \frac{1}{2} \right)  \cdot \mathbb{P}_{\V, T}[V \le v] - \max(0, a_{t} \cdot v - 1) + \frac{1}{2} \cdot  \mathbb{P}_{\V, T}[V > v, T = t_2] \cdot \mathbbm{1}[t = t_1].
\end{align*}
where (1) uses Lemma \ref{lemma:costs}.

We split the remainder of the analysis into several cases: (A) $v \le \frac{1}{a_{t_2}}$ and $t = t_1$, (B), $\frac{1}{a_{t_2}}  \le v \le \frac{1}{2 a_{t_2} \left(\frac{a_{t_1}}{a_{t_2}} - 1 \right)}$ and $t = t_1$, (C) $\frac{1}{2 a_{t_2} \left(\frac{a_{t_1}}{a_{t_2}} - 1 \right)} \le v \le \frac{1}{a_{t_2}} \cdot \left(2 - \frac{a_{t_1}}{2 \cdot a_{t_2}}\right)$ and $t = t_1$, (D) $\frac{1}{a_{t_2}}  \le v \le \frac{1}{a_{t_2}} \cdot \left(2 - \frac{a_{t_1}}{2 \cdot a_{t_2}}\right)$ and $t = t_2$, and (E) $v \ge \frac{1}{a_{t_2}} \cdot \left(2 - \frac{a_{t_1}}{2 \cdot a_{t_2}}\right)$.

\paragraph{Case A: $v \le \frac{1}{a_{t_2}}$ and $t = t_1$.} We observe that:
\begin{align*}
&\mathbb{E}_{\V, T} [U_1(h^{-1}(v, t); h^{-1}(\V, T)] \\
&= \left(\mathbbm{1}[t = t_1] + \mathbbm{1}[t = t_2] \cdot \frac{1}{2} \right)  \cdot \mathbb{P}_{\V, T}[V \le v] - \max(0, a_{t} \cdot v - 1) + \frac{1}{2} \cdot  \mathbb{P}_{\V, T}[V > v, T = t_2] \cdot \mathbbm{1}[t = t_1] \\
&= \underbrace{\mathbb{P}_{\V, T}[V \le v] - \max(0, a_{t} \cdot v - 1)}_{(1)} + \underbrace{\frac{1}{2} \cdot  \mathbb{P}_{\V, T}[T = t_2]}_{(2)}. 
\end{align*}
For term (1), we observe that:
\[\mathbb{P}_{\V, T}[V \le v] - \max(0, a_{t} \cdot v - 1) = a_{t}  \left(v - \frac{1}{a_{t}} \right) - (a_{t} \cdot v - 1) = 0. \]
For term (2), we apply Lemma \ref{lemma:property2types} to see that 
\begin{align*}
\frac{1}{2} \cdot  \mathbb{P}[T = t_2] 
&= \frac{a_{t_1}}{2 \cdot a_{t_2}} - \frac{1}{2}
\end{align*}
as desired. 

Putting this all together, we see that:
\[\mathbb{E}_{\V, T} [U_1(h^{-1}(v, t); h^{-1}(\V, T)] = \frac{a_{t_1}}{2 \cdot a_{t_2}} - \frac{1}{2}. \]

\paragraph{Case B: $\frac{1}{a_{t_2}}  \le v \le \frac{1}{2 a_{t_2} \left(\frac{a_{t_1}}{a_{t_2}} - 1 \right)}$ and $t = t_1$. } We observe that:
\begin{align*}
&\mathbb{E}_{\V, T} [U_1(h^{-1}(v, t); h^{-1}(\V, T)] \\
&= \left(\mathbbm{1}[t = t_1] + \mathbbm{1}[t = t_2] \cdot \frac{1}{2} \right)  \cdot \mathbb{P}_{\V, T}[V \le v] - \max(0, a_{t} \cdot v - 1) + \frac{1}{2} \cdot  \mathbb{P}_{\V, T}[V > v, T = t_2] \cdot \mathbbm{1}[t = t_1] \\
&= \mathbb{P}_{\V, T}[V \le v] - \max(0, a_{t} \cdot v - 1) + \frac{1}{2} \cdot  \mathbb{P}_{\V, T}[V > v, T = t_2]  \\
&= \mathbb{P}_{\V, T}[V \le v] -\frac{1}{2} \cdot  \mathbb{P}_{\V, T}[V \le v, T = t_2]  - \max(0, a_{t} \cdot v - 1) + \frac{1}{2} \cdot  \mathbb{P}_{\V, T}[T = t_2] \\
&= \underbrace{\mathbb{P}_{\V, T}[V \le v] -\frac{1}{2} \cdot  \mathbb{P}_{\V, T}[V \le v, T = t_2] - \max(0, a_{t} \cdot v - 1)}_{(1)} + \underbrace{\frac{1}{2} \cdot  \mathbb{P}_{\V, T}[T = t_2]}_{(2)}. 
\end{align*}
For term (1), we observe that:
\begin{align*}
&\mathbb{P}_{\V, T}[V \le v] -\frac{1}{2} \cdot  \mathbb{P}_{\V, T}[V \le v, T = t_2] - \max(0, a_{t} \cdot v - 1) \\
&= \mathbb{P}_{\V, T}[V \le 1/a_{t_2}] + \mathbb{P}_{\V, T}[1/a_{t_2} \le V \le v] \left(1 - \frac{1}{2} \cdot  \mathbb{P}_{\V, T}[T = t_2 \mid 1/a_{t_2} \le V \le v] \right) - (a_{t} \cdot v - 1) \\
&= a_{t_1} \left(\frac{1}{a_{t_2}} -  \frac{1}{a_{t_1}}\right) + 2  a_{t_2} \cdot \left(v - \frac{1}{a_{t_2}} \right) \cdot \left(1 - \frac{1}{2} \left(2 - \frac{a_{t_1}}{a_{t_2}} \right) \right)  -  (a_{t_1} \cdot v - 1) \\
&= a_{t_1} \left(\frac{1}{a_{t_2}} -  \frac{1}{a_{t_1}}\right) +  a_{t_1} \cdot \left(v - \frac{1}{a_{t_2}} \right)-  (a_{t_1} \cdot v - 1) \\
&= 0.
\end{align*}
For term (2), we apply Lemma \ref{lemma:property2types} to see that 
\begin{align*}
\frac{1}{2} \cdot  \mathbb{P}[T = t_2] 
&= \frac{a_{t_1}}{2 \cdot a_{t_2}} - \frac{1}{2}
\end{align*}
as desired.

Putting this all together, we see that:
\[\mathbb{E}_{\V, T} [U_1(h^{-1}(v, t); h^{-1}(\V, T)] = \frac{a_{t_1}}{2 \cdot a_{t_2}} - \frac{1}{2}. \]

\paragraph{Case C: $\frac{1}{2 a_{t_2} \left(\frac{a_{t_1}}{a_{t_2}} - 1 \right)} \le v \le \frac{1}{a_{t_2}} \cdot \left(2 - \frac{a_{t_1}}{2 \cdot a_{t_2}}\right)$ and $t = t_1$.} We use the same analysis as in Case B to see that: 
\begin{align*}
&\mathbb{E}_{\V, T} [U_1(h^{-1}(v, t); h^{-1}(\V, T)] \\
&= \underbrace{\mathbb{P}_{\V, T}[V \le v] -\frac{1}{2} \cdot  \mathbb{P}_{\V, T}[V \le v, T = t_2] - \max(0, a_{t} \cdot v - 1)}_{(1)} + \underbrace{\frac{1}{2} \cdot  \mathbb{P}_{\V, T}[T = t_2]}_{(2)}. 
\end{align*}
For term (1), we observe that:
\begin{align*}
&\mathbb{P}_{\V, T}[V \le v] -\frac{1}{2} \cdot  \mathbb{P}_{\V, T}[V \le v, T = t_2] - \max(0, a_{t} \cdot v - 1) \\
&= \mathbb{P}_{\V, T}[V \le 1/a_{t_2}] \\
&+ \mathbb{P}_{\V, T}\left[1/a_{t_2} \le V \le \frac{1}{2 a_{t_2} \left(\frac{a_{t_1}}{a_{t_2}} - 1 \right)}\right] \left(1 - \frac{1}{2} \cdot  \mathbb{P}_{\V, T}\left[T = t_2 \mid 1/a_{t_2} \le V \le \frac{1}{2 a_{t_2} \left(\frac{a_{t_1}}{a_{t_2}} - 1 \right)}\right] \right) \\
&+ \mathbb{P}_{\V, T}\left[\frac{1}{2 a_{t_2} \left(\frac{a_{t_1}}{a_{t_2}} - 1 \right)} \le V \le v\right] \left(1 - \frac{1}{2} \right) - (a_{t} \cdot v - 1) \\
&= a_{t_1} \left(\frac{1}{a_{t_2}} -  \frac{1}{a_{t_1}}\right) +  a_{t_1} \cdot \left(\frac{1}{2 a_{t_2} \left(\frac{a_{t_1}}{a_{t_2}} - 1 \right)} - \frac{1}{a_{t_2}} \right) +  2 a_{t_2} \cdot \frac{1}{2} \left(v - \frac{1}{2 a_{t_2} \left(\frac{a_{t_1}}{a_{t_2}} - 1 \right)} \right)  - (a_{t_1} \cdot v - 1) \\
&= -1 \cdot \left(a_{t_1} - a_{t_2} \right) \cdot \left(v - \frac{1}{2 a_{t_2} \left(\frac{a_{t_1}}{a_{t_2}} - 1 \right)} \right) \\
&\le 0.
\end{align*}
For term (2), we apply Lemma \ref{lemma:property2types} to see that 
\begin{align*}
\frac{1}{2} \cdot  \mathbb{P}[T = t_2] 
&= \frac{a_{t_1}}{2 \cdot a_{t_2}} - \frac{1}{2}
\end{align*}
as desired. 

Putting this all together, we see that:
\[\mathbb{E}_{\V, T} [U_1(h^{-1}(v, t); h^{-1}(\V, T)] \le \frac{a_{t_1}}{2 \cdot a_{t_2}} - \frac{1}{2}. \]

\paragraph{Case D: $\frac{1}{a_{t_2}}  \le v \le \frac{1}{a_{t_2}} \cdot \left(2 - \frac{a_{t_1}}{2 \cdot a_{t_2}}\right)$ and $t = t_2$.} We observe that:
\begin{align*}
\mathbb{E}_{\V, T} [U_1(h^{-1}(v, t); h^{-1}(\V, T)] &= \frac{1}{2}  \cdot \mathbb{P}_{\V, T}[V \le v] - (a_{t_2} \cdot v - 1) \\
&= \frac{1}{2} \cdot  a_{t_1} \left(\frac{1}{a_{t_2}} - \frac{1}{a_{t_1}} \right) + \frac{1}{2} \cdot  2 a_{t_2} \cdot \left(v -  \frac{1}{a_{t_2}} \right)  - (a_{t_2} \cdot v - 1)\\
&= \frac{a_{t_1}}{2 \cdot a_{t_2}} - \frac{1}{2} + a_{t_2} \cdot \left(v -  \frac{1}{a_{t_2}} \right)  - (a_{t_2} \cdot v - 1)\\
&= \frac{a_{t_1}}{2 \cdot a_{t_2}} - \frac{1}{2},
\end{align*}
as desired. 

\paragraph{Case E: $v \ge \frac{1}{a_{t_2}} \cdot \left(2 - \frac{a_{t_1}}{2 \cdot a_{t_2}}\right)$.} We observe that: 
\begin{align*}
\mathbb{E}_{\V, T} [U_1(h^{-1}(v, t); h^{-1}(\V, T)] &= \left(\mathbbm{1}[t = t_1] + \frac{1}{2} \cdot \mathbbm{1}[t = t_2] \right) - \max(0, a_{t} \cdot v - 1). 
\end{align*}
For each value of $t$, since $\frac{1}{a_{t_2}} \cdot \left(2 - \frac{a_{t_1}}{2 \cdot a_{t_2}}\right) = \max(\text{supp}(\V))$, we see that the above expression is upper bounded by the case of $v = \frac{1}{a_{t_2}} \cdot \left(2 - \frac{a_{t_1}}{2 \cdot a_{t_2}}\right)$. If $t = t_1$, we can apply the analysis from Case C; if $t = t_2$, we can apply the analysis from Case E. Altogether, this means that: 
\[\mathbb{E}_{\V, T} [U_1(h^{-1}(v, t); h^{-1}(\V, T)] \le \frac{a_{t_1}}{2 \cdot a_{t_2}} - \frac{1}{2} \]
as desired. 

\end{proof}

\begin{lemma}
\label{lemma:case3}
Consider the setup of Definition \ref{def:2types}, and and assume that $1\le a_{t_1} / a_{t_2} \le (5 - \sqrt{5})/2 $. Let $(\V, T) = h^{-1}(W)$ be the reparameterized distribution of $W \sim \muengagement(P, c, u, \mathcal{T})$. For any $(v, t) \in S$, the expected utility satisfies:
\[\mathbb{E}_{\V, T} [U_1(h^{-1}(v, t); (\V, T))] \le \frac{a_{t_1}}{2 \cdot a_{t_2}} - \frac{1}{2}. \]
Moreover, for $(v, t) \in \text{supp}((\V, T))$, it holds that:
\[\mathbb{E}_{\V, T} [U_1(h^{-1}(v, t); (\V, T))] = \frac{a_{t_1}}{2 \cdot a_{t_2}} - \frac{1}{2}. \]
\end{lemma}
\begin{proof}[Proof of Lemma \ref{lemma:case3}]

  We first simplify the expected utility.
\begin{align*}
&\mathbb{E}_{\V, T} [U_1(h^{-1}(v, t); h^{-1}(\V, T)] \\
&= \mathbb{P}_{\V, T}[v \ge V] \cdot \mathbbm{1}[t = t_1] + \frac{1}{2} \cdot  \mathbb{P}_{\V, T}[V > v, T = t_2] \cdot \mathbbm{1}[t = t_1] + \frac{1}{2} \cdot  \mathbb{P}_{\V, T}[v \ge V] \cdot \mathbbm{1}[t = t_2] - c(h^{-1}(v, t)) \\
&=_{(1)} \left(\mathbbm{1}[t = t_1] + \mathbbm{1}[t = t_2] \cdot \frac{1}{2} \right)  \cdot \mathbb{P}_{\V, T}[V \le v] - \max(0, a_{t} \cdot v - 1) + \frac{1}{2} \cdot  \mathbb{P}_{\V, T}[V > v, T = t_2] \cdot \mathbbm{1}[t = t_1].
\end{align*}
where (1) uses Lemma \ref{lemma:costs}.

We split the remainder of the analysis into several cases: (A) $v \le \frac{1}{a_{t_2}}$ and $t = t_1$, (B), $\frac{1}{a_{t_2}}  \le v \le \frac{3 - \frac{a_{t_1}}{a_{t_2}} }{2 a_{t_2} \left(2 - \frac{a_{t_1}}{a_{t_2}} \right) } $ and $t = t_1$, (C) $\frac{3 - \frac{a_{t_1}}{a_{t_2}} }{2 a_{t_2} \left(2 - \frac{a_{t_1}}{a_{t_2}} \right) } \le v \le \frac{1}{a_{t_1}} + \left( \frac{1}{a_{t_1}} - \frac{1}{2 a_{t_2}} \right) \left(\frac{3 - \frac{a_{t_1}}{a_{t_2}}}{2 - \frac{a_{t_1}}{a_{t_2}} }\right)$ and $t = t_1$, (D) $\frac{1}{a_{t_2}}  \le v \le \frac{3 - \frac{a_{t_1}}{a_{t_2}} }{2 a_{t_2} \left(2 - \frac{a_{t_1}}{a_{t_2}} \right) }$ and $t = t_2$, (E) $\frac{3 - \frac{a_{t_1}}{a_{t_2}} }{2 a_{t_2} \left(2 - \frac{a_{t_1}}{a_{t_2}} \right) } \le v \le \frac{1}{a_{t_1}} + \left( \frac{1}{a_{t_1}} - \frac{1}{2 a_{t_2}} \right) \left(\frac{3 - \frac{a_{t_1}}{a_{t_2}}}{2 - \frac{a_{t_1}}{a_{t_2}} }\right)$ and $t = t_2$, and (F) $v \ge \max(\text{supp}(\V))$. 

\paragraph{Case A: $v \le \frac{1}{a_{t_2}}$ and $t = t_1$.} We observe that:
\begin{align*}
&\mathbb{E}_{\V, T} [U_1(h^{-1}(v, t); h^{-1}(\V, T)] \\
&= \left(\mathbbm{1}[t = t_1] + \mathbbm{1}[t = t_2] \cdot \frac{1}{2} \right)  \cdot \mathbb{P}_{\V, T}[V \le v] - \max(0, a_{t} \cdot v - 1) + \frac{1}{2} \cdot  \mathbb{P}_{\V, T}[V > v, T = t_2] \cdot \mathbbm{1}[t = t_1] \\
&= \underbrace{\mathbb{P}_{\V, T}[V \le v] - \max(0, a_{t} \cdot v - 1)}_{(1)} + \underbrace{\frac{1}{2} \cdot  \mathbb{P}_{\V, T}[T = t_2]}_{(2)} \\
\end{align*}
For term (1), we see that:
\[\mathbb{P}_{\V, T}[V \le v] - \max(0, a_{t} \cdot v - 1) = a_{t_1} \left( v - \frac{1}{a_{t_1}}\right) - (a_{t_1} \cdot v - 1) = 0. \]
For term (2), we apply Lemma \ref{lemma:property2types} to see that it is equal to $\frac{a_{t_1}}{2 \cdot a_{t_2}} - \frac{1}{2}$. Putting this all together, we see that:
\[\mathbb{E}_{\V, T} [U_1(h^{-1}(v, t); h^{-1}(\V, T)] = \frac{a_{t_1}}{2 \cdot a_{t_2}} - \frac{1}{2}.  \]

\paragraph{Case B: $\frac{1}{a_{t_2}}  \le v \le \frac{3 - \frac{a_{t_1}}{a_{t_2}} }{2 a_{t_2} \left(2 - \frac{a_{t_1}}{a_{t_2}} \right) } $ and $t = t_1$.} We observe that:
\begin{align*}
&\mathbb{E}_{\V, T} [U_1(h^{-1}(v, t); h^{-1}(\V, T)] \\
&= \left(\mathbbm{1}[t = t_1] + \mathbbm{1}[t = t_2] \cdot \frac{1}{2} \right)  \cdot \mathbb{P}_{\V, T}[V \le v] - \max(0, a_{t} \cdot v - 1) + \frac{1}{2} \cdot  \mathbb{P}_{\V, T}[V > v, T = t_2] \cdot \mathbbm{1}[t = t_1] \\
&= \mathbb{P}_{\V, T}[V \le v] - \max(0, a_{t} \cdot v - 1) + \frac{1}{2} \cdot  \mathbb{P}_{\V, T}[V > v, T = t_2] \\
&= \underbrace{\mathbb{P}_{\V, T}[V \le v] - \frac{1}{2} \cdot \mathbb{P}_{\V, T}[V \le v, T = t_2] - \max(0, a_{t} \cdot v - 1)}_{(1)} + \underbrace{\frac{1}{2} \cdot  \mathbb{P}_{\V, T}[T = t_2]}_{(2)}.
\end{align*}
For term (1), we see that:
\begin{align*}
 &\mathbb{P}_{\V, T}[V \le v] - \frac{1}{2} \cdot \mathbb{P}_{\V, T}[V \le v, T = t_2] - \max(0, a_{t} \cdot v - 1) \\
 &= a_{t_1} \left(\frac{1}{a_{t_2}} - \frac{1}{a_{t_1}} \right) + 2  a_{t_2} \left( v - \frac{1}{a_{t_2}}\right) - \frac{1}{2} \cdot  2 \cdot a_{t_2} \left( v - \frac{1}{a_{t_2}}\right) \cdot \left(2 - \frac{a_{t_1}}{a_{t_2}} \right)  
 - (a_{t_1} \cdot v  -1)  \\
  &= a_{t_1} \left(\frac{1}{a_{t_2}} - \frac{1}{a_{t_1}} \right) + 2 a_{t_2} \left( v - \frac{1}{a_{t_2}}\right) \left(1 - \frac{1}{2} \left( 2 - \frac{a_{t_1}}{a_{t_2}} \right) \right) - (a_{t_1} \cdot v  -1)  \\
  &= a_{t_1} \left(v - \frac{1}{a_{t_1}} \right) - (a_{t_1} \cdot v  -1)  \\
  &= 0.
\end{align*}
For term (2), we apply Lemma \ref{lemma:property2types} to see that it is equal to $\frac{a_{t_1}}{2 \cdot a_{t_2}} - \frac{1}{2}$. Putting this all together, we see that:
\[\mathbb{E}_{\V, T} [U_1(h^{-1}(v, t); h^{-1}(\V, T)] = \frac{a_{t_1}}{2 \cdot a_{t_2}} - \frac{1}{2}.  \]

\paragraph{Case C: $\frac{3 - \frac{a_{t_1}}{a_{t_2}} }{2 a_{t_2} \left(2 - \frac{a_{t_1}}{a_{t_2}} \right) } \le v \le \frac{1}{a_{t_1}} + \left( \frac{1}{a_{t_1}} - \frac{1}{2 a_{t_2}} \right) \left(\frac{3 - \frac{a_{t_1}}{a_{t_2}}}{2 - \frac{a_{t_1}}{a_{t_2}} }\right)$ and $t = t_1$.} We use the same argument as in Case B to see that: 
\begin{align*}
\mathbb{E}_{\V, T} [U_1(h^{-1}(v, t); h^{-1}(\V, T)] &= \underbrace{\mathbb{P}_{\V, T}[V \le v] - \frac{1}{2} \cdot \mathbb{P}_{\V, T}[V \le v, T = t_2] - \max(0, a_{t} \cdot v - 1)}_{(1)} + \underbrace{\frac{1}{2} \cdot  \mathbb{P}_{\V, T}[T = t_2]}_{(2)}. 
\end{align*}
For term (1), we see that:
\begin{align*}
 &\mathbb{P}_{\V, T}[V \le v] - \frac{1}{2} \cdot \mathbb{P}_{\V, T}[V \le v, T = t_2] - \max(0, a_{t} \cdot v - 1) \\
 &= a_{t_1} \cdot \mathbb{P}\left[\frac{1}{a_{t_1}} \le \V \le \frac{1}{a_{t_2}}\right] + 2 a_{t_2} \cdot \mathbb{P}\left[\frac{1}{a_{t_2}} \le \V \le \frac{3 - \frac{a_{t_1}}{a_{t_2}} }{2 a_{t_2} \left(2 - \frac{a_{t_1}}{a_{t_2}} \right) } \right] \left(1 - \frac{1}{2} \left( 2 - \frac{a_{t_1}}{a_{t_2}} \right) \right) \\
  &+ a_{t_1} \cdot \mathbb{P}\left[\frac{3 - \frac{a_{t_1}}{a_{t_2}} }{2 a_{t_2} \left(2 - \frac{a_{t_1}}{a_{t_2}} \right) } \le \V \le v \right] - (a_{t_1} \cdot v - 1) \\
   &= a_{t_1} \cdot \left(\frac{1}{a_{t_2}} - \frac{1}{a_{t_1}} \right) + a_{t_1} \cdot \left( \frac{3 - \frac{a_{t_1}}{a_{t_2}} }{2 a_{t_2} \left(2 - \frac{a_{t_1}}{a_{t_2}} \right) } - \frac{1}{a_{t_2}} \right)  + a_{t_1} \left(v - \frac{3 - \frac{a_{t_1}}{a_{t_2}} }{2 a_{t_2} \left(2 - \frac{a_{t_1}}{a_{t_2}} \right) } \right) - (a_{t_1} \cdot v - 1)  \\
 &= a_{t_1} \left(v - \frac{1}{a_{t_1}} \right) -(a_{t_1} \cdot v - 1) \\
 &= 0.
\end{align*}
For term (2), we apply Lemma \ref{lemma:property2types} to see that it is equal to $\frac{a_{t_1}}{2 \cdot a_{t_2}} - \frac{1}{2}$. Putting this all together, we see that:
\[\mathbb{E}_{\V, T} [U_1(h^{-1}(v, t); h^{-1}(\V, T)] = \frac{a_{t_1}}{2 \cdot a_{t_2}} - \frac{1}{2}.  \]

\paragraph{Case D: $\frac{1}{a_{t_2}}  \le v \le \frac{3 - \frac{a_{t_1}}{a_{t_2}} }{2 a_{t_2} \left(2 - \frac{a_{t_1}}{a_{t_2}} \right) }$ and $t = t_2$.} 
We observe that:
\begin{align*}
\mathbb{E}_{\V, T} [U_1(h^{-1}(v, t); h^{-1}(\V, T)] &= \frac{1}{2} \mathbb{P}[\V \le v] - (a_{t_2} \cdot v - 1) \\
&= \frac{1}{2}\cdot a_{t_1} \left(\frac{1}{a_{t_2}} - \frac{1}{a_{t_1}} \right) + \frac{1}{2}\cdot 2 a_{t_2} \left( v - \frac{1}{a_{t_2}} \right) - (a_{t_2} \cdot v - 1) \\
&= \frac{a_{t_1}}{2 \cdot a_{t_2}} - \frac{1}{2} + a_{t_2} \cdot \left(v - \frac{1}{a_{t_2}} \right) - (a_{t_2} \cdot v - 1) \\
&= \frac{a_{t_1}}{2 \cdot a_{t_2}} - \frac{1}{2}.
\end{align*}

\paragraph{Case E: $\frac{3 - \frac{a_{t_1}}{a_{t_2}} }{2 a_{t_2} \left(2 - \frac{a_{t_1}}{a_{t_2}} \right) } \le v \le \frac{1}{a_{t_1}} + \left( \frac{1}{a_{t_1}} - \frac{1}{2 a_{t_2}} \right) \left(\frac{3 - \frac{a_{t_1}}{a_{t_2}}}{2 - \frac{a_{t_1}}{a_{t_2}} }\right)$ and $t = t_2$.} We observe that:
\begin{align*}
&\mathbb{E}_{\V, T} [U_1(h^{-1}(v, t); h^{-1}(\V, T)] \\
&= \frac{1}{2} \mathbb{P}[\V \le v] - (a_{t_2} \cdot v - 1) \\
&= \frac{1}{2}\cdot a_{t_1} \left(\frac{1}{a_{t_2}} - \frac{1}{a_{t_1}} \right) + \frac{1}{2}\cdot 2 a_{t_2} \left( \frac{3 - \frac{a_{t_1}}{a_{t_2}} }{2 a_{t_2} \left(2 - \frac{a_{t_1}}{a_{t_2}} \right) } - \frac{1}{a_{t_2}} \right) + \frac{1}{2} \cdot a_{t_1} \cdot \left(v - \frac{3 - \frac{a_{t_1}}{a_{t_2}} }{2 a_{t_2} \left(2 - \frac{a_{t_1}}{a_{t_2}} \right) } \right)
- (a_{t_2} \cdot v - 1) \\
&= \frac{a_{t_1}}{2 \cdot a_{t_2}} - \frac{1}{2} - a_{t_2} \left( v - \frac{3 - \frac{a_{t_1}}{a_{t_2}} }{2 a_{t_2} \left(2 - \frac{a_{t_1}}{a_{t_2}} \right) } \right) + \frac{1}{2} \cdot a_{t_1} \cdot \left(v - \frac{3 - \frac{a_{t_1}}{a_{t_2}} }{2 a_{t_2} \left(2 - \frac{a_{t_1}}{a_{t_2}} \right) } \right) \\
&= \frac{a_{t_1}}{2 \cdot a_{t_2}} - \frac{1}{2} + \left(\frac{a_{t_1}}{2} - a_{t_2} \right) \left(v - \frac{3 - \frac{a_{t_1}}{a_{t_2}} }{2 a_{t_2} \left(2 - \frac{a_{t_1}}{a_{t_2}} \right) } \right) \\
&\le  \frac{a_{t_1}}{2 \cdot a_{t_2}} - \frac{1}{2} .
\end{align*}

\paragraph{Case F: $v \ge \max(\text{supp}(\V))$.}
In this case, we see that:
\[\mathbb{E}_{\V, T} [U_1(h^{-1}(v, t); h^{-1}(\V, T)] = \left(\mathbbm{1}[t = t_1] + \frac{1}{2} \cdot \mathbbm{1}[t = t_2] \right) - (a_{t} \cdot v - 1). \]
For each value of $t$, this expression is maximized at $v = \max(\text{supp}(\V))$. If $t = t_1$, we can thus apply Case C; if $t = t_2$, we can thus apply Case E. Putting this all together, we see that:
\[\mathbb{E}_{\V, T} [U_1(h^{-1}(v, t); h^{-1}(\V, T)] \le \frac{a_{t_1}}{2 \cdot a_{t_2}} - \frac{1}{2}.  \]

\end{proof}

We now prove \Cref{thm:2types} using Lemma \ref{lemma:case2} and Lemma \ref{lemma:case3}. 
\begin{proof}[Proof of \Cref{thm:2types}]

For Case 1 (where $a_{t_1}/a_{t_2} \ge 1.5$),
we directly obtain the result from the analysis for $N$ well-separated types. In particular, we can apply Theorem \ref{thm:Ntypes} (or Proposition \ref{prop:2typeswellseparated}). Applying Proposition \ref{prop:NtypesMdistribution}, we see that the reparameterization of the distribution specified in Definition \ref{def:2types} is identical to the distribution in Definition \ref{def:2typeswellseparated}, which yields the desired statement. 

The remainder of the proof boils down to analyzing Case 2 ($(5 - \sqrt{5})/2 \le a_{t_1} / a_{t_2} \le 1.5$)and Case 3 ($1 <a_{t_1} / a_{t_2} \le (5 - \sqrt{5})/2 $). 

We first claim that we can work over the reparameterized space $(\V, T)$ defined in \Cref{subsec:tools}. This follows from an analogous argument to the proof of \Cref{thm:Ntypes} which we repeat for completeness. Note that every $w \in \text{supp}(\muengagement(P, c, u, \mathcal{T}))$ is in the 
image $\cup_{t \in \mathcal{T}} \Ct$, which means that it is associated with a unique value $h^{-1}(w) = (v, t) \in S$. It thus suffices to show that there exists a best response $w \in \mathbb{R}_{\ge 0}^2$ to:
\[\argmax_{w \in \mathbb{R}_{\ge 0}^2} \mathbb{E}_{\textbf{w}_{-i} \sim (\muengagement(P, c, u, \mathcal{T}))^{P-1}} [U_i(w; \textbf{w}_{-i})]  \]
that is also in the image $\cup_{t \in \mathcal{T}} \Ct$; this follows from Lemma \ref{lemma:bestresponsestronger}. 

We thus work over the reparameterized space for the remainder of the analysis. For Case 2, 
by Lemma \ref{lemma:case2}, we see that for any $(v, t) \in S$, the expected utility satisfies:
\[\mathbb{E}_{\V, T} [U_1(h^{-1}(v, t); (\V, T))] \le \frac{a_{t_1}}{2 \cdot a_{t_2}} - \frac{1}{2} . \]
Moreover, for $(v, t) \in \text{supp}((\V, T))$, it holds that:
\[\mathbb{E}_{\V, T} [U_1(h^{-1}(v, t); (\V, T))] = \frac{a_{t_1}}{2 \cdot a_{t_2}} - \frac{1}{2} . \]
This proves that $(\V, T)$ is an equilibrium in the reparameterized space $S$. Similarly, for Case 3, by Lemma \ref{lemma:case3}, we see that for any $(v, t) \in S$, the expected utility satisfies:
\[\mathbb{E}_{\V, T} [U_1(h^{-1}(v, t); (\V, T))] \le \frac{a_{t_1}}{2 \cdot a_{t_2}} - \frac{1}{2}. \]
Moreover, for $(v, t) \in \text{supp}((\V, T))$, it holds that:
\[\mathbb{E}_{\V, T} [U_1(h^{-1}(v, t); (\V, T))] = \frac{a_{t_1}}{2 \cdot a_{t_2}} - \frac{1}{2}. \]
This proves that $(\V, T)$ is an equilibrium in the reparameterized space $S$.

Putting this all together, this proves that  $\muengagement(P, c, u, \mathcal{T})$ is an equilibrium in the original space $\mathbb{R}_{\ge 0}^2$. 
\end{proof}

\section{Proofs for \Cref{sec:equilibriumcharacterizationsbaselines}}\label{appendix:equilibriumcharacterizationsbaselines}

\subsection{Proof of Theorem \ref{thm:investmentbased}}
We prove Theorem \ref{thm:investmentbased}. 
\begin{proof}[Proof of Theorem \ref{thm:investmentbased}]
Let $\mu = \muideal(P, c, u, \mathcal{T})$ for notational convenience, and let $(\Wcostly, \Wcheap) \sim \mu$. 
We analyze the expected utility of  
\[H(w) = \mathbb{E}_{\mathbf{w}_{-i} \sim \mu_{-i}}[U_i(w; w_{-i})]\] of a content creator if all of the creators choose the strategy $\mu$. We show that $H(w) = 0$ if $w \in \text{supp}(\mu)$ and $H(w) \le 0$ for any $w \in \mathbb{R}_{\ge 0}^2$.

Let $w \in \mathbb{R}_{\ge 0}^2$ be any content vector, and let $w' = [\wcostly, 0]$ be the vector with identical quality but no gaming tricks. Since 
\[\UtilityCostly(\wcostly, t) = u(w', t) \ge u(w, t),\] $\MIdeal(w) = \MIdeal(w')$, and $c(w) \ge c(w')$,
it holds that $H(w) \le H(w')$. Since all $w'' \in \text{supp}(\mu)$ also satisfy $\wcheap'' = 0$, we can restrict the rest of our analysis to $w$ such that $\wcheap = 0$. 

We split the remainder of the analysis into two cases: (1) $\mathcal{T} = \left\{t \right\}$ and (2) $\beta_t = 0$ for all $t \in \mathcal{T}$. 

\paragraph{Case 1: $\mathcal{T} = \left\{t \right\}$.} It suffices to show that $H(w) = 0$ if $w \in \text{supp}(\mu)$ and $H(w) \le 0$ for any $w = [\wcostly, 0] \in \mathbb{R}_{\ge 0} \times \left\{0\right\}$. 

First, we show that $H(w) \le 0$ for any $w \in \mathbb{R}_{\ge 0}^2$ such that $\wcheap = 0$. If $\wcostly < \beta_t$, then it follows immediately that $H(w) \le 0$. If $\wcostly \ge \beta_t$, then we see that: 
\begin{align*}
 H(w) &= \mathbb{E}_{\mathbf{w}_{-i} \sim \mu_{-i}}[U_i(w'; w_{-i})]    \\
 &= \left(\min\left(1, \CCostlyBaseline(\wcostly) \right)\right) \cdot \mathbbm{1}[\UtilityCostly(\wcostly, t) \ge 0] - \CCostlyBaseline(\wcostly) \\
 &\le \CCostlyBaseline(\wcostly) - \CCostlyBaseline(\wcostly) \\
 &= 0.
\end{align*}

Next, we show that if $w \in \text{supp}(\mu)$, then it holds that $H(w) = 0$. If $\wcostly = 0$, then it follows easily that $H(w) = 0$. Otherwise, we see that: 
\begin{align*}
  H(w)  &= \left(\min\left(1, \CCostlyBaseline(\wcostly) \right)\right) \cdot \mathbbm{1}[\UtilityCostly(\wcostly, t) \ge 0] - \CCostlyBaseline(\wcostly) \\
   &= \CCostlyBaseline(\wcostly) \cdot \mathbbm{1}[\UtilityCostly(\wcostly, t) \ge 0] - \CCostlyBaseline(\wcostly) \\
  &= 0.  
\end{align*}
This proves that $\mu$ is an equilibrium as desired. 

\paragraph{Case 2: $\beta_t = 0$ for all $t \in \mathcal{T}$.} It suffices to show that $H(w) = 0$ if $w \in \text{supp}(\mu)$ and $H(w) \le 0$ for any $w = [\wcostly, 0] \in \mathbb{R}_{\ge 0} \times \left\{0\right\}$. 

First, we show that $H(w) \le 0$ for any $w \in \mathbb{R}_{\ge 0}^2$ such that $\wcheap = 0$. Then we see that: 
\begin{align*}
 H(w) &= \mathbb{E}_{\mathbf{w}_{-i} \sim \mu_{-i}}[U_i(w'; w_{-i})]    \\
 &= \min\left(1, \CCostlyBaseline(\wcostly) \right) - \CCostlyBaseline(\wcostly) \\
 &\le \CCostlyBaseline(\wcostly) - \CCostlyBaseline(\wcostly) \\
 &= 0.
\end{align*}

Next, we show that if $w \in \text{supp}(\mu)$, then it holds that $H(w) = 0$. We see that: 
\begin{align*}
  H(w)  &= \CCostlyBaseline(\wcostly) - \CCostlyBaseline(\wcostly) \\
  &= 0.  
\end{align*}
This proves that $\mu$ is an equilibrium as desired.

\end{proof}

\subsection{Proof of Theorem \ref{thm:randomrecommendations}}
We prove Theorem \ref{thm:randomrecommendations}. 
\begin{proof}
Let $\mu = \murandom(P, c, u, \mathcal{T})$ for notational convenience, and let $(\Wcostly, \Wcheap) \sim \mu$. We analyze the expected utility of  
\[H(w) = \mathbb{E}_{\mathbf{w}_{-i} \sim \mu_{-i}}[U_i(w; w_{-i})]\] of a content creator if all of the creators choose the strategy $\mu$. We show that $w \in \argmax_{w'} H(w')$ for any $w  \in \text{supp}(\mu)$.  

Let $w \in \mathbb{R}_{\ge 0}^2$ be any content vector, and let $w' = [\wcostly, 0]$ be the vector with identical quality but no gaming tricks. Since 
\[\UtilityCostly(\wcostly, t) = u(w', t) \ge u(w, t),\] $\MTrivial(w) = \MTrivial(w')$, and $c(w) \ge c(w')$,
it holds that $H(w) \le H(w')$. Since all $w'' \in \text{supp}(\mu)$ also satisfy $\wcheap'' = 0$, we can restrict the rest of our analysis to $w$ such that $\wcheap = 0$.

We split the remainder of the analysis into two cases: (1) $\mathcal{T} = \left\{t \right\}$ and (2) $\beta_t = 0$ for all $t \in \mathcal{T}$. 

\paragraph{Case 1: $\mathcal{T} = \left\{ t \right\}$.}  We split into two subcases: $\kappa \le 1/P$ and $\kappa \in (1/P, 1]$. 

If $\kappa \le 1/P$, then $\Wcostly$ is a point mass at $\beta_t$. Note that: 
\[H(w) = \mathbb{E}_{\mathbf{w}_{-i} \sim \mu_{-i}}[U_i(w; w_{-i})] = \frac{\mathbbm{1}[\UtilityCostly(\wcostly, t) \ge 0]}{P} - \CCostlyBaseline(w) \le \frac{1}{P} - \kappa.\]
Moreover, for $\wcostly = \wcostly^*$, it holds that $H(w) = \frac{1}{P} - \kappa$. This proves that $w \in \argmax_{w'} H(w')$ for any $w  \in \text{supp}(\mu)$, as desired. 

If $\kappa \in (1/P, 1]$, then we see that $\nu$ is the unique value such $\sum_{i=1}^{P-1} \nu^i = P \cdot \kappa$. Note that $\Wcostly$ is $\beta_t$ with probability $1 - \nu$ and $0$ with probability $\nu$. Moreover, note that:  
\[H(w) = \mathbb{E}_{\mathbf{w}_{-i} \sim \mu_{-i}}[U_i(w; w_{-i})] = \mathbbm{1}[\UtilityCostly(\wcostly, t) \ge 0] \cdot \mathbb{E}_Y\left[\frac{1}{1 + Y}\right] - \CCostlyBaseline(w),\]
where $Y \sim \text{Bin}(P - 1, 1 - \nu)$ is distributed as a binomial random variable with success probability $1 - \nu$. (The second equality holds because $Y$ is distributed as the number of creators $j \neq i$ who choose content generating nonnegative utility for the user.) A simple calculation shows that:
\[ \mathbb{E}_Y\left[\frac{1}{1 + Y}\right] = \frac{1}{P} \sum_{i=0}^{P-1} \nu^i = \kappa,\]
where the last equality follows from the definition of $\eta$. 
This means that $H(w) \le 0$ for all $w$. For $\wcostly = \beta_t$ and $\wcostly = 0$, it holds that $H(w) = 0$. This means that  $w \in \argmax_{w'} H(w')$ for any $w  \in \text{supp}(\mu)$, as desired. 

\paragraph{Case 2: $\beta_t = 0$ for all $t \in \mathcal{T}$.} In this case, $\Wcostly$ is a point mass at $0$. Note that: 
\[H(w) = \mathbb{E}_{\mathbf{w}_{-i} \sim \mu_{-i}}[U_i(w; w_{-i})] = \frac{\mathbbm{1}[\UtilityCostly(\wcostly, t) \ge 0]}{P} - \CCostlyBaseline(w) \le \frac{1}{P}.\]
Moreover, for $\wcostly = 0$, it holds that $H(w) = \frac{1}{P}$. This proves that $w \in \argmax_{w'} H(w')$ for any $w  \in \text{supp}(\mu)$, as desired. 
\end{proof}

\section{Proofs for \Cref{sec:equilibriuminvestment}}

\subsection{Proofs of Theorem \ref{thm:comparisonuserconsumptiongaming} and Theorem \ref{thm:comparisonuserconsumptioninvestment}}\label{appendix:comparisonuserconsumptiongaming}
The main lemma is the following characterization of user consumption of utility as the maximum investment in quality across the content landscape. We slightly abuse notation and for a content landscape $\mathbf{w} = (w_1, \ldots, w_P)$, we use the notation $w \in \mathbf{w}$ to denote $w_j$ for $j \in [P]$. 
\begin{lemma}
\label{lemma:comparisonuserconsumption}
Consider the setup of Theorem \ref{thm:comparisonuserconsumptiongaming}. For $\textbf{w} \in \text{supp}(\muideal(P, c, u, \mathcal{T})^P)$, it holds that \[\UserConsumption(\MIdeal; \textbf{w}) = \max_{w \in \textbf{w}} \wcostly \]
and for $\textbf{w} \in \text{supp}(\muideal(P, c, u, \mathcal{T})^P)$, it holds that 
\[\UserConsumption(\MPlatform; \textbf{w}) = \max_{w \in \textbf{w}} \wcostly.\]
\end{lemma}

We now prove Lemma \ref{lemma:comparisonuserconsumption}. 
\begin{proof}[Proof of Lemma \ref{lemma:comparisonuserconsumption}]
We observe that for $w \in \text{supp}(\muideal(P, c, u, \mathcal{T}))$, it holds that if $\mathbbm{1}[u(w, t)] < 0$, then $w = [0, 0]$. Thus, for $\textbf{w} \in \text{supp}(\muideal(P, c, u, \mathcal{T}))^P$, it holds that:
\[\UserConsumption(\MIdeal; \textbf{w}) = \mathbb{E}\left[w^{\text{costly}}_{i^*(\MIdeal; \textbf{w})} \cdot \mathbbm{1}[u(w_{i^*(M; \textbf{w})}, t) \ge 0]\right] =\mathbb{E}\left[w^{\text{costly}}_{i^*(\MIdeal; \textbf{w})}\right]. \]
Moreover, since $\wcheap = 0$ for all $w \in \text{supp}(\muideal(P, c, u, \mathcal{T}))$ and by the definition of $\MIdeal$, we see that $w^{\text{costly}}_{i^*(\MIdeal; \textbf{w})} = \max_{w \in \textbf{w}} \wcostly$. This means that:
\[\UserConsumption(\MIdeal; \textbf{w}) = \mathbb{E}\left[\max_{w \in \textbf{w}} \wcostly \right]. \]

Similarly, we observe that for $w \in \text{supp}(\muengagement(P, c, u, \mathcal{T}))$, it holds that if $\mathbbm{1}[u(w, t)] < 0$, then $w = [0, 0]$. 
Thus, for $\textbf{w} \in \text{supp}(\muengagement(P, c, u, \mathcal{T}))^P$, it holds that:
\[\UserConsumption(\MPlatform; \textbf{w}) = \mathbb{E}\left[w^{\text{costly}}_{i^*(M; \textbf{w})} \cdot \mathbbm{1}[u(w_{i^*(M; \textbf{w})}, t)] \ge 0\right] =\mathbb{E}\left[w^{\text{costly}}_{i^*(M; \textbf{w})}\right]. \]
Moreover, by the definition of $\text{supp}(\muengagement(P, c, u, \mathcal{T}))$ and by the definition of $\MPlatform$, we see that $w^{\text{costly}}_{i^*(\MPlatform; \textbf{w})} = \max_{w \in \textbf{w}} \wcostly$. This means that:
\[\UserConsumption(\MPlatform; \textbf{w}) = \mathbb{E}\left[\max_{w \in \textbf{w}} \wcostly \right]. \]
\end{proof}

Using Lemma \ref{lemma:comparisonuserconsumption}, we prove Theorem \ref{thm:comparisonuserconsumptiongaming} and Theorem \ref{thm:comparisonuserconsumptioninvestment}. 
\begin{proof}[Proof of Theorem \ref{thm:comparisonuserconsumptiongaming} and Theorem \ref{thm:comparisonuserconsumptioninvestment}]

By Lemma \ref{lemma:comparisonuserconsumption}, it suffices to analyze 
\[\mathbb{E}_{\textbf{w} \sim \mu^P} \left[\max_{w \in \textbf{w}} \wcostly \right]\]
where $\mu \in \left\{\muengagement(P, c, u, \mathcal{T}), \muideal(P, c, u, \mathcal{T}) \right\}$.
By \Cref{prop:equilibriumhomogenousengagement}, we see that $\mathbb{P}_{(\Wcostly, \Wcheap) \sim (\muengagement(P, c, u, \mathcal{T}))}[\Wcostly \le \wcostly]$ is equal to: 
\[
\begin{cases}
(-\alpha)^{1/(P-1)} & \text{ if }  0 \le \wcostly \le -\alpha \\
  \left(\min(1, \wcostly + \gamma \cdot t \cdot (\wcostly + \alpha))\right)^{1/(P-1)} & \text{ if } \wcostly \ge \max(0, -\alpha).
\end{cases}
\]
By \Cref{prop:equilibriumhomogenousinvestment}, we see that: 
\[\mathbb{P}_{(\Wcostly, \Wcheap) \sim (\muideal(P, c, u, \mathcal{T}))}[\Wcostly \le \wcostly] = 
\begin{cases}
(-\alpha)^{1/(P-1)} & \text{ if }  0 \le \wcostly \le -\alpha \\
  \left(\min(1, \wcostly)\right)^{1/(P-1)} & \text{ if } \wcostly \ge \max(0, -\alpha).
\end{cases}
\]

\paragraph{Proof of Theorem \ref{thm:comparisonuserconsumptiongaming}.}
The marginal distribution of $\Wcostly$ for $\muengagement(P, c, u, \mathcal{T})$ implies for engagement-based optimization, the distribution of $\Wcostly$ for higher values of $\gamma$ is stochastically dominated by the distribution of $\Wcostly$ for lower values of $\gamma$. This implies that $\mathbb{E}_{\textbf{w} \sim (\muengagement(P, c, u, \mathcal{T}))^P} \left[\max_{w \in \textbf{w}} \wcostly \right]$ is strictly increasing in $\gamma$, which implies that $\mathbb{E}_{\textbf{w} \sim (\muengagement(P, c, u, \mathcal{T}))^P} \left[\UserConsumption(\MPlatform; \textbf{w})\right]$
is strictly increasing in $\gamma$. 

\paragraph{Proof of Theorem \ref{thm:comparisonuserconsumptioninvestment}.}
Observe that the marginal distribution of $\Wcostly$ for $\muideal(P, c, u, \mathcal{T})$ stochastically dominates the distribution of $\Wcostly$ for $\muengagement(P, c, u, \mathcal{T})$, with strict stochastic dominance for $\gamma > 0$. This implies that if $\gamma > 0$:
\[ \mathbb{E}_{\textbf{w} \sim (\muengagement(P, c, u, \mathcal{T}))^P} \left[\max_{w \in \textbf{w}} \wcostly \right] < \mathbb{E}_{\textbf{w} \sim (\muideal_{P, \alpha,  t})^P} \left[\max_{w \in \textbf{w}} \wcostly \right],\]
which implies that 
\[ \mathbb{E}_{\textbf{w} \sim (\muengagement(P, c, u, \mathcal{T}))^P} \left[\UserConsumption(\MPlatform; \textbf{w}) \right] < \mathbb{E}_{\textbf{w} \sim (\muideal_{P, \alpha,  t})^P} \left[\UserConsumption(\MIdeal; \textbf{w})\right].\]
Moreover, if $\gamma = 0$, the two distributions are equal, which implies that 
\[ \mathbb{E}_{\textbf{w} \sim (\muengagement(P, c, u, \mathcal{T}))^P} \left[\max_{w \in \textbf{w}} \wcostly \right] = \mathbb{E}_{\textbf{w} \sim (\muideal_{P, \alpha,  t})^P} \left[\max_{w \in \textbf{w}} \wcostly \right],\] 
which implies that 
\[ \mathbb{E}_{\textbf{w} \sim (\muengagement(P, c, u, \mathcal{T}))^P} \left[\UserConsumption(\MPlatform; \textbf{w}) \right] = \mathbb{E}_{\textbf{w} \sim (\muideal_{P, \alpha,  t})^P} \left[\UserConsumption(\MIdeal; \textbf{w})\right].\]
\end{proof}

\subsection{Proof of Proposition \ref{prop:userconsumptionNtypes}}\label{appendix:userconsumptionNTypes}
We prove \Cref{prop:userconsumptionNtypes}. 
\begin{proof}[Proof of \Cref{prop:userconsumptionNtypes}]
We construct the following instantation of Example \ref{example:linear} with $N \ge 2$ types. Let the type space be $\mathcal{T}_{N, \epsilon} = \left\{(1+ \epsilon) (1 + 1/N)^{i-1} - 1 \mid 1 \le i \le N \right\}$, and let $P =2$, $\alpha = 1$, and $\gamma = 0$. It suffices to show that $\mathbb{E}_{\mathbf{w} \sim \left(\muideal(2, c, u, \mathcal{T}_{N, \epsilon})\right)^2}[\UserConsumption(\MIdeal; \mathbf{w})] = 2/3$ and $\mathbb{E}_{\mathbf{w} \sim \left(\muengagement(2, c, u, \mathcal{T}_{N, \epsilon})\right)^2} \le 1/N$. 

Before analyzing these two expressions, we first compute $f_t$ and $C_t$ for this example: we observe that $f_t(\wcheap) = \wcheap / t$ and $C_t(\wcheap) = \wcheap / t$. 

First, we show that $\mathbb{E}_{\mathbf{w} \sim \left(\muideal(2, c, u, \mathcal{T}_{N, \epsilon})\right)^2}[\UserConsumption(\MIdeal; \mathbf{w})] = 2/3$. We use the characterization in Theorem \ref{thm:investmentbased}. It is easy to see that:
\[\mathbb{E}_{\mathbf{w} \sim \left(\muideal(2, c, u, \mathcal{T}_{N, \epsilon})\right)^2}[\UserConsumption(\MIdeal; \mathbf{w})] = \int_{0}^{\infty} (1 - \mathbb{P}[\Wcostly \le \wcostly]^2) d\wcostly. \]
We see that: 
\[\mathbb{P}[\Wcostly \le \wcostly] =\mathbb{P}[\Wcheap \le t \cdot \wcostly] = \min(1, C_t(t \cdot \wcostly)) = \min(1, \wcostly).  \]
This means that:
\[\mathbb{E}_{\mathbf{w} \sim \left(\muideal(2, c, u, \mathcal{T}_{N, \epsilon})\right)^2}[\UserConsumption(\MIdeal; \mathbf{w})] = \int_0^{1} (1 - \wcostly^2) d\wcostly = \frac{2}{3}, \]
as desired.

Next, we show that $\mathbb{E}_{\mathbf{w} \sim \left(\muengagement(2, c, u, \mathcal{T}_{N, \epsilon})\right)^2} \le 1/N$. We use the characterization in Theorem \ref{thm:Ntypes}. We observe that $\max(\text{supp}(\Wcheap^i)) \le \frac{1}{N t_i}$. This follows from the fact that $\max(\text{supp}(\Wcheap^i)) =  \frac{t_i}{N}$ for $1 \le i \le N' - 1$ and 
\begin{align*}
  \max(\text{supp}(\Wcheap^i)) &=  t_i \cdot \frac{N}{N - N'+1} \left(1 - \sum_{j=1}^{N' -1} \frac{1}{N- j+1}\right)^{-1} \\
  &\le t_i \cdot \frac{N-N'+1}{N } \cdot \left(1 - \sum_{j=1}^{N' -1} \frac{1}{N- j+1}\right) \\
  &\le t_i \cdot \frac{N-N'+1}{N } \cdot \frac{1}{N- N'+1} \\
  &= \frac{t_i}{N}  
\end{align*}
for $i = N'$. 
This implies that:
\[\max_{1 \le i \le N'} \max(\text{supp}(\Wcostly^i)) = \max_{1 \le i \le N'} t \cdot \frac{1}{N t_i} = \frac{1}{N}.  \]
Thus, we see that:
\[\mathbb{E}_{\mathbf{w} \sim \left(\muideal(2, c, u, \mathcal{T}_{N, \epsilon})\right)^2}[\UserConsumption(\MIdeal; \mathbf{w})] \le \max_{w \in \text{supp}(\muideal(2, c, u, \mathcal{T}_{N, \epsilon}))} \wcostly \le \frac{1}{N} \]
as desired. 
\end{proof}

\section{Proofs for \Cref{sec:platformimplications}}

\subsection{Proof of Theorem \ref{thm:comparisonengagement}}\label{appendix:proofrealizedengagement}

The main ingredient of the proof of Theorem \ref{thm:comparisonengagement} is constructing and analyzing an instance where engagement-based optimization achieves a low realized engagement. 
We construct the instance to be Example \ref{example:linear} with costless gaming ($\gamma = 0$), baseline utility $\alpha = 1$, and $P = 2$ creators, and type space $\mathcal{T}_{N, \epsilon} := \left\{(1+ \epsilon) (1 + 1/N)^{i-1} - 1 \mid 1 \le i \le N \right\}$ for sufficiently small $\epsilon > 0$ and sufficiently large $N$.

In order to prove Theorem \ref{thm:comparisonengagement}, we first show the following lemma that relates the realized engagement to the maximum engagement achieved by any content in the content landscape. We again slightly abuse notation and for a content landscape $\mathbf{w} = (w_1, \ldots, w_P)$, we use the notation $w \in \mathbf{w}$ to denote $w_j$ for $j \in [P]$. 
\begin{lemma}
\label{lemma:comparisonrealizedengagement}
Consider Example \ref{example:linear} with costless gaming ($\gamma = 0$), baseline utility $\alpha = 1$, and $P = 2$ creators, and type space $\mathcal{T}_{N, \epsilon} := \left\{(1+ \epsilon) (1 + 1/N)^{i-1} - 1 \mid 1 \le i \le N \right\}$ for some $N \ge 1$ and $\epsilon > 0$. For $\textbf{w} \in \text{supp}(\muideal(2, c, u, \mathcal{T}_{N, \epsilon})^2)$, it holds that \[\RealizedEngagement(\MIdeal; \textbf{w}) = \max_{w \in \textbf{w}} \MPlatform(w) \]
and for $\textbf{w} \in \text{supp}(\muengagement(2, c, u, \mathcal{T}_{N, \epsilon})^2)$, it holds that 
\[\RealizedEngagement(\MPlatform; \textbf{w}) \le \max_{w \in \textbf{w}} \MPlatform(w).\]
\end{lemma}
\noindent We defer the proof of Lemma \ref{lemma:comparisonrealizedengagement} to \Cref{appendix:comparisonrealizedengagement}

Given Lemma \ref{lemma:comparisonrealizedengagement}, it suffices to analyze the engagement distribution at equilibrium for engagement-based optimization and investment-based optimization. More formally, within the instance constructed above, let $\V^{I, \epsilon, \mathcal{T}_{N, \epsilon}}$ be the distribution of $\MPlatform(w) + s$
where $w \sim \muideal(2, c, u, \mathcal{T}_{N, \epsilon})$. The distribution $\V^{I, \epsilon, \mathcal{T}_{N, \epsilon}}$ can be characterized in closed-form as follows:
\begin{lemma}
\label{lemma:investmentcdf}
The distribution $\V^{I, \epsilon, \mathcal{T}_{N, \epsilon}}$ has cdf equal to: 
\[\mathbb{P}[\V^{I, \epsilon, \mathcal{T}_{N, \epsilon}} \le v] = 
\begin{cases}
0 & \text{ if } v \le 1 \\
v - 1  & \text{ if } 1 \le v \le 2 \\
1 & \text { if } v \ge 2
\end{cases}
\]
\end{lemma}
Moreover, let $\V^{E, \epsilon, \mathcal{T}_{N, \epsilon}}$ be the distribution of $\MPlatform(w) + s$ where $w \sim \muengagement(2, c, u, \mathcal{T})$. While the distribution $\V^{E, \epsilon, \mathcal{T}_{N, \epsilon}}$ is somewhat messy, for each $\epsilon > 0$, we show pointwise convergence to a simpler distribution $V^{e, \epsilon, \infty}$ defined by:
\[
\mathbb{P}[V^{e, \epsilon, \infty} \le v] = \begin{cases}
 0 & \text{ if } v \le 1 + \epsilon \\
 \ln\left(\frac{1}{1 - \ln\left(\frac{v}{1+\epsilon}\right)} \right)  & \ln\left(\frac{1}{1 - \ln\left(\frac{v}{1+\epsilon}\right)} \right) \text{ if } v \in \left[1 + \epsilon , (1+\epsilon) e^{1-e} \right]\\
 1 & \text{ if } v \ge (1+\epsilon) e^{1-e}
\end{cases}
\]
as formalized in the following lemma: 
\begin{lemma}
\label{lemma:engagementlimitcdf}
For each $\epsilon > 0$, the cdf of $\V^{E, \epsilon, \mathcal{T}_{N, \epsilon}}$ as $N \rightarrow \infty$ converges pointwise a.e. to the cdf of $V^{e, \epsilon, \infty}$. 
\end{lemma}
\noindent We defer the proof of Lemma \ref{lemma:engagementlimitcdf} to \Cref{appendix:prooflimitcdflemma}

Using Lemma \ref{lemma:comparisonrealizedengagement}, Lemma \ref{lemma:investmentcdf}, Lemma \ref{lemma:engagementcdf}, and Lemma \ref{lemma:engagementlimitcdf}, we prove Theorem \ref{thm:comparisonengagement}.
\begin{proof}[Proof of Theorem \ref{thm:comparisonengagement}]
By Lemma \ref{lemma:comparisonrealizedengagement}, it suffices to show that the following limits exist and that 
\begin{equation}
\label{eq:sufficientconditionengagement}
 \lim_{\epsilon \rightarrow 0}  \lim_{N \rightarrow \infty} \mathbb{E}_{\muengagement(2, c, u, \mathcal{T}_{N, \epsilon})}\left[\max_{w \in \textbf{w}} \MPlatform(w) \right] < \mathbb{E}_{\muideal(2, c, u, \mathcal{T}_{N, \epsilon})}\left[\max_{w \in \textbf{w}} \MPlatform(w) \right]   
\end{equation}
\noindent We analyze the left-hand side of  \eqref{eq:sufficientconditionengagement}, then analyze the right-hand side of \eqref{eq:sufficientconditionengagement}, and then use these analyses to prove \eqref{eq:sufficientconditionengagement}. 
\paragraph{Analysis of left-hand side of \eqref{eq:sufficientconditionengagement}.}
We first analyze the left-hand side of \eqref{eq:sufficientconditionengagement}. We see that:  
\begin{align*}
   \mathbb{E}_{\muengagement(2, c, u, \mathcal{T}_{N, \epsilon})}\left[\max_{w \in \textbf{w}} \MPlatform(w) \right] &= \mathbb{E}_{(V_1, V_2) \sim (\V^{E, \epsilon, \mathcal{T}_{N, \epsilon}})^2}\left[\max(V_1, V_2) \right] \\
   &= \int_{0}^{\infty} \left(1 - \left(\mathbb{P}[\V^{E, \epsilon, \mathcal{T}_{N, \epsilon}} \le v]\right)^2\right) dv.
\end{align*}
We take a limit of this expression as $N \rightarrow \infty$. By Lemma \ref{lemma:engagementlimitcdf}, the function  $\left(1 - \left(\mathbb{P}[\V^{E, \epsilon, \mathcal{T}_{N, \epsilon}} \le v]\right)^2\right)$ pointwise approaches the function $\left(1 - \left(\mathbb{P}[\V^{E, \epsilon, \infty } \le v]\right)^2\right)$. Moreover, we see that $\left(1 - \left(\mathbb{P}[\V^{E, \epsilon, \mathcal{T}_{N, \epsilon}} \le v]\right)^2\right) \le g(v)$, where $g(v) = 1$ if $0 \le v \le 3$ and $g(v) = 0$ if $v \ge 3$. Applying dominated convergence with dominating function $g$, we see that: 
\[\lim_{N \rightarrow \infty}  \int_{0}^{\infty} \left(1 - \left(\mathbb{P}[\V^{E, \epsilon, \mathcal{T}_{N, \epsilon}} \le v]\right)^2\right) dv = \int_{0}^{\infty} \left(1 - \left(\mathbb{P}[\V^{E, \epsilon, \infty} \le v]\right)^2\right) dv. \]

We next take a limit of $\int_{0}^{\infty} \left(1 - \left(\mathbb{P}[\V^{E, \epsilon, \infty} \le v]\right)^2\right) dv$ as $\epsilon \rightarrow 0$. We see that $\V^{E, \epsilon, \infty}$ pointwise a.e. approaches $\V^{E, \infty}$ defined by:
\[ 
\mathbb{P}[\V^{E, \infty} \le v] =
\begin{cases}
    0 & \text{ if } v \le 1 \\
    \ln \left(\frac{1}{1 - \ln(x)} \right) &\text{ if } 1 \le v \le e^{1-1/e} \\
    1 & \text{ if } v \ge e^{1-1/e}
\end{cases}.
\]
We again apply dominated convergence with the same function $g$ as above to see that:
\begin{align*}
\int_{0}^{\infty} \left(1 - \left(\mathbb{P}[\V^{E, \infty} \le v]\right)^2\right) dv &= \lim_{\epsilon \rightarrow 0}   \int_{0}^{\infty} \left(1 - \left(\mathbb{P}[\V^{E, \epsilon, \infty} \le v]\right)^2\right) dv.
\end{align*}
Putting this all together, we see that:
\[\lim_{\epsilon \rightarrow 0}  \lim_{N \rightarrow \infty} \mathbb{E}_{\muengagement(2, c, u, \mathcal{T}_{N, \epsilon})}\left[\max_{w \in \textbf{w}} \MPlatform(w) \right] = \int_{0}^{\infty} \left(1 - \left(\mathbb{P}[\V^{E, \infty} \le v]\right)^2\right) dv.\]

\paragraph{Analysis of right-hand side of \eqref{eq:sufficientconditionengagement}.}
We next analyze the right-hand side of \eqref{eq:sufficientconditionengagement} as
\begin{align*}
   \mathbb{E}_{\muideal(2, c, u, \mathcal{T}_{N, \epsilon})}\left[\max_{w \in \textbf{w}} \MPlatform(w) \right] &= \mathbb{E}_{(V_1, V_2) \sim (\V^{I, \epsilon, \mathcal{T}_{N, \epsilon}})^2}\left[\max(V_1, V_2) \right] \\
   &= \int_{0}^{\infty} \left(1 - \left(\mathbb{P}[\V^{I, \epsilon, \mathcal{T}_{N, \epsilon}} \le v]\right)^2\right) dv \\
   &= \int_{0}^{\infty} \left(1 - \left(\min(1, \max(0, v-1)) \right)^2\right) dv.
\end{align*}

\paragraph{Proof of \eqref{eq:sufficientconditionengagement}.}
We are now ready to prove \eqref{eq:sufficientconditionengagement}. Using the above calculations, it suffices to show:
\[ \int_{0}^{\infty} \left(1 - \left(\mathbb{P}[\V^{E, \infty} \le v]\right)^2\right) dv < \int_{0}^{\infty} \left(1 - \left(\min(1, \max(0, v-1)) \right)^2\right) dv.\]
To show this, it suffices to  show that $\mathbb{P}[\V^{E, \infty} \le v] \ge \min(1, \max(0, v-1))$ for all $v$ and $\mathbb{P}[\V^{E, \infty} \le v] > \min(1, \max(0, v-1))$ for $v \in (e^{1-1/e}, 1)$. The fact that $\mathbb{P}[\V^{E, \infty} \le v] > \min(1, \max(0, v-1))$ for $v \in (e^{1-1/e}, 1)$ follows easily from the functional form of $\mathbb{P}[\V^{E, \infty} \le v]$. 
Moreover,  $\mathbb{P}[\V^{E, \infty} \le v] = 0 = \min(1, \max(0, v-1))$ for $v \le 1$. To see that $\mathbb{P}[\V^{E, \infty} \le v] \ge \min(1, \max(0, v-1))$ for $v \in [1, e^{1-1/e}]$, we apply Lemma \ref{lemma:inequality} (see \Cref{appendix:prooflemmainequality} for the statement and proof). 

\end{proof}

\subsubsection{Proof of Lemma \ref{lemma:comparisonrealizedengagement}}\label{appendix:comparisonrealizedengagement}

We prove Lemma \ref{lemma:comparisonrealizedengagement}.
\begin{proof}[Proof of Lemma \ref{lemma:comparisonrealizedengagement}]

We first prove the lemma statement for investment-based optimization, and then we prove the statement for engagement-based optimization. 
\paragraph{Proof for investment-based optimization.}  Recall that for the instance that we have constructed, the baseline utility is $\alpha = 1$ and the engagement metric is $\MPlatform(w) = \wcostly + \wcheap$. By \Cref{prop:equilibriumhomogenousinvestment} and \Cref{thm:investmentbased}, there is a symmetric mixed Nash equilibrium $\muideal(2, c, u, \mathcal{T})$ given by the joint distribution $(\Wcostly, \Wcheap)$ where $\Wcheap$ is a point mass at $0$ and $\Wcostly$ is specified by:
\[\mathbb{P}[\Wcostly \le \wcostly] = \min(1, \wcostly). \]
This means that for $w \in \text{supp}(\muideal(2, c, u, \mathcal{T}))$, it holds that 
\[\MPlatform(w) = \MIdeal(w) = \wcostly.\]
This means that for $\textbf{w} \in \text{supp}(\muideal(2, c, u, \mathcal{T})^2)$, it holds that:
\[w_{i^*(\MIdeal; \textbf{w})} = \argmax_{w \in \textbf{w}} \MIdeal(w) = \argmax_{w \in \textbf{w}} \MPlatform(w).  \]
This means that: 
\[ \RealizedEngagement(\MIdeal; \textbf{w}) = \MPlatform\left(\argmax_{w \in \textbf{w}} \MPlatform(w)\right) = \max_{w \in \textbf{w}} \MPlatform(w)  \]
as desired. 

\paragraph{Proof for engagement-based optimization.} Suppose that $\textbf{w} \in \text{supp}(\muengagement(2, c, u, \mathcal{T})^2)$. Since $w_{i^*(\MPlatform; \textbf{w})} \in \textbf{w}$, it holds that: 
\[ \RealizedEngagement(\MPlatform; \textbf{w}) = \MPlatform(w_{i^*(\MPlatform; \textbf{w})}) \le \max_{w \in \textbf{w}} \MPlatform(w)  \]
as desired. 
    
\end{proof}

\subsubsection{Proof of Lemma \ref{lemma:engagementlimitcdf}}\label{appendix:prooflimitcdflemma}

To prove Lemma \ref{lemma:engagementlimitcdf}, we first compute the cdf of the distribution $\V^{E, \epsilon, \mathcal{T}_{N, \epsilon}}$.
\begin{lemma}
\label{lemma:engagementcdf}
Let $N'$ be the minimum number such that $\sum_{i=1}^{N'} \frac{1}{N-i+1} \ge 1$. Let $\alpha^i$ be equal to $\frac{1}{N-i+1}$ for $1 \le i \le N' - 1$ and be equal to $1 - \sum_{i'=1}^{N'-1} \frac{1}{N - i'+1}$ for $i = N'$. The distribution $\V^{E, \epsilon, \mathcal{T}_{N, \epsilon}}$ has cdf $\mathbb{P}[\V^{E, \epsilon, \mathcal{T}_{N, \epsilon}} \le v]$ equal to: 
\[
\tiny{
\begin{cases}
 0 & \text{ if } v \le (1+\epsilon) \\
 \sum_{i'=1}^{i-1} \frac{1}{N-i'+1} + \frac{N}{(N-i+1)} \left(\frac{v}{(1+\epsilon)(1+1/N)^{i-1}} - 1 \right)   & \text{ if } v \in \left[(1+\epsilon) \cdot (1+1/N)^{i-1}, (1+\epsilon) \cdot (1+1/N)^{i}\right] \text{    for some } 1 \le i \le N'-1  \\
\sum_{i'=1}^{N'-1} \frac{1}{N-i'+1} + \frac{N}{(N-N'+1)} \left(\frac{v}{(1+\epsilon)(1+1/N)^{N'-1}} - 1 \right)   & \text{ if } v \ge (1+\epsilon) \cdot (1+1/N)^{N'-1} \\
 & \text{ and } v \le (1+\epsilon) \cdot (1+1/N)^{N'-1} \cdot \left( 1 + \frac{N - N'+1}{N} \cdot \left(1 - \sum_{j=1}^{N'-1} \frac{1}{N - j+1} \right) \right) \\
 1 & \text{ if } v \ge (1+\epsilon) \cdot (1+1/N)^{N'-1} \cdot \left( 1 + \frac{N - N'+1}{N} \cdot \left(1 - \sum_{j=1}^{N'-1} \frac{1}{N - j+1} \right) \right).
\end{cases}
}
\]
\end{lemma}
\begin{proof}
We apply Proposition \ref{prop:NtypesMdistribution}. The statement follows from this specification along with the fact that the $\text{supp}(V \mid T = t_i)$ and $\text{supp}(V \mid T = t_j)$ are disjoint for $i \neq j$ by the assumption of well-separated types. 
\end{proof}

We next bound the value $N'$ in Lemma \ref{lemma:engagementcdf}. 
\begin{lemma}
\label{lemma:N'bound}
Let $N'$ be the minimum number such that $\sum_{i=1}^{N'} \frac{1}{N-i+1} \ge 1$. For sufficiently large $N$, it holds that:
\[\frac{N+1}{e} - 1 < N - N'+1 < \frac{N}{e^{1 - \frac{3}{N+1}}}.\]
\end{lemma}
\begin{proof}
We first rewrite:
\[\sum_{i=1}^{M} \frac{1}{N-i+1} = \sum_{i=N-M+1}^{N} \frac{1}{i}.  \]
Using an integral bound, we observe that
\[\int_{N-M+1}^{N+1} \frac{1}{x} dx \le \sum_{i=N-M+1}^{N} \frac{1}{i} \le \frac{1}{N-M+1} + \int_{N-M+1}^{N} \frac{1}{x} dx. \]
This implies that: 
\[\ln\left(\frac{N+1}{N-M+1} \right) \le \sum_{i=N-M+1}^{N} \frac{1}{i} \le \frac{1}{N-M+1} + \ln\left(\frac{N}{N-M+1} \right). \]

Since $N'$ is the minimum number such that $\sum_{i=N-N'+1}^{N} \frac{1}{i} \ge 1$, it must hold that: (1) $\sum_{i=N-N'+1}^{N} \frac{1}{i} \ge 1$, and (2) $\sum_{i=N-N'+2}^{N} \frac{1}{i} < 1$. 

Condition (2) implies that:
\[\ln\left(\frac{N+1}{N-N'+2} \right) \le \sum_{i=N-N'+2}^{N} \frac{1}{i} < 1,\]
which we can rewrite as:
\[N+1 < e \cdot (N-N'+2), \]
which we can rewrite as:
\[N - N'+1 > \frac{N+1}{e} - 1. \]

Condition (1) implies that: 
\[\frac{1}{N-N'+1} + \ln\left(\frac{N}{N-N'+1} \right) \ge \sum_{i=N-N'+1}^{N} \frac{1}{i} \ge 1, \]
which we can rewrite as:
\[ \ln\left(\frac{N}{N-N'+1} \right) \ge 1 - \frac{1}{N-N'+1}.\]
Using Condition (2), we see that $N - N'+1 > \frac{N+1}{e} - 1 \ge \frac{N+1}{3}$ for sufficiently large $N$, so:
\[ \ln\left(\frac{N}{N-N'+1} \right) \ge 1 - \frac{1}{N-N'+1} > 1 - \frac{3}{N+1}.\]
We can write this as:
\[ N - N'+1 < \frac{N}{e^{1 - \frac{3}{N+1}}}  \]

where the last inequality uses the upper bound on $N'$ derived from Condition (2). 

\end{proof}

We prove Lemma \ref{lemma:engagementlimitcdf}. 
\begin{proof}[Proof of Lemma \ref{lemma:engagementlimitcdf}]
Fix $\epsilon > 0$. We apply Lemma \ref{lemma:engagementcdf}. Let $F_{N,\epsilon}$ be the cdf of $\V^{E, \epsilon, \mathcal{T}_{N, \epsilon}}$. 

The first case is $v \le 1+\epsilon$. It follows easily that $F_{N,\epsilon}(x) = 0$ for all $v \le (1+\epsilon)$ for all $N \ge 2$, which means that $\lim_{N \rightarrow \infty} F_{N, \epsilon}(v) = 0$ for all $x \le (1+\epsilon)$. 

The next case is $v > e^{1 - \frac{1}{e}} (1+\epsilon)$. We see that for 
\[v \ge (1+\epsilon) \cdot (1+1/N)^{N'-1} \cdot \left( 1 + \frac{N - N'+1}{N} \cdot \left(1 - \sum_{j=1}^{N'-1} \frac{1}{N - j+1} \right) \right),\] it holds that $F_{N, \epsilon}(x) = 1$. Observe that $\left(1 - \sum_{j=1}^{N'-1} \frac{1}{N - j+1} \right) \le \frac{1}{N-N'+1}$ by the definition of $N'$, which means that:
\[(1+\epsilon) \cdot (1+1/N)^{N'-1} \cdot \left( 1 + \frac{N - N'+1}{N} \cdot \left(1 - \sum_{j=1}^{N'-1} \frac{1}{N - j+1} \right) \right) \le (1+\epsilon) \cdot (1+1/N)^{N'}. \]
For sufficiently large $N$, applying Lemma \ref{lemma:N'bound}, we see that:
\[(1+\epsilon) \cdot (1+1/N)^{N'} = (1+\epsilon) \cdot \left((1+1/N)^{N}\right)^{N'/N} \ge (1+\epsilon) \cdot \left((1+1/N)^{N}\right)^{\frac{N+1}{N} - \frac{1}{e^{1 - \frac{3}{N+1}}}}. \]
This expression approaches $e^{1 - \frac{1}{e}} \cdot (1+\epsilon)$, which means that for any $v > e^{1 - \frac{1}{e}} \cdot (1+\epsilon)$, for sufficiently large $N$, it holds that $F_{N, \epsilon}(v) = 1$ as desired. Thus, for any $v > e^{1 - \frac{1}{e}} \cdot (1+\epsilon)$, it holds that $\lim_{N \rightarrow \infty} F_{N, \epsilon}(v) = 1$.

The next case is $(1+ \epsilon) < v < e^{1 - \frac{1}{e}}(1+\epsilon)$. In this case, for sufficiently large $N$, it holds that:
\begin{align*}
 (1 + \epsilon)\left(1+\frac{1}{N}\right)^{N'-1} &= (1 + \epsilon)\left(\left(1+\frac{1}{N}\right)^{N}\right)^{\frac{N'-1}{N}} \\
 &>  (1 + \epsilon)\left((1+1/N)^{N}\right)^{1 - \frac{1}{e^{1 - \frac{3}{N+1}}}},
\end{align*}
which approaches $e^{1-\frac{1}{e}} \cdot (1+\epsilon)$ in the limit. This means that for sufficiently large $N$, it holds that $x < (1 + \epsilon)\left(1+\frac{1}{N}\right)^{N'-1}$. For $v \in \left[(1+\epsilon) \cdot (1+1/N)^{i-1}, (1+\epsilon) \cdot (1+1/N)^{i}\right]$,
applying Lemma \ref{lemma:engagementcdf}, we see that:
\[ \sum_{i'=1}^{i-1} \frac{1}{N-i'+1}  \le F_{N,\epsilon}(v) \le \sum_{i'=1}^{i} \frac{1}{N-i'+1}. \]
Using an integral bound, we can lower bound the left-hand side as:
\[\sum_{i'=1}^{i-1} \frac{1}{N-i'+1} \ge \int_{N-i+2}^{N+1} \frac{1}{x} dx = \ln\left(\frac{N+1}{N-i+2} \right) = \ln\left(\frac{1+\frac{1}{N}}{1-\frac{i}{N}+\frac{2}{N}} \right). \]
We can also upper bound, for sufficiently large $N$, the right-hand side as: 
\begin{align*}
  \sum_{i'=1}^{i} \frac{1}{N-i'+1} &\le \frac{1}{N-i+1} + \int_{N-i+1}^N \frac{1}{x} dx \\
  &=\frac{1}{N-i+1} + \ln\left(\frac{N}{N-i+1} \right) \\
  &\le \frac{1}{N-N'+1} + \ln\left(\frac{N}{N-i+1} \right) \\ 
  &\le_{(A)} \frac{1}{\frac{N+1}{e} - 1} + \ln\left(\frac{1}{1-\frac{i}{N}+\frac{1}{N}} \right) 
\end{align*}
where (A) follows from Lemma \ref{lemma:N'bound}. 
Putting this all together, we see that 
\[v \in \left[(1+\epsilon) \cdot (1+1/N)^{i-1}, (1+\epsilon) \cdot (1+1/N)^{i}\right],\]
then 
\[ \ln\left(\frac{1+\frac{1}{N}}{1-\frac{i}{N}+\frac{2}{N}} \right) \le F_{N,\epsilon}(v)  \le \frac{1}{\frac{N+1}{e} - 1} + \ln\left(\frac{1}{1-\frac{i}{N}+\frac{1}{N}} \right).\]
We next reparameterize $i$ as $\alpha = i/N$. We can rewrite $v \in \left[(1+\epsilon) \cdot (1+1/N)^{i-1}, (1+\epsilon) \cdot (1+1/N)^{i}\right]$ as $v \in \left[(1+\epsilon) \cdot ((1+1/N)^N)^{\frac{i}{N}-\frac{1}{N}}, (1+\epsilon) \cdot ((1+1/N)^N)^{\frac{i}{N}}\right]$, or alternatively as 
\[v \in \left[(1+\epsilon) \cdot ((1+1/N)^N)^{\alpha-\frac{1}{N}}, (1+\epsilon) \cdot ((1+1/N)^N)^{\alpha}\right],\]
and the bound as:
 \[ \ln\left(\frac{1+\frac{1}{N}}{1-\alpha +\frac{2}{N}} \right) \le F_{N,\epsilon}(v)  \le \frac{1}{\frac{N+1}{e} - 1} + \ln\left(\frac{1}{1-\alpha+\frac{1}{N}} \right).\]
For any $v = (1+\epsilon) \cdot e^{\beta}$ where $\beta < 1 - 1/e$, we see that for sufficiently large $N$, it holds that:
\[ \lim_{N \rightarrow \infty} F_{N, \epsilon}((1+\epsilon) \cdot e^{\beta}) = \ln\left(\frac{1}{1 - \beta} \right) \]
Reparameterizing $\beta = \ln\left(\frac{v}{1+\epsilon}\right)$, we obtain:
\[\lim_{N \rightarrow \infty} F_{N, \epsilon}(v) = \ln\left(\frac{1}{1 - \ln\left(\frac{v}{1+\epsilon}\right)} \right) \]
\end{proof}

\subsubsection{Statement and proof of Lemma \ref{lemma:inequality}}\label{appendix:prooflemmainequality}

\begin{lemma}
\label{lemma:inequality}
For $x \in [1, e^{1-1/e}]$, it holds that:
\begin{equation}
\label{eq:original}
 \ln\left(\frac{1}{1 - \ln(x)}\right) \ge x - 1   
\end{equation}
\[ \]
\end{lemma}
\begin{proof}
In the proof, we will use the following two standard bounds: (F1) $\ln(1+z) \ge \frac{2z}{2+z}$ for $z \ge -1$ and (F2) $e^{2z} \le \frac{1+z}{1-z}$ for $z \in (0,1)$. 

 First, let's reparamterize and set $x_1 = \ln(x)$, so that the range of $x_1$ is now $(0, 1-1/e)$. Then we can rewrite \eqref{eq:original} as:
 \[\ln\left(\frac{1}{1 - x_1}\right) \ge e^{x_1} - 1. \]
 We can now apply (F2) to $z = x_1/2 \in (0, 0.5 - \frac{1}{2e})$ to see that it suffices to show that:
 \[\ln\left(\frac{1}{1 - x_1}\right) \ge \frac{1+x_1}{1-x_1} - 1, \]
 which can be simplified to
 \[\ln\left(\frac{1}{1 - x_1}\right) \ge \frac{x_1}{1-x_1/2}, \]
 which can be simplified to
 \[
   \ln\left(\frac{1}{1 - x_1}\right) \ge \frac{2x_1}{2-x_1}  
 \]
 Let's reparameterize again to set $x_2 = \frac{x_1}{1-x_1}$ so that the range of $x_2$ is now $[0, e-1]$. Using that $x_1 = \frac{x_2}{1+x_2}$, it thus suffices to show: 
  \[
   \ln\left(1 + x_2\right) \ge \frac{ \frac{2 x_2}{1+x_2}}{2 - \frac{x_2}{1+x_2}},  
 \]
 which can be simplified to:
  \[
   \ln\left(1 + x_2\right) \ge \frac{2x_2}{2+x_2}.  
 \]
 This follows from (F1). 
  
\end{proof}

\subsection{Proof of Proposition \ref{prop:comparisonengagementhomogeneous}}\label{appendix:comparisonengagementhomogeneous}
We prove \Cref{prop:comparisonengagementhomogeneous}.
\begin{proof}[Proof of \Cref{prop:comparisonengagementhomogeneous}]

Observe that for the instance that we have constructed, the minimum investment level is $\beta_t = \max(0, -\alpha)$, the cost function is $c([\wcostly, 0]) = \wcostly + \gamma \cdot \wcheap$, and the engagement metric is $\MPlatform(w) = \wcostly + \wcheap$. The function $f_t(\wcheap)$ is equal to $\max(0, (\wcheap/t) - \alpha)$. 

We define the following quantities:
\begin{align*}
T^{\text{I}} &:= \mathbb{E}_{\mathbf{w} \sim \left(\muideal(P, c, u, \mathcal{T})\right)^P}[\RealizedEngagement(\MIdeal; \mathbf{w})] \\
T^{\text{E}} &:= \mathbb{E}_{\mathbf{w} \sim \left(\muengagement(P, c, u, \mathcal{T})\right)^P}[\RealizedEngagement(\MPlatform; \mathbf{w})].
\end{align*}
It suffices to show that $T^{\text{E}} \ge T^{\text{I}}$. 

We first analyze each term separately. 

For the term $T^{\text{I}}$, we apply \Cref{prop:equilibriumhomogenousinvestment} and \Cref{thm:investmentbased}. This means that $\muideal(P, c, u, \mathcal{T})$ is specified by joint distribution $(\Wcostly, \Wcheap)$ where $\Wcheap$ is a point mass at $0$ and $\Wcostly$ is distributed as:
\begin{align*}
   \mathbb{P}_{(\Wcostly, \Wcheap) \sim \muideal(P, c, u, \mathcal{T})} [\Wcostly \le \wcostly] &= \begin{cases}
(-\alpha)^{1/(P-1)} & \text{ if }  0 \le \wcostly \le -\alpha \\
  \left(\min(1, \wcostly)\right)^{1/(P-1)} & \text{ if } \wcostly \ge \max(0, -\alpha).
\end{cases}
\end{align*}
To analyze $T^{I}$, we observe that for $w \in \text{supp}(\muideal(P, c, u, \mathcal{T}))$, it holds that 
\[\MPlatform(w) = \MIdeal(w) = \wcostly.\]
This means that for $\textbf{w} \in \text{supp}(\muideal(P, c, u, \mathcal{T})^P)$, it holds that:
\[w_{i^*(\MIdeal; \textbf{w})} = \argmax_{w \in \textbf{w}} \MIdeal(w) = \argmax_{w \in \textbf{w}} \wcostly.  \]
This means that: 
\[ \RealizedEngagement(\MIdeal; \textbf{w}) = \MPlatform\left(\argmax_{w \in \textbf{w}} \wcostly\right) = \max_{w \in \textbf{w}} \wcostly, \]
which means that: 
\[T^{\text{I}} = \mathbb{E}_{\mathbf{w} \sim \left(\muideal(P, c, u, \mathcal{T})\right)^P}[\RealizedEngagement(\MIdeal; \mathbf{w})] = \mathbb{E}_{\mathbf{w} \sim \left(\muideal(P, c, u, \mathcal{T})\right)^P}\left[\max_{w \in \textbf{w}} \wcostly \right].
\]
We rewrite this expression as follows. Let $Z^{\text{I}}$ be a random variable given by the maximum of $P$ i.i.d. realizations of a random variable distributed as $\Wcostly$. This random variable has cdf: 
\begin{align*}
   \mathbb{P} [Z^{\text{I}} \le z] &= \begin{cases}
(-\alpha)^{P/(P-1)} & \text{ if }  0 \le z \le -\alpha \\
  \left(\min(1, z)\right)^{P/(P-1)} & \text{ if } z \ge \max(0, -\alpha).
\end{cases}
\end{align*}
This means that:
\[T^{\text{I}} = \mathbb{E}_{\mathbf{w} \sim \left(\muideal(P, c, u, \mathcal{T})\right)^P}\left[Z^{\text{I}} \right].\]

For the term $T^{\text{E}}$, we apply \Cref{prop:equilibriumhomogenousengagement} and \Cref{thm:onetype}. This means that $\muengagement(P, c, u, \mathcal{T})$ is specified by joint distribution $(\Wcostly, \Wcheap)$ where $ \mathbb{P}_{(\Wcostly, \Wcheap) \sim \muengagement(P, c, u, \mathcal{T})}$ is equal to:
\begin{align*}
  \begin{cases}
(-\alpha)^{1/(P-1)} & \text{ if }  0 \le \wcostly \le -\alpha \\
  \left(\min(1, \wcostly + \gamma \cdot t \cdot (\wcostly + \alpha))\right)^{1/(P-1)} & \text{ if } \wcostly \ge \max(0, -\alpha).
\end{cases}
\end{align*}
Moreover, this means that the distribution $\Wcheap \mid \Wcostly = \wcostly$ for $\wcostly \in \text{supp}(\Wcostly)$ takes the following form. If $\wcostly > 0$, the distribution $\Wcheap \mid \Wcostly = \wcostly$ is a point mass at $f_t^{-1}(\wcostly) =  t \cdot (\wcostly + \alpha)$. If $\wcostly = 0$, then $\Wcheap \mid \Wcostly = \wcostly$ 
is a point mass at $0$ if $\alpha \le 0$, $\Wcheap \mid \Wcostly = \wcostly$  is distributed according to the cdf $\min\left(1, \left(\frac{\wcheap}{t \cdot \alpha}\right)^{1/(P-1)}\right)$ if $\alpha > 0$ and $\gamma > 0$, and $\Wcheap \mid \Wcostly = \wcostly$  is distributed as a point mass at $t \cdot \alpha$ if $\alpha > 0$ and $\gamma = 0$. To analyze $T^{\text{E}}$, we observe that:
\[w_{i^*(\MPlatform; \textbf{w})} = \argmax_{w \in \textbf{w}} \MPlatform(w).  \]
This means that: 
\[ \RealizedEngagement(\MPlatform; \textbf{w}) = \MPlatform\left(\argmax_{w \in \textbf{w}} \MPlatform(w) \right) = \max_{w \in \textbf{w}} (\wcostly + \wcheap), \]
which means that: 
\[T^{\text{E}} = \mathbb{E}_{\mathbf{w} \sim \left(\muengagement(P, c, u, \mathcal{T})\right)^P}[\RealizedEngagement(\MIdeal; \mathbf{w})] = \mathbb{E}_{\mathbf{w} \sim \left(\muengagement(P, c, u, \mathcal{T})\right)^P}\left[\max_{w \in \textbf{w}} (\wcostly + \wcheap) \right].
\]
We rewrite this expression as follows. Let $Z^{\text{E}}$ be a random variable given by the maximum of $P$ i.i.d. realizations of a random variable distributed as $\Wcostly + \Wcheap$. To formalize the cdf of $Z^{\text{E}}$, we need to rewrite $\wcostly + \gamma \cdot t \cdot (\wcostly + \alpha)$ in terms of $z := \wcostly + \wcheap$. 
Using the fact that for $w \in \text{supp}(\muengagement(P, c, u, \mathcal{T}))$ such that $\wcostly > 0$, it holds that:
\[z :=\wcostly + \wcheap = \wcostly + t \cdot (\wcostly + \alpha) =  \wcostly \cdot (1 + t) + t \cdot \alpha. \]
and 
\begin{align*}
\wcostly + \gamma \cdot t \cdot (\wcostly + \alpha)
&= \wcostly (1 + \gamma \cdot t ) + \gamma \cdot t \cdot \alpha \\
&= \frac{1+ \gamma \cdot t}{1 + t} (1 + t) \wcostly + \gamma \cdot t \cdot \alpha \\
&= \frac{1+ \gamma \cdot t}{1 + t} \left((1 + t) \wcostly + t \cdot \alpha\right) - t \cdot \alpha \cdot  \frac{1+ \gamma \cdot t}{1 + t} + \gamma  \cdot t \cdot \alpha\\
&= \frac{1+ \gamma \cdot t}{1 + t} \cdot z + t \cdot \alpha \cdot  \left(\gamma - \frac{1+ \gamma \cdot t}{1 + t} \right) \\
&= \frac{1+ \gamma \cdot t}{1 + t} \cdot z - t \cdot \alpha \cdot \frac{1 - \gamma}{1 + t}.
\end{align*}
If $\alpha \le 0$, then this random variable has cdf: 
\begin{align*}
   \mathbb{P} [Z^{\text{E}} \le z] &= \begin{cases}
(-\alpha)^{P/(P-1)} & \text{ if }  0 \le z \le -\alpha \\
  \left(\min\left(1, \frac{1+ \gamma \cdot t}{1 + t} \cdot z - t \cdot \alpha \cdot \frac{1 - \gamma}{1 + t} \right)\right)^{P/(P-1)} & \text{ if } z \ge \max(0, -\alpha).
\end{cases}
\end{align*}
Otherwise, if $\alpha > 0$, then this random variable has cdf:
\begin{align*}
   \mathbb{P} [Z^{\text{E}} \le z] &= \begin{cases}
  (z \cdot \gamma)^{1/(P-1)} & \text{ if }  z \le t \cdot \alpha \\
  \left(\min\left(1, \frac{1+ \gamma \cdot t}{1 + t} \cdot z - t \cdot \alpha \cdot \frac{1 - \gamma}{1 + t} \right)\right)^{P/(P-1)} & \text{ if } z > t \cdot \alpha.
\end{cases}
\end{align*}
This means that:
\[T^{\text{E}} = \mathbb{E}_{\mathbf{w} \sim \left(\muideal(P, c, u, \mathcal{T})\right)^P}\left[Z^{\text{E}} \right].\]

% t *  (wcostly  + alpha) 
We now combine these expressions and compare $T^{\text{E}}$ and $T^{\text{I}}$. First, we see that $Z^{\text{E}}$ stochastically dominates $Z^{\text{I}}$, since:
\[\frac{1+ \gamma \cdot t}{1 + t} \cdot z - t \cdot \alpha \cdot \frac{1 - \gamma}{1 + t} \le  z.  \]
This implies that:
\[T^{\text{E}} = \mathbb{E}_{\mathbf{w} \sim \left(\muideal(P, c, u, \mathcal{T})\right)^P}\left[Z^{\text{E}} \right] \ge \mathbb{E}_{\mathbf{w} \sim \left(\muideal(P, c, u, \mathcal{T})\right)^P}\left[Z^{\text{I}} \right]  = T^{\text{I}}
\]
as desired. 
\end{proof}

\section{Proofs for \Cref{subsec:userimplications}}

\subsection{Proof of Theorem \ref{thm:comparisonuserutility}}\label{appendix:comparisonuserutility}

The proof of Theorem \ref{thm:comparisonuserutility} follows from the following characterizations of the realized user utility for engagement-based optimization (Lemma \ref{lemma:comparisonuserutilityengagement}) and random recommendations (Lemma \ref{lemma:comparisonuserutilityrandom}) , stated and proved below. 
\begin{lemma}
\label{lemma:comparisonuserutilityrandom}
Consider the setup of Theorem \ref{thm:comparisonuserutility}. Then it holds that: 
\[\mathbb{E}_{\mathbf{w} \sim \left(\murandom(P, c, u, \mathcal{T})\right)^P}[\UserUtility(\MTrivial; \mathbf{w})] = \begin{cases}
\alpha & \text{ if }  \alpha  > 0\\
 0 & \text{ if } \alpha \le 0. 
\end{cases} 
\]
\end{lemma}
\begin{proof}
If $\alpha > 0$, then we see that $\UtilityCostly(0, t) = \alpha$. This means that:
\[\min_{\wcostly} \left\{ \CCostlyBaseline(\wcostly) \mid \UtilityCostly(\wcostly, t) \ge 0 \right\} = 0\]
and moreover the $\min$ is achieved at $w = [0, 0]$. This means that $\nu = 0$ and $\murandom(P, c, u, \mathcal{T})$ is a point mass at $[0, 0]$. This means that:
\[\mathbb{E}_{\mathbf{w} \sim \left(\murandom(P, c, u, \mathcal{T})\right)^P}[\UserUtility(\MTrivial; \mathbf{w})]  = \UtilityCostly(0, t) = \alpha. \]

If $\alpha \le 0$, then 
\[\wcostly^* := \argmin_{\wcostly'} \left\{\CCostlyBaseline(\wcostly') \mid \UtilityCostly(\wcostly', t) \ge 0 \right\}\]
satisfies $\UtilityCostly(\wcostly^*, t) = 0$. This means that if $(\Wcostly, \Wcheap) \sim \murandom(P, c, u, \mathcal{T})$, it holds that $\text{supp}(\Wcostly) \subseteq \left\{\wcostly^*, 0 \right\}$. Moreover, for any content landscape $\mathbf{w} \in \text{supp}(\murandom(P, c, u, \mathcal{T}))^P$, we see that:
\[\UserUtility(\MTrivial; \textbf{w}) := \mathbb{E}[u(w_{i^*(\MTrivial; \textbf{w})}, t) \cdot \mathbbm{1}[u(w_{i^*(\MTrivial; \textbf{w})}, t) \ge 0]] = 0.\] This means that:
\[\mathbb{E}_{\mathbf{w} \sim \left(\murandom(P, c, u, \mathcal{T})\right)^P}[\UserUtility(\MTrivial; \mathbf{w})]  = \UtilityCostly(\wcostly^*) = 0. \]

\end{proof}
\begin{lemma}
\label{lemma:comparisonuserutilityengagement}
Consider the setup of Theorem \ref{thm:comparisonuserutility}. If $\alpha > 0$, then it holds that:
\[\mathbb{E}_{\mathbf{w} \sim \left(\muengagement(P, c, u, \mathcal{T})\right)^P}[\UserUtility(\MPlatform; \mathbf{w})] < \alpha.\]
If $\alpha \le 0$, then it holds that:
\[\mathbb{E}_{\mathbf{w} \sim \left(\muengagement(P, c, u, \mathcal{T})\right)^P}[\UserUtility(\MPlatform; \mathbf{w})] = 0.\]
\end{lemma}
\begin{proof}
First, suppose that $\alpha \le 0$. In this case, we see that $u(w, t) \le 0$ for all $w \in \text{supp}(\muengagement(P, c, u, \mathcal{T}))^P)$. This implies that for any content landscape $\mathbf{w} \in \text{supp}(\muengagement(P, c, u, \mathcal{T}))^P$, we see that:
\[\UserUtility(\MPlatform; \textbf{w}) := \mathbb{E}[u(w_{i^*(\MPlatform; \textbf{w})}, t) \cdot \mathbbm{1}[u(w_{i^*(\MPlatform; \textbf{w})}, t) \ge 0]] = 0.\]
This means that:
\[\mathbb{E}_{\mathbf{w} \sim \left(\muengagement(P, c, u, \mathcal{T})\right)^P}[\UserUtility(\MPlatform; \mathbf{w})]  = 0. \]

For $\alpha > 0$, we see that $w = [0,0]$ is the unique value such that $w \in \Caugt$ and $u(w, t) \ge \alpha$. Moreover, by Lemma \ref{lemma:one-to-one}, we know that $\left\{[f_t(\wcheap), \wcheap] \mid \wcheap \ge 0\right\} = \Caugt$. We observe that $\text{supp}(\muengagement(P, c, u, \mathcal{T}))^P$ is contained in $\left\{[f_t(\wcheap), \wcheap] \mid \wcheap \ge 0\right\} = \Caugt$. This means that $u(w, t) < \alpha$ for all $w \in \text{supp}(\muengagement(P, c, u, \mathcal{T}))^P$ such that $w \neq [0, 0]$. Since there is no point mass at $0$, this means that the probability $[0, 0]$ shows up in the content landscape is $0$, so 
\[\mathbb{P}[\UserUtility(\MPlatform; \textbf{w})  < \alpha] = \mathbb{P}\left[\mathbb{E}[u(w_{i^*(\MPlatform; \textbf{w})}, t) \cdot \mathbbm{1}[u(w_{i^*(\MPlatform; \textbf{w})}, t) \ge 0]] < \alpha \right] = 1. \]
This means that:
\[\mathbb{E}_{\mathbf{w} \sim \left(\muengagement(P, c, u, \mathcal{T})\right)^P}[\UserUtility(\MPlatform; \mathbf{w})]  < \alpha. \]
\end{proof}

Using Lemma \ref{lemma:comparisonuserutilityengagement} and Lemma \ref{lemma:comparisonuserutilityrandom}, we prove Theorem \ref{thm:comparisonuserutility}. 
\begin{proof}[Proof of Theorem \ref{thm:comparisonuserutility}]
We apply Lemma \ref{lemma:comparisonuserutilityengagement} and Lemma \ref{lemma:comparisonuserutilityrandom}. When $\alpha >0$, we see that:
\[\mathbb{E}_{\mathbf{w} \sim \left(\muengagement(P, c, u, \mathcal{T})\right)^P}[\UserUtility(\muengagement; \mathbf{w})] < \alpha = \mathbb{E}_{\mathbf{w} \sim \left(\murandom(P, c, u, \mathcal{T})\right)^P}[\UserUtility(\MTrivial; \mathbf{w})]. \]
When $\alpha \le 0$, we see that:
\[\mathbb{E}_{\mathbf{w} \sim \left(\muengagement(P, c, u, \mathcal{T})\right)^P}[\UserUtility(\muengagement; \mathbf{w})] = 0 = \mathbb{E}_{\mathbf{w} \sim \left(\murandom(P, c, u, \mathcal{T})\right)^P}[\UserUtility(\MTrivial; \mathbf{w})] \]
\end{proof}
\subsection{Proofs of Proposition \ref{prop:userutility2types} and Proposition \ref{prop:userutility2typeswellseparated}}\label{appendix:userutility2types}

Both Proposition \ref{prop:userutility2types} and Proposition \ref{prop:userutility2typeswellseparated} leverage instantiations of \Cref{example:KMR} with $P = 2$ and $\gamma = 0$, where the type space is 
\begin{equation}
\label{eq:typespace2types}
 \mathcal{T}_{2, \epsilon, c} = \left\{\epsilon, c(1+\epsilon) - 1\right\}   
\end{equation}
for some $\epsilon > 0$ and $c > 1$. We first analyze the user utility for engagement-based optimization and investment-based optimization for instantations of this form.
\begin{lemma}
\label{lemma:engagementtwotypes}
Consider \Cref{example:KMR} with $P = 2$, $\gamma = 0$, and type space $ \mathcal{T}_{2, \epsilon, c}$ defined by \eqref{eq:typespace2types}.  If $c \ge 1.5$, then 
\[\mathbb{E}_{\mathbf{w} \sim \left(\muengagement(2, c, u, \mathcal{T}_{2, \epsilon, c})\right)^2}[\UserUtility(\MPlatform; \mathbf{w})] \le \frac{3}{16} \cdot W \cdot (c-1) \cdot (1+\epsilon).\]
If $1 < c \le (5- \sqrt{5})/2$, then $\mathbb{E}_{\mathbf{w} \sim \left(\muengagement(2, c, u, \mathcal{T}_{2, \epsilon, c})\right)^2}$ is at least 
\[\frac{1}{2} \cdot \left( 1 - \left(\frac{(3 -c) (c-1)}{2 - c} \right)^2\right)  \cdot W \cdot \frac{c \cdot (3 - c) \cdot (c-1) \cdot (1+\epsilon)}{2 (2 - c)}. \]
\end{lemma}
\begin{proof}
We apply \Cref{thm:2types}. Let $t_1 = \epsilon$ and $t_2 = c (1+\epsilon) - 1$. We see that $a_{t_1} = \frac{1}{1 + \epsilon}$, $a_{t_2} = \frac{1}{c(1+\epsilon)}$, and $s = 0$ in \Cref{example:KMR}. This implies that $c = \frac{a_{t_1}}{a_{t_2}}$. Moreover, we see that $f_t(\wcheap) = \max(0, (\wcheap) / t) - 1)$. Recall that our goal is to analyze
\[ \mathbb{E}_{\mathbf{w} \sim \left(\muengagement(2, c, u, \mathcal{T}_{2, \epsilon, c})\right)^2}[\UserUtility(\MPlatform; \mathbf{w})] = \mathbb{E}_{\mathbf{w} \sim \left(\muengagement(2, c, u, \mathcal{T}_{2, \epsilon, c})\right)^2}[\mathbb{E}[u(w_{i^*(\MPlatform; \textbf{w})}, t) \cdot \mathbbm{1}[u(w_{i^*(\MPlatform; \textbf{w})}, t) > 0]]].\]
To analyze this expression, we consider the reparameterization of the equilibrium given by the random vector $(\V, T)$ defined in \Cref{appendix:characterization2types}. As described in \Cref{appendix:characterization2types}, the function $h(v, t)$ corresponds to unique content over the form 
\[h(v, t) = [f_t(\wcheap), \wcheap] = [\max(0, (\wcheap) / t) - 1), \wcheap]\] 
that satisfies 
\[\MPlatform(h(v, t)) = \MPlatform([\max(0, (\wcheap) / t) - 1), \wcheap]) = \max(0, (\wcheap) / t) - 1) + \wcheap + 1
= v.\]

We first show that if $u(h(v,t), t') > 0$ for $(v,t) \in \text{supp}((\V, T))$ and $i \in \left\{1,2\right\}$, it must hold that 
$t' = t_2$ and $t = t_1$. We first observe that $u(h(v,t), t_1) \le 0$. Moreover, it holds that $u(h(v,t), t_2) = 0$ if $t = t_2$. This proves the desired statement.

We next show that $(v,t_1) \in \text{supp}((\V, T))$, it holds that: 
\[u(h(v,t_1), t_2) = \frac{W}{t_1} \cdot \wcheap \cdot (c - 1) \cdot ( 1 + \epsilon).\] 
Let $w =h(v, t_1)$. Observe that the utility function satisfies:
\[ u(w, t_2) = W \cdot t_2 \cdot (\wcostly - \wcheap/t_2 + 1)\] by definition. Since $(v, t_1) \in \text{supp}((\V, T))$, the equilibrium structure tells us that $w = h(v, t_1)$ satisfies $\wcheap \ge t_1$ and $\wcostly = (\wcheap / t_1) - 1$. This implies that 
\begin{align*}
  u(h(v,t_1), t_2) &= u(w, t_2) \\
  &= W \cdot t_2 \cdot (\wcostly - \wcheap/t_2 + 1) \\
  &= W \cdot t_2 \cdot \wcheap \cdot \left(\frac{1}{t_1} - \frac{1}{t_2}\right) = \frac{W}{t_1} \cdot \wcheap \cdot \left(t_2 - t_1 \right) \\
  &=  \frac{W}{t_1} \cdot \wcheap \cdot (c - 1) \cdot ( 1 + \epsilon)
\end{align*}
as desired. 

For the remainder of the analysis, we split into two cases: $c \ge 1.5$ and $1 < c \le (5- \sqrt{5})/2$.

\paragraph{Case 1: $c \ge 1.5$.} This corresponds to the first case of \Cref{thm:2types}. To analyze the user welfare expression, it suffices to restrict to the cases where the winning content $w_{i^*(\MPlatform; \textbf{w})}$ satisfies 
$u(w_{i^*(\MPlatform; \textbf{w})}, t') > 0$. By the above argument, it suffices to restrict to the case where the user has type $t' = t_2$ which occurs with probability $1/2$. 

We next show that $(v,t_1) \in \text{supp}((\V, T))$, it holds that: 
\[u(h(v,t_1), t_2) = \frac{3}{2} \cdot W \cdot (c - 1) \cdot ( 1 + \epsilon).\] We see that 
\[v = \wcheap (1 + 1/t_1) = \wcheap \cdot \frac{t_1 + 1}{t_1} \le \frac{3}{2 a_{t_1}} = \frac{3 (t_1 +1)}{2}\]
which means that:
\[\wcheap \le \frac{3 t_1}{2}. \] 
This implies that:
\begin{align*}
 u(h(v,t_1), t_2) &= \frac{W}{t_1} \cdot \wcheap \cdot (c - 1) \cdot ( 1 + \epsilon) \\
 &\le \frac{3}{2} \cdot W \cdot (c - 1) \cdot ( 1 + \epsilon)
\end{align*}
as desired. 

Let the content landscape $\textbf{w} = [w_1, w_2]$ in the reparameterized space be such that $w_1 = h(v^1, t^1)$ and $w_2 = h(v^2, t^2)$, and let $w_{i^*(\MPlatform; \textbf{w})}$ be $h(v^*, t^*)$. Since the user has type $t' = t_2$, the structure of the equilibrium in \Cref{thm:2types} implies that $t^* = t_1$ only if $t^1 = t^2 = t_1$. This implies that if $u(w_{i^*(\MPlatform; \textbf{w})}, t') > 0$, then it must hold that $t^1 = t^2 = t_1$. Moreover, in this case, the user utility is at most:
\[u(w_{i^*(\MPlatform; \textbf{w})}, t') \le \frac{3}{2} \cdot W \cdot (c-1) \cdot (1+\epsilon). \]

Putting this all together, we see that:
\begin{align*}
& \mathbb{E}_{\mathbf{w} \sim \left(\muengagement(2, c, u, \mathcal{T}_{2, \epsilon, c})\right)^2}[\UserUtility(\MPlatform; \mathbf{w})] \\
&= 
\mathbb{E}_{\mathbf{w} \sim \left(\muengagement(2, c, u, \mathcal{T}_{2, \epsilon, c})\right)^2}[\mathbb{E}[u(w_{i^*(\MPlatform; \textbf{w})}, t) \cdot \mathbbm{1}[u(w_{i^*(\MPlatform; \textbf{w})}, t) > 0]]] \\
&\le \mathbb{P}_{t' \sim T}[t' = t_2] \cdot \mathbb{P}_{(v^1, t^1) \sim (V,T)} [t^1 = t_1] \cdot \mathbb{P}_{(v^2, t^2) \sim (V,T)} [t^2 = t_1] \cdot \frac{3}{2} \cdot W \cdot (c-1) \cdot (1+\epsilon) \\
&= \frac{3}{16} \cdot W \cdot (c-1) \cdot (1+\epsilon).
\end{align*}

\paragraph{Case 2: $1 < c \le (5- \sqrt{5})/2$.} This corresponds to the third case of \Cref{thm:2types}. To analyze the user welfare expression, it suffices to restrict to the cases where the winning content $w_{i^*(\MPlatform; \textbf{w})}$ satisfies 
$u(w_{i^*(\MPlatform; \textbf{w})}, t') > 0$. By the above argument, it suffices to restrict to the case where the user has type $t' = t_2$ which occurs with probability $1/2$. 

Let the content landscape $\textbf{w} = [w_1, w_2]$ in the reparameterized space be such that $w_1 = h(v^1, t^1)$ and $w_2 = h(v^2, t^2)$, and let $w_{i^*(\MPlatform; \textbf{w})}$ be $h(v^*, t^*)$. Since the user has type $t' = t_2$, the structure of the equilibrium in \Cref{thm:2types} implies that if $v^1 > \frac{3 - \frac{a_{t_1}}{a_{t_2}}}{2 a_{t_2} \left(2 - \frac{a_{t_1}}{a_{t_2}} \right)} $ or if $v^2 > \frac{3 - \frac{a_{t_1}}{a_{t_2}}}{2 a_{t_2} \left(2 - \frac{a_{t_1}}{a_{t_2}} \right)} $, then it holds that $t^* = t_1$ and $v^* > \frac{3 - \frac{a_{t_1}}{a_{t_2}}}{2 a_{t_2} \left(2 - \frac{a_{t_1}}{a_{t_2}} \right)} $.

Moreover, we claim that if $w = h(v, t)$ is such that $v > \frac{3 - \frac{a_{t_1}}{a_{t_2}}}{2 a_{t_2} \left(2 - \frac{a_{t_1}}{a_{t_2}} \right)}$ and $t = t_1$, then the user utility is at least:
\[u(w, t') \ge W \cdot \frac{c \cdot (3 - c) \cdot (c-1) \cdot (1+\epsilon)}{2 (2 - c)}. \]
Note that:
\[\frac{3 - \frac{a_{t_1}}{a_{t_2}}}{2 a_{t_2} \left(2 - \frac{a_{t_1}}{a_{t_2}} \right)}  = \frac{c \cdot (1+\epsilon) \cdot (3 - c)}{2 (2 - c)}. \]
Moreover, we see that 
\[v = \wcheap(1 + 1/t_1) = \wcheap \cdot \frac{1 + \epsilon}{\epsilon} \ge \frac{c \cdot (1+\epsilon) \cdot (3 - c)}{2 (2 - c)}.\]
This implies that:
\[ \wcheap \ge \frac{c \cdot \epsilon \cdot (3 - c)}{2 (2 - c)}. \]
This implies that:
\[u(w, t') \ge \frac{W}{t_1} \cdot \wcheap \cdot (c-1) \cdot (1+\epsilon) \ge \frac{W}{\epsilon} \cdot \frac{c \cdot \epsilon \cdot (3 - c)}{2 (2 - c)} \cdot (c-1) \cdot (1+\epsilon) = W \cdot \frac{c \cdot (3 - c) \cdot (c-1) \cdot (1+\epsilon)}{2 (2 - c)}. \]

Putting this all together, we see that:
\begin{align*}
& \mathbb{E}_{\mathbf{w} \sim \left(\muengagement(2, c, u, \mathcal{T}_{2, \epsilon, c})\right)^2}[\UserUtility(\MPlatform; \mathbf{w})] \\
&= 
\mathbb{E}_{\mathbf{w} \sim \left(\muengagement(2, c, u, \mathcal{T}_{2, \epsilon, c})\right)^2}[\mathbb{E}[u(w_{i^*(\MPlatform; \textbf{w})}, t) \cdot \mathbbm{1}[u(w_{i^*(\MPlatform; \textbf{w})}, t) > 0]]] \\
&\ge \mathbb{P}_{t' \sim T}[t' = t_2] \cdot \mathbb{P}_{(v^1, t^1) \sim (V,T), (v^2, t^2) \sim (V,T)} \left[v^1 \text{ or } v^2 > \frac{3 - \frac{a_{t_1}}{a_{t_2}}}{2 a_{t_2} \left(2 - \frac{a_{t_1}}{a_{t_2}} \right)} \right] \cdot W \cdot \frac{c \cdot (3 - c) \cdot (c-1) \cdot (1+\epsilon)}{2 (2 - c)}. \\
&= \frac{1}{2} \cdot \left(1 - \left( 1 - \mathbb{P}_{(v, t) \sim (V,T)} \left[v > \frac{3 - \frac{a_{t_1}}{a_{t_2}}}{2 a_{t_2} \left(2 - \frac{a_{t_1}}{a_{t_2}} \right)} \right]\right)^2 \right) \cdot W \cdot \frac{c \cdot (3 - c) \cdot (c-1) \cdot (1+\epsilon)}{2 (2 - c)}. \\
\end{align*}
To simplify this expression, we see that: 
\begin{align*}
 1 - \mathbb{P}_{(v, t) \sim (V,T)} \left[v > \frac{3 - \frac{a_{t_1}}{a_{t_2}}}{2 a_{t_2} \left(2 - \frac{a_{t_1}}{a_{t_2}} \right)} \right] &= 1 - a_{t_1} \cdot \left(\frac{1}{a_{t_1}} -  \frac{3 - \frac{a_{t_1}}{a_{t_2}}}{2 - \frac{a_{t_1}}{a_{t_2}}} \left(\frac{1}{a_{t_2}} - \frac{1}{a_{t_1}} \right)\right) \\
 &= 1 - a_{t_1} \cdot \frac{1}{a_{t_1}} \left( 1 - \frac{3 - \frac{a_{t_1}}{a_{t_2}}}{2 - \frac{a_{t_1}}{a_{t_2}}} \left(\frac{a_{t_1}}{a_{t_2}} - 1 \right) \right) \\
 &= 1 - \left(1 - \frac{(3 -c) (c-1)}{2 - c} \right)  \\ 
 &= \frac{(3 -c) (c-1)}{2 - c}.
\end{align*}
Plugging this into the above expression, we obtain:
\[\frac{1}{2} \cdot \left( 1 - \left(\frac{(3 -c) (c-1)}{2 - c} \right)^2\right)  \cdot W \cdot \frac{c \cdot (3 - c) \cdot (c-1) \cdot (1+\epsilon)}{2 (2 - c)}. \]
\end{proof}

\begin{lemma}
\label{lemma:randomtwotypes}
Consider \Cref{example:KMR} with $P = 2$, $\gamma = 0$, and type space $ \mathcal{T}_{2, \epsilon, c}$ defined by \eqref{eq:typespace2types}.  If $c \ge 1.5$, then 
\[\mathbb{E}_{\mathbf{w} \sim \left(\murandom(2, c, u, \mathcal{T}_{2, \epsilon, c})\right)^2}[\UserUtility(\MTrivial; ƒ\mathbf{w})] = \frac{W}{2} \cdot \left(\epsilon + c(1+\epsilon) - 1 \right).\]
\end{lemma}
\begin{proof}
 We apply Theorem \ref{thm:randomrecommendations}. In the construction, we see that the minimum investment level $\beta_t = 0$. This means that $\kappa =0$ so $\nu = 0$. Thus we see that the distribution $\murandom(P, c, u, \mathcal{T})$ is a point mass at $w = [0, 0]$. Observe that $u([0, 0], t) = W \cdot t$, which means that:
\[\mathbb{E}_{\mathbf{w} \sim \left(\murandom(2, c, u, \mathcal{T}_{2, \epsilon, c})\right)^2}[\UserUtility(\MTrivial; \mathbf{w})] = \frac{1}{2} \left(u([0, 0], \epsilon) + u([0, 0], c(1+\epsilon) - 1)\right) = \frac{W}{2} \cdot \left(\epsilon + c(1+\epsilon) - 1 \right). \]
\end{proof}

Using Lemma \ref{lemma:engagementtwotypes} and Lemma \ref{lemma:randomtwotypes}, we prove \Cref{prop:userutility2types}.
\begin{proof}[Proof of \Cref{prop:userutility2types}]
We construct an instances with two types where the user welfare of engagement-based optimization exceeds the user welfare of random recommendations. Interestingly, users are \textit{nearly homogeneous} in these instances. Consider \Cref{example:KMR} with $P = 2$, $\gamma = 0$, and type space $ \mathcal{T}_{2, \epsilon, c}$ defined by \eqref{eq:typespace2types}. 
Let $\muengagement(2, c, u,  \mathcal{T}_{2, \epsilon, c})$ be the symmetric mixed equilibrium specified in Definition \ref{def:2types}, and let $\muideal_{2, c, u} := \muideal(P, c, u, \mathcal{T}_{2, \epsilon, c})$ be the symmetric mixed equilibrium specified in Theorem \ref{thm:investmentbased}.\footnote{A simple consequence of the formulation in \Cref{sec:investment} is that $\muideal(P, c, u, \mathcal{T}_{2, \epsilon, c})$ is independent of $c$ and $\epsilon$ for the type spaces that we have constructed.} It suffices to show that there exists $c > 1$ such that 
  \[\limsup_{\epsilon \rightarrow 0} \frac{\mathbb{E}_{\mathbf{w} \sim \left(\muengagement(2, c, u, \mathcal{T}_{2, \epsilon, c})\right)^2}[\UserUtility(\MPlatform; \mathbf{w})]}{ \mathbb{E}_{\mathbf{w} \sim \left(\murandom(2, c, u, \mathcal{T}_{2, \epsilon, c})\right)^2}[\UserUtility(\MPlatform; \mathbf{w})]} > 1. \]

We apply Lemma \ref{lemma:engagementtwotypes} and Lemma \ref{lemma:randomtwotypes} to see that:
  \begin{align*}
   & \limsup_{\epsilon \rightarrow 0} \frac{\mathbb{E}_{\mathbf{w} \sim \left(\muengagement(2, c, u, \mathcal{T}_{2, \epsilon, c})\right)^2}[\UserUtility(\MPlatform; \mathbf{w})]}{ \mathbb{E}_{\mathbf{w} \sim \left(\murandom(2, c, u, \mathcal{T}_{2, \epsilon, c})\right)^2}[\UserUtility(\MPlatform; \mathbf{w})]}[\UserUtility(\MPlatform; \mathbf{w})] \\
   &\ge \limsup_{\epsilon \rightarrow 0} \frac{\frac{1}{2} \cdot \left( 1 - \left(\frac{(3 -c) (c-1)}{2 - c} \right)^2\right)  \cdot W \cdot \frac{c \cdot (3 - c) \cdot (c-1) \cdot (1+\epsilon)}{2 (2 - c)}}{\frac{W}{2} \cdot \left(\epsilon + c(1+\epsilon) - 1 \right)}  \\
   &= \limsup_{\epsilon \rightarrow 0} \frac{\left( 1 - \left( \frac{(3 -c) (c-1)}{2 - c}  \right)^2\right)  \cdot \frac{c \cdot (3 - c) \cdot (c-1) \cdot (1+\epsilon)}{2 (2 - c)}}{\left(\epsilon + c(1+\epsilon) - 1 \right)}  \\
   &=  \frac{\left( 1 - \left( \frac{(3 -c) (c-1)}{2 - c} \right)^2\right)  \cdot \frac{c \cdot (3 - c) \cdot (c-1)}{2 (2 - c)}}{c - 1}  \\
   &= \left( 1 - \left( \frac{(3 -c) (c-1)}{2 - c} \right)^2\right)  \cdot \frac{c \cdot (3 - c)}{2 (2 - c)}.
   \end{align*}
   It is easy to see that there exists $c > 1$ such that the above expression is strictly greater than $1$, as desired.  
\end{proof}

Using Lemma \ref{lemma:engagementtwotypes} and Lemma \ref{lemma:randomtwotypes}, we prove \Cref{prop:userutility2typeswellseparated}.

\begin{proof}[Proof of \Cref{prop:userutility2typeswellseparated}]
 We construct instances with two types where user welfare of random recommendations exceeds the user welfare of engagement-based optimization. Interestingly, users are \textit{well-separated} in these instances. Consider \Cref{example:KMR} with $P = 2$, $\gamma = 0$, and type space $ \mathcal{T}_{2, \epsilon, c}$ defined by \eqref{eq:typespace2types}. It suffices to show that if $c \ge 1.5$, then it holds that: 
  \[\mathbb{E}_{\mathbf{w} \sim \left(\muengagement(2, c, u, \mathcal{T}_{2, \epsilon, c})\right)^2}[\UserUtility(\MPlatform; \mathbf{w})] < \mathbb{E}_{\mathbf{w} \sim \left(\murandom(2, c, u, \mathcal{T}_{2, \epsilon, c})\right)^2}[\UserUtility(\MTrivial; \mathbf{w})]. \]
  We apply Lemma \ref{lemma:engagementtwotypes} and Lemma \ref{lemma:randomtwotypes} to see that:
  \begin{align*}
   \mathbb{E}_{\mathbf{w} \sim \left(\muengagement(2, c, u, \mathcal{T}_{2, \epsilon, c})\right)^2}[\UserUtility(\MPlatform; \mathbf{w})] &\le \frac{3}{16} \cdot W \cdot (c-1) \cdot (1+\epsilon) \\
   &< \frac{W}{2} \cdot (c-1) \cdot (1+\epsilon) \\
   &\le \frac{W}{2} \cdot \left(\epsilon + c(1+\epsilon) - 1 \right) \\
   &= \mathbb{E}_{\mathbf{w} \sim \left(\murandom(2, c, u, \mathcal{T}_{2, \epsilon, c})\right)^2}[\UserUtility(\MTrivial; \mathbf{w})],
  \end{align*}
  as desired.   
\end{proof}

\end{document}